\documentclass[conference]{IEEEtran}
\IEEEoverridecommandlockouts

\usepackage[top=2cm, bottom=2cm, left=2cm, right=2cm]{geometry}
\usepackage{multirow,url,diagbox,amsmath,amsthm,amssymb,
graphicx,color,cite,algorithm,algpseudocode,algorithmicx,amsfonts,comment,framed}
\usepackage[table,xcdraw]{xcolor}

\newtheorem{theorem}{Theorem}
\newtheorem{definition}{Definition}
\floatname{algorithm}{Algorithm}
\usepackage{url,caption}
\usepackage{amsmath,amssymb,amsfonts}
\usepackage{tikz}
\usepackage{filecontents}
\usepackage{graphicx}
\usepackage{textcomp}
\usepackage{pifont,framed}
\usepackage{algorithm,algorithmicx}
\usepackage{xcolor}
\usepackage{color}
\usepackage{verbatim}
\usepackage[subfigure]{tocloft}
\usepackage{subfigure}
\usepackage{multirow}
  \usepackage{mathrsfs}
  \usepackage{ upgreek }
  
  \makeatletter
  \def\ps@IEEEtitlepagestyle{%
  	\def\@oddhead{\hfil\textcolor{red}{\textit{Accepted for publication at IEEE International Conference on Data Engineering (ICDE 2025)}}\hfil}
  	\def\@evenhead{}%
  }
  \makeatother
\allowdisplaybreaks[4]
\def\BibTeX{{\rm B\kern-.05em{\sc i\kern-.025em b}\kern-.08em
    T\kern-.1667em\lower.7ex\hbox{E}\kern-.125emX}}
\algtext*{EndWhile}
\algtext*{EndIf}
\newtheorem{lemma}{\bf Lemma}[section]

\begin{document}
%\title{Differential Mean Aggregation against General Colluding Attackers}
\title{Dual Utilization of Perturbation for Stream Data Publication under Local Differential Privacy}
%\author{
%{Rong Du, Qingqing Ye, Yue Fu, Haibo Hu}\\
%The Hong Kong Polytechnic University
%\and
%{Chengfang Fang, Jie Shi}\\
%Huawei International
%\email{{roong.du, yuesandy.fu}@connect.polyu.hk, {qqing.ye, haibo.hu}@polyu.edu.hk,  {fang.chengfang, SHI.JIE1}@huawei.com}
%}
\author{
\IEEEauthorblockN{Rong Du, Qingqing Ye$^*$, Yaxin Xiao, Liantong Yu, Yue Fu, Haibo Hu}

\IEEEauthorblockA{Department of Electrical and Electronic Engineering, The Hong Kong Polytechnic University}
{\it roong.du@connect.polyu.hk, qqing.ye@polyu.edu.hk, yaxin.xiao@connect.polyu.hk,} \\ {\it liantong2001.yu@connect.polyu.hk, yuesandy.fu@connect.polyu.hk, haibo.hu@polyu.edu.hk}
\thanks{$^*$Corresponding author: Qingqing Ye}
}

\maketitle

\begin{abstract}
	
Stream data from real-time distributed systems such as IoT, tele-health, and crowdsourcing has become an important data source. However, the collection and analysis of user-generated stream data raise privacy concerns due to the potential exposure of sensitive information. To address these concerns, local differential privacy (LDP) has emerged as a promising standard. Nevertheless, applying LDP to stream data presents significant challenges, as stream data often involves a large or even infinite number of values. Allocating a given privacy budget across these data points would introduce overwhelming LDP noise to the original stream data.  
	
%Nevertheless, since privacy protection is required at each time point, according to the LDP composition theorem, the privacy budget must be allocated across the entire stream (or one of its subset, e.g., $w$ time-slots in $w$-event LDP), which significantly degrades data utility.
	
Beyond existing approaches that merely use perturbed values for estimating statistics, our design leverages them for both perturbation and estimation. This dual utilization arises from a key observation: each user knows their own ground truth and perturbed values, enabling a precise computation of the deviation error caused by perturbation. By incorporating this deviation into the perturbation process of subsequent values, the previous noise can be calibrated. 
%Based on these deviations and the fact that adjacent data points typically exhibit small variations, users can utilize their perturbation information to calibrate subsequent perturbation processes. 
Following this insight, we introduce the Iterative Perturbation Parameterization (IPP) method, which utilizes current perturbed results to calibrate the subsequent perturbation process. To enhance the robustness of calibration and reduce sensitivity, two algorithms, namely Accumulated Perturbation Parameterization (APP) and Clipped Accumulated Perturbation Parameterization (CAPP) are further developed. We prove that these three algorithms satisfy $w$-event differential privacy while significantly improving utility. %Additionally, we propose a time slot sampling scheme to further enhance the accuracy of mean estimation statistics of subsequences while ensuring the precision of the published data stream. 
Experimental results demonstrate that our techniques outperform state-of-the-art LDP stream publishing solutions in terms of utility, while retaining the same privacy guarantee. 
\end{abstract}

\section{Introduction}
In the era of big data, collecting and analyzing real-time data streams are essential in many distributed systems such as IoT, tele-health, and crowdsourcing, which continuously generate ample amount of data. Typical examples include real-time traffic updates for navigation systems such as Google Maps and Waze~\cite{gumasing2023antecedents}, and consumer sentiment analysis on peer review platforms like Yelp~\cite{zhang2023can}. While these stream data offer substantial benefits for big data analysis and training AI models, they simultaneously bring forth privacy issues due to the potential exposure of sensitive personal information~\cite{harris1994time, dowding2012impact}.

Local Differential Privacy (LDP)\cite{chen2016private, duchi2013local, kasiviswanathan2011can}, a cryptographic technique extensively deployed in various real-world contexts~\cite{erlingsson2014rappor,team2017learning,ding2017collecting,nguyen2016collecting}, provides formal privacy guarantees by bounding the information leakage of individual users through a privacy parameter $\epsilon$. However, its direct application to stream data faces fundamental limitations. In particular, user-level LDP~\cite{bao2021cgm} considers the worst-case scenario and partitions the privacy budget according to the composition theorem~\cite{li2016differential}, where privacy budget division across timestamps leads to exponential utility degradation as accumulated noise overwhelms the original data.  To address the challenge of limited privacy budget in stream data collection, two mainstream approaches have emerged in existing research, i.e., relaxing privacy definitions and reducing the amount of transmitted data.

For privacy relaxation, event-level LDP~\cite{bao2021cgm, wang2021continuous} assigns independent privacy budgets $\epsilon$ to individual data points, which improves utility at the cost of weakened privacy guarantees across multiple timestamps. $w$-event LDP\cite{kellaris2014differentially} provides a more rigorous privacy guarantee by constraining the budget allocation within sliding windows of length $w$. For reducing data transmission, methods proposed in \cite{erlingsson2014rappor, wang2020towards} allow users to report only at turning points, allocating more privacy budget to each reported point. However, reporting at turning points reveals temporal information, as the transmission time itself indicates when the original data changes. To decouple the relationship between turning points and time, Mao et al. \cite{mao2024privshape} proposed PrivShape, which first employs symbolic aggregate approximation to convert numerical sequences into short strings, and then uses tree structures to identify frequent string patterns. However, due to the loss of detailed information, this method does not perform well in estimating means and distributions for range queries. Therefore, enhancing the utility while preserving user privacy remains a key challenge in LDP-enabled stream data analysis. %[TODO: I suggest not mentioning correlation issue, as we don't target on it. We mainly focus on utility, so we can list some relevant work and highlight their weakness. For example, memorization technology suffers from privacy issue.]

Departing from privacy budget allocation strategies, we propose a novel approach that leverages perturbation results of steam data in a dual manner. Beyond their conventional role in statistical analysis, these results can serve as valuable components for calibrating perturbation parameters. This dual-purpose approach stems from a key observation that each user knows both their own ground truth and perturbed values, which enables the deviation error to be computed precisely and in turn incorporated into subsequent values to calibrate the previous noise. Leveraging this auxiliary information, 
%This dual-purpose approach stems from a key observation about stream data characteristics: adjacent values typically exhibit small variations, making it possible to use perturbed information from previous time slots to enhance utility [TODO: is this assumption necessary? I would suggest following what we have claimed in abstract. ]. 
users can employ their locally perturbed values to adjust subsequent inputs for the perturbation mechanism $\mathcal{A}$, effectively mitigating errors from previous perturbations. We rigorously prove that this approach satisfies $w$-event LDP, as the adjust inputs process naturally dilutes individual value information. 

To implement this idea, we first introduce the Iterative Perturbation Parameterization (IPP) algorithm as a strawman proposal, which incorporates the deviation from the previous perturbation to the current stream value as input for $\mathcal{A}$.  Since the current perturbation is a cumulative effect of all previous data points rather than just the nearest one, we present the Accumulated Perturbation Parameterization (APP) algorithm to determine the deviations up to the current stream value. Through this process, the sensitivity increases as the range of the aggregated result expands. To mitigate this, we propose the Clipped Accumulated Perturbation Parameterization (CAPP), an optimized version of the APP that can reduce sensitivity while retaining most of the perturbation information. Additionally, we extend our perturbation parameterization scheme to time-slot sampling, which is particularly advantageous for mean statistics of subsequences while ensuring the accuracy of the published stream data.

Our contributions are summarized as follows:
\begin{itemize}
	\item To the best of our knowledge, this is the first work that parameterizes input value based on perturbation for publishing stream data under LDP, which significantly enhances utility over existing state-of-the-art techniques.
	\item We define a new problem of subsequence data collection under LDP, for statistical analysis on subsequences influenced by LDP while retaining rigorous privacy guarantee, i.e., $w$-event LDP.
	\item We propose a time-slot sampling approach for subsequence mean statistics estimation over user sampling, which further improves the utility of subsequence mean estimation.
\end{itemize}

The remainder of this paper is organized as follows. In Section \ref{preliminaries}, we introduce some preliminaries of LDP. In Section \ref{problemdefinition}, we formally define the problem and present a baseline approach to demonstrate our core ideas. In Section \ref{ASMs}, we propose two optimized algorithms, followed by a sampling-based solution in Section \ref{Delay}. In Section \ref{exp}, we provide extensive experimental evaluations. Then we review related studies in Section \ref{relatedwork}, and conclude the paper in Section \ref{conclusion}.

\section{Preliminaries}
\label{preliminaries}
\subsection{Local Differential Privacy}
Local Differential Privacy (LDP)~\cite{duchi2013local, kasiviswanathan2011can} is a modern privacy-preserving standard that enables users to anonymize their data before publishing. It is defined as follows:

\begin{definition}
	A randomized algorithm $\mathcal{A}$ satisfies $\epsilon$-local differential privacy ($\epsilon$-LDP), if and only if for any two values $x$ and $x'$, and all possible outputs $y \subseteq \text{Range}(\mathcal{A})$, the following condition holds:
	\begin{equation*}
		P[\mathcal{A}(x)\in y] \leq e^{\epsilon} \cdot P[\mathcal{A}(x')\in y].
	\end{equation*}
\end{definition}

The intuition of LDP is the probabilistic nature that $\mathcal{A}$ maps any particular input to an output according to a distribution, controlled by the privacy budget $\epsilon$. Therefore, anyone is unable to tell any individual's true answer from the observed output $y$ with high confidence. Particularly, LDP lifts up the dependency on a trusted collector that is inherent in centralized differential privacy (DP)~\cite{dwork2008differential,dwork2006calibrating,mcsherry2007mechanism}.

The following composition theorems of LDP provide a way to analyze the cumulative privacy loss when multiple LDP mechanisms are applied, either sequentially or in parallel.

\begin{theorem}
	\textbf{(Sequential Composition)} For any $k$ mechanisms providing $\epsilon_i$-local differential privacy for each, the sequence of all these mechanisms provides $\epsilon_{\text{seq}}$-local differential privacy, where
	\begin{equation*}
		\setlength{\abovedisplayskip}{0pt}
		\setlength{\belowdisplayskip}{0pt}
		\epsilon_{\text{seq}} = \sum\nolimits_{i=1}^{k}\epsilon_i.
	\end{equation*}
\end{theorem}

\begin{theorem}
	\textbf{(Parallel Composition)} Each of the $k$ mechanisms provides $\epsilon_i$-local differential privacy and operates on a disjoint subset of the entire dataset, then the union of these mechanisms provides $\epsilon_{\text{par}}$-local differential privacy, where
	\begin{equation*}
		\epsilon_{\text{par}} = \max_i\{\epsilon_i\}.
	\end{equation*}
\end{theorem}

These theorems are useful in theoretical analysis, as they allow the calculation of overall privacy loss in complex LDP scenarios and help in designing privacy-preserving algorithms by managing the privacy budget. The sequential composition theorem is particularly important for iterative algorithms or when a user's data is used multiple times, whereas the parallel composition theorem is applied when the dataset can be partitioned and separate analyses are run on each partition.

\subsection{$w$-event Privacy}
\label{wevent}
$w$-event privacy~\cite{kellaris2014differentially} is an LDP model tailed for data streams, which quantifies and limits privacy exposure within sliding windows containing $w$ events. Specifically, this privacy model is defined on $w$-neighboring streams as follows. 

\begin{definition}
	(\textbf{$w$-neighboring Streams}) Two streams $S=\{S_1, ..., S_t\}$ and $S'=\{S'_1, ..., S'_t\}$ of length $t$ are $w$-neighboring streams if they have at most $w$ consecutive different elements. Formally, for any $i, j \in \{1,\dots,t\}$ and $i \leq j$, if $S_{i} \neq S'_{i}$ and $S_{j} \neq S'_{j}$, then $j- i + 1 \leq w$.
\end{definition}

Then we have the following definition of $w$-event privacy.

\begin{definition}
	(\textbf{$w$-event Privacy}) Let $M$ be a mechanism that takes as input a stream of arbitrary size $t$ and produces $O$ as its output. $M$ satisfies $w$-event $\epsilon$-DP (or simply, $w$-event privacy) if for any two $w$-neighboring streams $S_t$ and $S'_{t}$, the following inequality holds
	\begin{equation*}
		\setlength{\abovedisplayskip}{1pt}
		\setlength{\belowdisplayskip}{0pt}
		\text{Pr}[M(S) \in O] \leq e^{\epsilon} \cdot \text{Pr}[M(S') \in O].
	\end{equation*}
\end{definition}

A mechanism satisfying $w$-event privacy can provide $\epsilon$-LDP guarantee in any sliding window of size $w$. In other words, for any mechanism with $w$-event privacy, $\epsilon$ can be viewed as the total available privacy budget in any sliding window of size $w$. 

\subsection{Square Wave Mechanism}
\label{SW_algorithm}
Square Wave (SW)~\cite{li2020estimating} is a typical LDP mechanism for estimating numerical value distribution. Each user processes a numerical value in $[0,1]$, and generates a sanitized value in $[-b, 1+b]$, where $b=\frac{\epsilon e^{\epsilon}-e^{\epsilon}+1}{2e^{\epsilon}(e^{\epsilon}-\epsilon-1)}$. Given an input value $v$, the randomized output can be expressed as following: 
$$Pr[SW(v)=v']=\left\{ \begin{array}{rcl} p, & & if|v-v'|\leq b,\\ q, & & otherwise,\\ \end{array} \right. $$
where $p=\frac{e^\epsilon}{2be^\epsilon+1}$ and $q=\frac{1}{2be^\epsilon+1}$. Upon receiving the perturbed data, the data collector aggregates the original distribution by using the Maximum Likelihood Estimation (MLE)~\cite{rossi2018mathematical}, and reconstructs the distribution of original values.

\section{Problem Definition and Baseline}
\label{problemdefinition}

In this section, we elaborate on our problem setting, and then introduce a baseline algorithm.

\subsection{Data Collection Framework}
Our problem setting involves a data collector and a group of distributed users, as shown in Figure \ref{sysframe}. In Step \ding{172}, each user owns a continuous data stream. Since the data collector is untrusted, in Step \ding{173} the users utilize LDP mechanisms (e.g., SW~\cite{li2020estimating}) to perturb their data and report the sanitized version. Upon receiving the perturbed data streams from the users, the data collector aggregates the inputs, with the goal of reconstructing an estimated data stream closely approximating the original one in Step \ding{174}. Finally, the data collector releases the aggregated values, e.g., mean or trends.

\begin{figure}
	\includegraphics[width=0.45\textwidth]{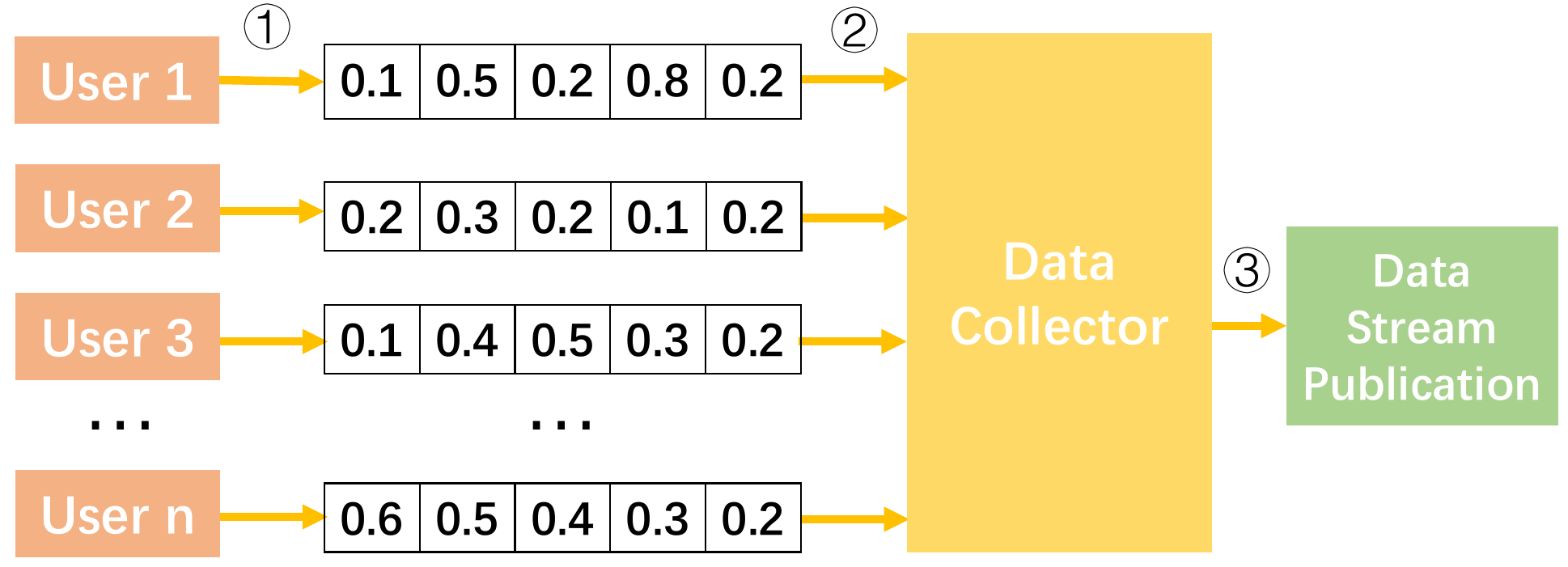}
	\caption{Illustration of the data streams collection framework}
	\label{sysframe}
	%	\vspace{-0.25in}
\end{figure}

\subsection{Problem Definition}
The objective of the data collector is to estimate subsequences of users' data while ensuring $w$-event privacy (shown in Section \ref{wevent}). Analyzing subsequences of data streams reveals localized meaningful patterns, which greatly aids data streams modeling and forecasting. To clarify, our discussion will primarily focus on a single data stream from one individual user, unless specified otherwise. Though this paper focuses on the single user scenario, the principles and algorithms discussed can be readily extended to accommodate multiple users.

The data collector aims to estimate a user's subsequence, denoted as $X_{(i,j)}$, which represents a continuous segment of the data stream spanning from the $i$-th time slot to the $j$-th time slot. Formally, it is defined as:
\begin{equation*}
	\setlength{\abovedisplayskip}{1pt}
	\setlength{\belowdisplayskip}{1pt}
	X_{(i,j)} = \{x_i, x_{i+1}, \ldots, x_j\}.
\end{equation*}

To protect their data, users will perturb the data using an LDP mechanism. In this paper, we employ the SW mechanism (shown in Section \ref{SW_algorithm}), which is considered as state-of-the-art method for collecting numerical data. We suppose the subsequence perturbed by SW is:
\begin{equation*}
	\setlength{\abovedisplayskip}{1pt}
	\setlength{\belowdisplayskip}{1pt}
	X'_{(i,j)} =\{x'_i, x'_{i+1}, \ldots, x'_j\}.
\end{equation*}

The data collector will then reconstruct a subsequence $\hat{X}_{(i,j)}$ by aggregating the perturbed data streams collected over that period as
\begin{equation*}
	\setlength{\abovedisplayskip}{1pt}
	\setlength{\belowdisplayskip}{1pt}
	\hat{X}_{(i,j)} =\{\hat{x}_i, \hat{x}_{i+1}, \ldots, \hat{x}_j\}.
\end{equation*}

Directly releasing $\hat{X}_{(i,j)}$ is referred to as \textbf{stream data publication}. Aside from directly releasing the stream data, the data collector may also perform statistical analysis on $\hat{X}_{(i,j)}$ and publish the corresponding statistical results, such as calculating and publishing the mean or identifying trends. Specifically, the mean estimation for $\hat{X}_{(i,j)}$ is denoted as:
\begin{equation*}
	\setlength{\abovedisplayskip}{1pt}
	\setlength{\belowdisplayskip}{1pt}
	\hat{M}_{(i,j)}=\frac{\sum_{t=i}^{j}\hat{x}_t}{j-i+1}.
\end{equation*} 

\subsection{Iterative Perturbation Parameterization (IPP)}
\label{onelookback}

In this section, we present a baseline algorithm named Iterative Perturbation Parameterization (IPP) for stream data publication, which iteratively adjusts input values based on perturbation results. The core idea of IPP is to integrate last perturbation deviation into the current value's perturbation process to mitigate errors. 
Specifically, let $x_t$ denote the original value at the $t$-th time slot, and $x'_t$ the corresponding perturbed value. Instead of directly perturbing the original value $x_t$, the user will calculate the deviation between $x_{t-1}$ and $x'_{t-1}$ and add it to $x_t$ as the input value in order to partially correct the error caused by the perturbation in $x_{t-1}$. We give the details of the IPP algorithm as follows.

\textbf{The procedure of IPP.} Let $x^I_t$ denote the input value at time slot $t$. As shown in Figure \ref{demo}, our method is explained as follows. For the original value $x_1=0.01$ of the first time slot, since no data has been uploaded previously, the deviation is zero, and we have input value $x^I_1=x_1$. After perturbing by SW, we obtain $x'_1=0$. Since the user knows the original value and the perturbed one, the user can calculate the deviation $d_1 = x_1 - {x}'_1 = 0.01$, which will be corrected in the next time slot when perturbing $x_2$. The specific approach is to add this deviation to the second original value, and we obtain the input value $x^I_2=d_1+x_2=0.16$. We then continue to perturb ${x}^I_2$ and get a new perturbed value $x'_2=0.19$. From this, we obtain a new deviation $d_2=x_2 - {x}'_2=-0.04$, which is then added to $x_3$. IPP repeats this process until all values are collected. 

\begin{figure}
	\includegraphics[width=0.45\textwidth]{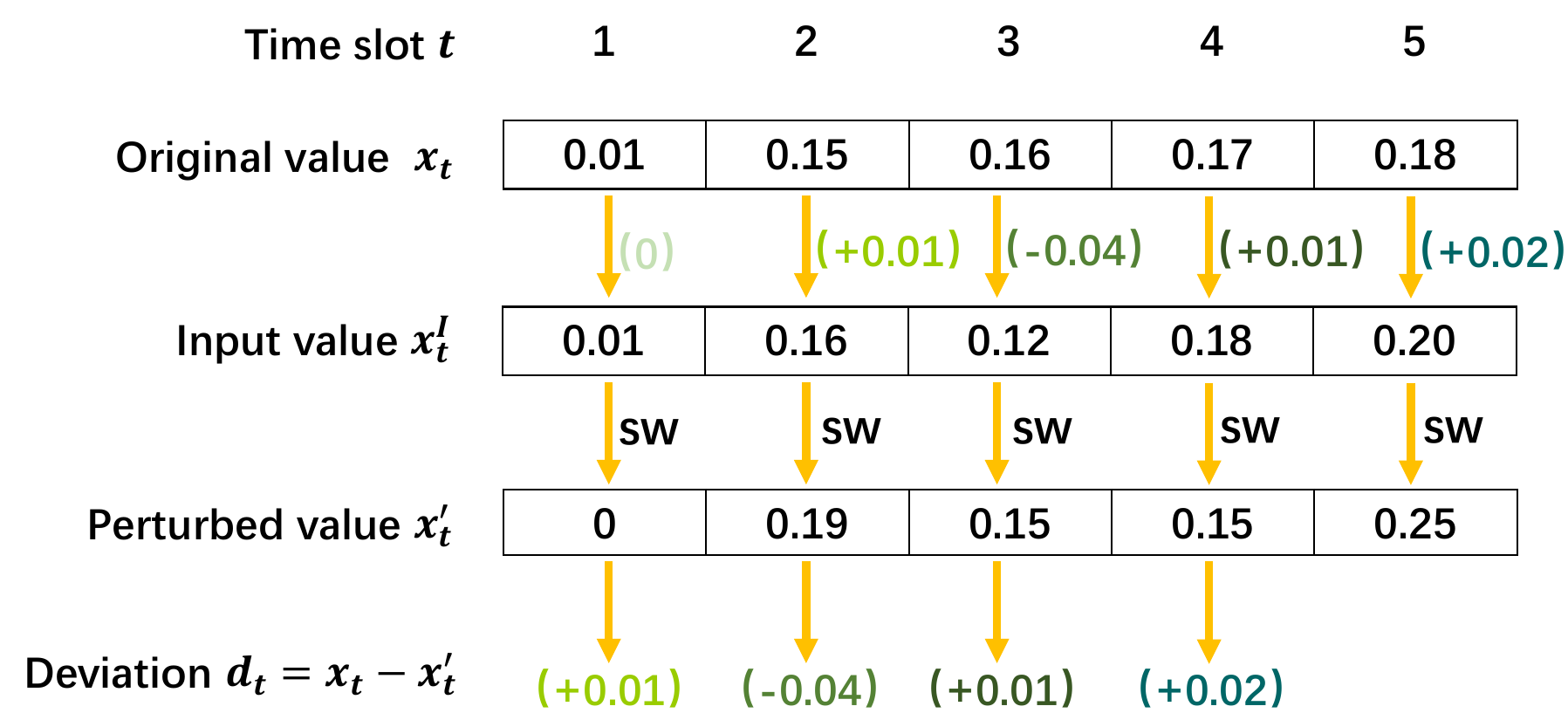}
	\caption{The procedure of IPP}
	\label{demo}
\end{figure}

It is important to note that when the original values are in the range of $[0,1]$, adding the deviation may result in the range of the input value exceeding this range. To address this issue, we simply clip $x^I_t$ to $[0,1]$. Specifically, if $x^I_t<0$, we set $x^I_t=0$, and if $x^I_t>1$, we set $x^I_t=1$.

This algorithm achieves better utility, as proven by Lemma \ref{LDP_1lookback}. It is worth noting that the $p$, $q$, and $b$ that appear in the subsequent proofs are parameters of SW mechanism, as illustrated in Section \ref{SW_algorithm}.

\begin{lemma}Given data stream values $X=\{x_1, x_2,\dots,x_n\}$, let $MD(M)=\sum \frac{M(x_i)}{n}-\sum \frac{x_i}{n}$ denote the mean deviation for algorithm $M$. We have $MD(IPP) < MD(SW)$.
	\label{LDP_1lookback}
\end{lemma}

\begin{proof}
	Let $d_1,\dots,d_n$ denote the deviations obtained from the SW mechanism, where the mean deviation $MD(SW)$ is given by:
	\begin{equation}
		\begin{split}
			\setlength{\abovedisplayskip}{1pt}
			\setlength{\belowdisplayskip}{1pt}
			MD(SW)  \!= \!  \frac{SW(x_1)+\dots+SW(x_n)-\!\sum x_i}{n}.
			%\\ &\!= \!  \frac{x_1-d_1+\dots+x_n-d_n-(x_1+\dots+x_n)}{n}.
		\end{split}
		\label{SW1}
	\end{equation}

	Let $d'_1,\dots,d'_n$ denote the deviations obtained from the IPP algorithm, where the mean deviation for IPP, $MD(IPP)$, can be calculated as:
	\begin{equation}
		\begin{split}
			MD(IPP) = \frac{SW(x_1^I)\!+\!\dots+SW(x_n^I)\!-\!\sum x_i}{n}.
			% \\ &\!= \! \frac{x_1-d'_1+\dots+x_n-d'_n-(x_1+\dots+x_n)}{n}.
		\end{split}
		\label{IPP1}
	\end{equation}
	Consider any two consecutive time slots $t$ and $t-1$ in Equation \ref{SW1}. By expanding $SW(x_{t-1})$, we have:
	\begin{equation*}
		T_{SW}(t-1,t)=\frac{x_{t-1}-d_{t-1}+SW(x_t)-(x_{t-1}+x_t)}{n}.
		%\label{sw2}
	\end{equation*}
	Consider any two consecutive time slots $t$ and $t-1$ in Equation \ref{IPP1}. By expanding $SW(x^I_{t-1})$, we have:
	\begin{equation*}
		\begin{split}
			& T_{IPP}(t-1,t)\!=\!\frac{x_{t-1}\!-\!d'_{t-1}\!+\!SW(x_t^I)\!-\!(x_{t-1}+x_t)}{n}\\ \!&=\! \frac{x_{t-1}\!-\!d'_{t-1}\!+\!SW(x_t+d'_{t-1})\!-\!(x_{t-1}+x_t)}{n}.
		\end{split}
		%\label{IPP2}
	\end{equation*}
	By design, the $SW$ mechanism generates outputs that closely approximate its input $x$. Consequently, the term $x_{t-1}-d'_{t-1}+SW(x_t+d'_{t-1})$ yields results that better approximate the true value $x_{t-1}+x_t$, whereas the term $x_{t-1}-d_{t-1}+SW(x_t)$ introduces an additional error component $d_{t-1}$, resulting in larger deviations from the true value. Thus, we have
	\begin{equation}
		T_{IPP}(t-1,t)<T_{SW}(t-1,t).
		\label{SW_IPP}
	\end{equation}
	Substituting Equation \ref{SW_IPP} into Equations \ref{SW1} and \ref{IPP1}, we obtain the following result:
	\begin{equation*}
		\begin{split}
			&MD(IPP) \!-\! MD(SW)\! =\\ &\! \frac{1}{2}(\sum_{t=2}^nT_{IPP}(t-1,t)\!-\! \sum_{t=2}^n T_{SW}(t-1,t))<0,
		\end{split}
	\end{equation*}
	which means IPP can always achieve lower mean deviation compared to directly perturbing stream data with the SW mechanism.
\end{proof}

\section{CAPP: Clipped Accumulated Perturbation Parameterization}
\label{ASMs}

Despite the benefits achieved by IPP in mitigating errors of the most recent perturbed value, it is limited to addressing a single perturbed value. A natural idea is extend our perturbation parameterization to encompass previous values, thereby enabling error correction from earlier perturbations. 
%Section \ref{onelookback} presents the baseline IPP algorithm, which only mitigates errors from the last perturbation. However, since each input value still requires perturbation, the randomness of perturbation may cause IPP to fail in mitigating the error introduced by the last perturbation. In the subsequent perturbation steps, tracking multiple historical deviations enables more flexible correction of errors from earlier perturbations. 
However, continuously accumulating deviations will lead to an unbounded increase in the input value range. Therefore, we need to constrain the range of input values reasonably. Based on these insights, we develop an advanced and effective algorithm called Clipped Accumulated Perturbation Parameterization (CAPP). To facilitate understanding, we first present Accumulated Perturbation Parameterization (APP) as an introductory solution, which considers all deviations caused by previously collected values and gives an effective post-processing method.

\subsection{Accumulated Perturbation Parameterization (APP)}
\label{Initialization}
To account for accumulated perturbation errors in stream data collection, we maintain an accumulated deviation $D$ of all perturbation-induced deviations from previously collected values. The key idea of APP is to adjust each new stream value by incorporating the deviation $D$ as part of the input. The procedure of APP is presented in Algorithm \ref{alg:APP}. 
\begin{algorithm}[H]
	\caption{ Accumulated Perturbation Parameterization}
	\begin{flushleft}
		{\bf Input:}
		Privacy budget $\epsilon$,	window size $w$, original stream $\{x_i,\ldots,x_j\}$ \\
		{\bf Output:}
		Collected data stream $\{x'_i,\ldots,x'_j\}$
	\end{flushleft}
	\begin{algorithmic}[1]
		\State $\epsilon_w = \epsilon / w$
		\State Initialize accumulated deviation $D=0$
		\For{each time slot $t$}
		\State $x^I_t = truncate(x_t + D, [0,1])$
		\State $x'_t=SW(x^I_t)$ with $\epsilon_w$
		\State Calculate deviation $d_t=x_t-x'_t$
		\State $D=D+d_t$
		\EndFor
		\State \Return $\{x'_i,\ldots,x'_j\}$  % Assuming you want to return the entire stream.
	\end{algorithmic}
	\label{alg:APP}
\end{algorithm}

The algorithm begins by calculating the privacy budget for each stream value, denoted by $\epsilon_w$, to ensure the algorithm satisfies $w$-event privacy (line 1). It then sets an initial accumulated deviation $D$ to zero since no value has been collected yet (line 2). We then compute $x^I_t$ by adding $D$ to the current value $x_t$, then clip $x^I_t$ to [0,1] (line 4). The next step is the perturbation process, yielding the perturbed output $x'_t$ (line 5). The algorithm then obtains the new deviation $d_t$ (line 6) and adds it to $D$ to update the accumulated deviation for future stream values (line 7). Finally, the data collector collects all the perturbed data (line 9). 

\textbf{Post-processing.}
After collecting the perturbed data, we proceed a smoothing step. The reason for including a smoothing step is that the perturbation process of the SW can introduce random positive or negative deviations. Smoothing allows positive and negative deviations to counteract, thereby further reducing errors.

For smoothing, we employ the simplest method, the Simple Moving Average (SMA)~\cite{booth2006hydrologic}. The smoothing value $\text{SMA}_{x'_t}$ at time slot $t$, with a smoothing window size of $2k+1$, is as follows:
\begin{equation*}
	\setlength{\abovedisplayskip}{0pt}
	\setlength{\belowdisplayskip}{0pt}
	\text{SMA}_{x'_t} = \frac{1}{2k+1} \sum_{r=t-k}^{t+k} x'_{r}.
	\label{SMA}
\end{equation*}
When dealing with boundary windows where the number of available values is less than $2k+1$, we simply average the available values. After smoothing, we can publish the stream data as follows:
\begin{equation*}
	\setlength{\abovedisplayskip}{1pt}
	\setlength{\belowdisplayskip}{1pt}
	\hat{X}_{(i,j)} = \{SMA_{x'_i},\dots,SMA_{x'_t},\dots,SMA_{x'_j}\}.
\end{equation*}

Smoothing has no impact on the mean of the results but can significantly improve the characterization of the data stream. The benefits of using SMA as post-processing of APP can be demonstrated by the following theoretical justifications:

\begin{lemma}
	Given a data stream $X = \{x_i, \ldots, x_{t}, \ldots, x_j\}$ and its APP-perturbed version $X' = \{x'_i, \ldots, x'_{t}, \ldots, x'_j\}$, define $Y = \{y_i, \ldots, y_{t}, \ldots, y_j\}$ as the smoothed stream obtained by applying simple moving average with window size $s$ ($s>1$) to $X'$. Then for any time point t, we have $Var(y_t) < Var(x'_t)$.
	\label{lemmaforsmoothing}
\end{lemma}
\begin{proof}
	Let $d_t = x_t - x'_t$ denote the deviation of $x_t$ from its perturbed counterpart $x'_t$. The smoothed data point $y_t$, computed with a window size of $s=2k+1$, is given by:
	\begin{equation*}
		\setlength{\abovedisplayskip}{1pt}
		\setlength{\belowdisplayskip}{1pt}
		\begin{split}
			y_t  &= \frac{x'_{t-k}+...+x'_t+...+x'_{t+k}}{2k+1}\\
			&= \frac{x_{t-k} - d_{t-1} + ...+ x_t - d_t + ...+ x_{t+k} - d_{t+k}}{2k+1}. \\
		\end{split}
	\end{equation*}
	The expected value of $y_t$ is formulated as:
	\begin{equation*}
		\setlength{\abovedisplayskip}{1pt}
		\setlength{\belowdisplayskip}{1pt}
		\begin{split}
			E[y_t]& = E\left[\frac{x_{t-k}...+ x_{t}+...+x_{t+k}}{2k+1}\right] \\ &- E\left[\frac{d_{t-k}...+ d_{t}+...+d_{t+k}}{2k+1}\right].
		\end{split}
	\end{equation*}
	Given that the deviations are bidirectional (both positive and negative) and exhibit compensatory behavior, where positive and negative deviations tend to counterbalance each other. Therefore, smoothing effectively mitigates the impact of perturbation noise for characterization. Next, we consider the variance of $y_i$. Since the noise is i.i.d., we have:
	\begin{equation*}
		\setlength{\abovedisplayskip}{1pt}
		\setlength{\belowdisplayskip}{1pt}
		Var(y_t) = Var\left(\frac{x'_{t-k}+... + x'_t + ...+x'_{t+k}}{2k+1}\right).
	\end{equation*}
	
	Due to the linearity of variance, we can express this as:
	\begin{equation*}
		\begin{split}
			\setlength{\abovedisplayskip}{1pt}
			\setlength{\belowdisplayskip}{1pt}
			&Var(y_t) = \frac{1}{({2k+1})^2}Var(x'_{t-k} +\dots+ x'_t +\dots+ x'_{t+k})\\
			&\!= \!\frac{1}{({2k+1})^2}(Var(x'_{t-k})\! +\!\dots\!+\! Var(x'_t) \!+\!\dots\! Var(x'_{t+k})).
		\end{split}
	\end{equation*}
	Because we do not know the specific results of the perturbed data, we consider the variance of all values to be the same as $Var(x'_t)$, from which we have:
	\begin{equation*}
		\setlength{\abovedisplayskip}{1pt}
		\setlength{\belowdisplayskip}{1pt}
		Var(y_t) = \frac{(2k+1)Var{(x'_t)}}{{(2k+1)}^2}= \frac{Var(x'_t)}{2k+1}.
	\end{equation*}
	
	Thus, the variance of the smoothed value is smaller than that of the originally perturbed value. 
\end{proof}

\textbf{Theoretical analysis.} The APP algorithm is designed to guarantee $w$-event privacy while achieving enhanced utility, as proven in Theorem \ref{LDPproof22}, Lemma \ref{lemmaforFLB} and Lemma \ref{lemma3}. 

\color{black}
\begin{theorem}
	Given any two $w$-neighboring streams $X$ and $Y$ and an arbitrary output $S$ from the APP algorithm, if each stream value is allocated a privacy budget of $\frac{\epsilon}{w}$, then the APP algorithm satisfies $w$-event privacy.
	\label{LDPproof22}
\end{theorem}

\begin{proof}
	Let's start with a straightforward scenario: a stream consisting of only two values, denoted as $X=\{x_1, x_2\}$ and $Y=\{y_1, y_2\}$, with a window size of $w=2$. Let $S=\{s_1, s_2\}$ represent a possible output. Our APP algorithm begins to calculate the input values $x^I_2$ and $y^I_2$ only after receiving $s_1$. Applying Bayes' theorem \cite{joyce2003bayes}, we arrive at the following equation:
	\begin{equation}
		\frac{P\{APP(x_1,x_2)={s_1,s_2}\}}{P\{APP(y_1,y_2)={s_1,s_2}\}}
		= \frac{P(s_1, s_2|x_1, x_2+x_1-s_1)}{P(s_1, s_2|y_1, y_2+y_1-s_1)}.
		\label{equ_ori}
	\end{equation}
	We will demonstrate the calculation of the joint distribution $P(s_1, s_2|x_1, $\ $x_2+x_1-s_1)$. Using the chain rule of probability, we can break down this joint distribution:
	\begin{equation}
		\begin{split}
			&P(s_1, s_2|x_1, x_2+x_1-s_1) \\&= P(s_1|x_1, x_2+x_1-s_1)P(s_2|s_1, x_1, x_2+x_1-s_1).
		\end{split}
		\label{hj1}
	\end{equation}
	
	Firstly, $P(s_1|x_1, x_2+x_1-s_1)$ simplifies to $P(s_1|x_1)$ since $s_1$ is generated solely based on $x_1$ through the SW mechanism. For $P(s_2|s_1, x_1, x_2+x_1-s_1)$, it simplifies to $P(s_2|x_2+x_1-s_1)$, as $s_1$ and $x_1$ become redundant given $x_2+x_1-s_1$. Therefore, Equation \ref{equ_ori} becomes:
	\begin{equation}
		\frac{P(s_1, s_2|x_1, x_2+x_1-s_1)}{P(s_1, s_2|y_1, y_2+y_1-s_1)}=\frac{P(s_1|x_1) P(s_2|x_2+x_1-s_1)}{P(s_1|y_1) P(s_2|y_2+y_1-s_1)}.
		\label{equen0}
	\end{equation}
	
	Since the APP algorithm primarily utilizes the SW mechanism, and each value is allocated a privacy budget of $\epsilon/2$ in this case, for any output $u$ and any output $o$, the probability bounds can be expressed as $q \leq P(\text{APP}(u)=o) \leq p$. From this, we deduce:
	\begin{equation}
		\setlength{\abovedisplayskip}{1pt}
		\setlength{\belowdisplayskip}{1pt}
		\frac{P(s_1|x_1)}{P(s_1|y_1)}
		\leq\frac{ p}{q} = e^{\frac{\epsilon}{2}}.
		\label{equen1}
	\end{equation}
	
	At the second time stamp, the user knows both $x_1$ and $x_2$. Moreover, $s_1$ is known to both the user and the data collector: the user knows it because the perturbation is executed locally, while the data collector receives $s_1$ during the first time step. Consequently, $x_2 + x_1 - s_1$ is a constant for user. Similar to Equation \ref{equen1}, we obtain:
	\begin{equation}
		\setlength{\abovedisplayskip}{1pt}
		\setlength{\belowdisplayskip}{1pt}
		\frac{P(s_2|x_2+x_1-s_1)}{P(s_2|y_2+y_1-s_1)}
		\leq\frac{ p}{q} = e^{\frac{\epsilon}{2}}.
		\label{equen2}
	\end{equation}
	While the sensitivity of $x_i$ in SW is 1, for $x_2 + x_1 - s_1$ it increases to 2. We apply clipping to restrict $x_2 + x_1 - s_1$ to [0,1], which reduces noise scale but without privacy leakage. Therefore, the clipping operation still maintains $\epsilon$-differential privacy. Substituting equations \eqref{equen1} and \eqref{equen2} into equation \eqref{equen0}, we obtain:
	\begin{equation*}
		\setlength{\abovedisplayskip}{1pt}
		\setlength{\belowdisplayskip}{1pt}
		\frac{P(s_1, s_2|x_1, x_2+x_1-s_1)}{P(s_1, s_2|y_1, y_2+y_1-s_1)}\leq e^{\frac{\epsilon}{2}}e^{\frac{\epsilon}{2}}=e^{{\epsilon}}.
		\label{equen3}
	\end{equation*}
	
	To demonstrate $w$-event privacy for the APP algorithm, we need to show that any perturbed subsequence of length $w$ satisfies $\epsilon$-differential privacy. Consider subsequences $\{x_1, \ldots, x_w\}$ from $X$ and $\{y_1, \ldots, y_w\}$ from $Y$. With a corresponding outcome subsequence $S=\{s_1, \ldots, s_w\}$ and its reverse $S^r$, we derive:
	\begin{equation}
		\setlength{\abovedisplayskip}{1pt}
		\setlength{\belowdisplayskip}{1pt}
		\begin{aligned}
			&\frac{P\{APP(x_1,x_2,\!...\!,x_w)\!=\!S\}}{P\{APP(y_1,y_2,\!...\!,y_w)\!=\!S\}} 
			\\	=& \frac{P\{APP(x_w,x_{w-1},\!...\!,x_1)\!=\!S^r\}}{P\{APP(y_w,y_{w-1},\!...\!,y_1)\!=\!S^r\}} \\
			=	&\frac{P(s_w|x^I_w)\dots P(s_2|x^I_2)P(s_1|x_1)}{P(s_w|y^I_w)\dots P(s_2|y^I_2)P(s_1|y_1)} ,
		\end{aligned}
		\label{pf1}
	\end{equation}
	where $x^I_t=x_t+D$. Since $x_t$ and $D$ are known constants for users, when we use $x_t^I$ as the input of SW, we have $\frac{P(s_t|x^I_t)}{P(s_t|y^I_t)}<\frac{p}{q}=e^{\epsilon/w}$. Therefore, we conclude that:
	\begin{equation}
		\setlength{\abovedisplayskip}{1pt}
		\setlength{\belowdisplayskip}{1pt}
		\begin{aligned}
			&\frac{P\{APP(x_1,x_2,\!...\!,x_w)\!=\!S\}}{P\{APP(y_1,y_2,\!...\!,y_w)\!=\!S\}} 
			\leq (\frac{p}{q})^w
			\leq (e^{\epsilon/w})^w
			\leq e^{\epsilon}.
		\end{aligned}
		\label{pf2}
	\end{equation}
	\label{proofof1LB}
\end{proof}
\color{black}
\begin{lemma}
	Given a subsequence ${x_i,\dots,x_t}$ with corresponding deviations ${b_i,\dots,b_t}$ and an accumulated deviation represented by $D$, let $ME(D)$ denote the deviation between the estimated mean and ground turth mean for the corresponding accumulated deviation $D$. We present the following inequality:
	\begin{equation*}
		\setlength{\abovedisplayskip}{1pt}
		\setlength{\belowdisplayskip}{1pt}
		ME({d_i,\dots,d_t}) <ME({d_{i+1},\dots,d_t}) < \dots < ME({d_t}).
	\end{equation*}
	\label{lemmaforFLB}
\end{lemma}

\begin{proof}
	We prove this result inductively by analyzing streams of increasing length. Starting with a basic case of two consecutive values ($x_{t-1}$ and $x_t$), we extend the analysis to three values, and finally generalize to an arbitrary stream length of $t-i+1$ values. We consider the case when the $x^I_t$ is not perturbed and the estimated mean is then:
	\begin{equation*}
		\setlength{\abovedisplayskip}{1pt}
		\setlength{\belowdisplayskip}{1pt}
		\frac{x'_{t-1} \!+\! x^I_t }{2}\!=\!\frac{x'_{t-1} + x_t + x_{t-1} - x'_{t-1}}{2} \!=\! \frac{x_{t-1} + x_t}{2},
	\end{equation*}
	which is the true mean of $\{x_{t-1}, x_t\}$. Then we introduce $x_{t-2}$, we consider the case when the $x^I_{t}$ is not perturbed similarly. The estimated mean is then:
	\begin{equation*}
		\setlength{\abovedisplayskip}{1pt}
		\setlength{\belowdisplayskip}{1pt}
		\begin{split}
			&\frac{x'_{t-2} \!+\!x'_{t-1} \!+\! x^I_t }{3}\!\\&=\!\frac{x'_{t-2} \!+\!x'_{t-1} \!+\! x_t \!+\! x_{t-1} \!-\! x'_{t-1}\!+\! x_{t-2} \!-\! x'_{t-2}}{3} \\ \!&=\! \frac{x_{t-2} +x_{t-1} + x_t}{3}.
		\end{split}
	\end{equation*}
	
	Similar to the case $\{x_{t-2},x_{t-1},x_{t}\}$, we can generalize this to $t$ values: when the $t$-th input value is not perturbed, the estimated mean is:
	\begin{equation*}
		\begin{split}
			\setlength{\abovedisplayskip}{0pt}
			\setlength{\belowdisplayskip}{0pt}
			\frac{x_i + x_2 + \dots+x_t}{t-i+1},
		\end{split}
	\end{equation*}
	which is the same as the ground truth. Assume the perturbation result for $x^I_t$ is $R$. The range of $R$ is $[-b, 1+b]$. As $\epsilon$ approaches 0, $b$ approaches $\frac{1}{2}$. Therefore, the maximum range of $R$ is $\left[-\frac{1}{2}, 1+\frac{1}{2}\right]$. The final estimated mean deviation depends only on the perturbation $R$ of the previous item. And we have:
	\begin{equation*}
		\setlength{\abovedisplayskip}{1pt}
		\setlength{\belowdisplayskip}{1pt}
		ME(d_r,\dots,d_t) = \frac{R}{t-r+1}.
	\end{equation*}
	
	When $t - r \leq 1$, the influence of $R$ on the numerator is less than that on the denominator, and we can conclude that the bias error increases as $r$ becomes larger.
	\label{initializationalgorithm}
\end{proof}

\color{black}
\begin{lemma}
	\label{lemma3}
	Let $\{x_1, x_2, \dots, x_n\}$ be the true time series, $\{x'_1, x'_2, \dots, x'_n\}$ be the directly perturbed time series without APP, and $\{y_1, y_2, \dots, y_n\}$ be the time series using APP and smoothing. Define the cosine similarities between the true time series and the two perturbed time series as $\cos(\theta{\mathrm{Direct}})$ and $\cos(\theta{\mathrm{APP}})$. Then, the following inequality holds:
	\begin{equation*}
		\setlength{\abovedisplayskip}{1pt}
		\setlength{\belowdisplayskip}{1pt}
		\mathbb{E}[\cos(\theta{\mathrm{APP}})] > \mathbb{E}[\cos(\theta_{\mathrm{Direct}})]
	\end{equation*}
\end{lemma}

\begin{proof}
	We demonstrate that APP combined with smoothing achieves higher cosine similarity with the original time series than direct perturbation. We define the cosine similarities as:
	\begin{equation*}
		\setlength{\abovedisplayskip}{1pt}
		\setlength{\belowdisplayskip}{1pt}
		\cos(\theta_{\mathrm{Direct}}) = \frac{\langle x, x' \rangle}{|x| |x'|} = \frac{\sum_{i=1}^n x_i x'i}{\sqrt{\sum{i=1}^n x_i^2} \sqrt{\sum_{i=1}^n (x'i)^2}}
	\end{equation*}
	\begin{equation*}
		\setlength{\abovedisplayskip}{1pt}
		\setlength{\belowdisplayskip}{1pt}
		\cos(\theta{\mathrm{APP}}) = \frac{\langle x, y \rangle}{|x| |y|} = \frac{\sum_{i=1}^n x_i y_i}{\sqrt{\sum_{i=1}^n x_i^2} \sqrt{\sum_{i=1}^n y_i^2}}
	\end{equation*}
	where $x'_i = x_i + \beta_i$ with $\beta_i$ being random noise, and $y_i$ is the APP-perturbed and smoothed time series. For direct perturbation, $\beta_i$ is independent noise with $\mathbb{E}[\beta_i] = 0$ and $\mathrm{Var}(\beta_i) = \sigma\beta^2$. This affects the cosine similarity in two ways. The numerator contains a noise term:
	\begin{equation*}
		\setlength{\abovedisplayskip}{1pt}
		\setlength{\belowdisplayskip}{1pt}
		\langle x, x' \rangle = \sum_{i=1}^n x_i(x_i + \beta_i) = |x|^2 + \sum_{i=1}^n x_i\beta_i
	\end{equation*}
	While $\mathbb{E}[\sum_{i=1}^n x_i\beta_i] = 0$, the variance of this term is:
	\begin{equation*}
		\setlength{\abovedisplayskip}{1pt}
		\setlength{\belowdisplayskip}{1pt}
		\mathrm{Var}\left(\sum_{i=1}^n x_i\beta_i\right) = \sum_{i=1}^n x_i^2\mathrm{Var}(\beta_i) = \sigma_\beta^2|x|^2
	\end{equation*}
	The denominator is inflated by noise, with expected squared norm:
	\begin{equation*}
		\begin{split}
			\setlength{\abovedisplayskip}{0pt}
			\setlength{\belowdisplayskip}{0pt}
			\mathbb{E}[|x'|^2]  &= \mathbb{E}\left[\sum_{i=1}^n (x_i + \beta_i)^2\right] \\&= \sum_{i=1}^n(x_i^2 + 2x_i\mathbb{E}[\beta_i] + \mathbb{E}[\beta_i^2]) = |x|^2 + n\sigma_\beta^2
		\end{split}
	\end{equation*}
	
	In contrast, APP leverages prior knowledge to generate adaptive noise $\eta_i$ with lower variance $\mathrm{Var}(\eta_i) = \sigma_\eta^2 < \sigma_\beta^2$ and adaptively generated to align better with the original signal structure. The APP-perturbed and smoothed time series is defined as:
	\begin{equation*}
		\setlength{\abovedisplayskip}{0pt}
		\setlength{\belowdisplayskip}{0pt}
		y_i = \frac{1}{s} \sum_{j=i-\lfloor s/2 \rfloor}^{i+\lfloor s/2 \rfloor} (x_j + \eta_j)
	\end{equation*}
	where $s$ is the smoothing window size. The smoothing operation further reduces noise variance. Assuming independent noise, the variance of smoothed noise becomes:
	\begin{equation*}
		\setlength{\abovedisplayskip}{0pt}
		\setlength{\belowdisplayskip}{0pt}
		\mathrm{Var}\left(\frac{1}{s}\sum_{j=i-\lfloor s/2 \rfloor}^{i+\lfloor s/2 \rfloor} \eta_j\right) = \frac{1}{s^2} \sum_{j=i-\lfloor s/2 \rfloor}^{i+\lfloor s/2 \rfloor} \mathrm{Var}(\eta_j) = \frac{\sigma_\eta^2}{s}
	\end{equation*}
	This dual advantage affects the cosine similarity: the numerator's noise component has reduced variance $\mathrm{Var}\left(\langle x, y \rangle - |x|^2\right) < \mathrm{Var}\left(\langle x, x' \rangle - |x|^2\right)$, and the denominator experiences less inflation:
	\begin{equation*}
		\setlength{\abovedisplayskip}{0pt}
		\setlength{\belowdisplayskip}{0pt}
		\mathbb{E}[|y|^2] \approx |x|^2 + \frac{n\sigma_\eta^2}{s}
	\end{equation*}
	
	Comparing the simplified expressions for expected cosine similarities:
	\begin{equation*}
		\setlength{\abovedisplayskip}{0pt}
		\setlength{\belowdisplayskip}{0pt}
		\mathbb{E}[\cos(\theta_{\mathrm{Direct}})] \approx \frac{1}{\sqrt{1 + \frac{n\sigma_\beta^2}{|x|^2}}}
	\end{equation*}
	versus
	\begin{equation*}
		\setlength{\abovedisplayskip}{0pt}
		\setlength{\belowdisplayskip}{0pt}
		\mathbb{E}[\cos(\theta_{\mathrm{APP}})] \approx \frac{1}{\sqrt{1 + \frac{n\sigma_\eta^2}{s|x|^2}}}
	\end{equation*}
	Since $\sigma_\eta^2 < \sigma_\beta^2$ and $s > 1$, we definitively have $\frac{n\sigma_\eta^2}{s|x|^2} < \frac{n\sigma_\beta^2}{|x|^2}$. Therefore $\mathbb{E}[\cos(\theta_{\mathrm{APP}})] > \mathbb{E}[\cos(\theta_{\mathrm{Direct}})]$. %The APP approach combined with smoothing provides higher data utility while maintaining equivalent privacy protection compared to direct perturbation through reduced noise magnitude through adaptive generation and further noise reduction via smoothing while preserving the low-frequency structure of the original signal.
\end{proof}
\color{black}

\subsection{Clipped Accumulated Perturbation Parameterization (CAPP)} 
In the IPP and APP algorithms, input values are directly clipped to [0,1] before applying SW perturbation, which is a simplistic approach that does not consider the privacy budget or the impact of accumulated deviation on utility. In CAPP, we first clip the accumulated values to a range $[l,u]$, then normalize them to $[0,1]$ for SW perturbation, and finally denormalize back to $[l,u]$. This approach provides more flexibility in handling accumulated deviations while maintaining privacy guarantees. The main content of this subsection introduces the CAPP algorithm and demonstrates how to determine optimal $[l,u]$ ranges to achieve better utility.

\begin{algorithm}
	\caption{Clipped Accumulated Perturbation Parameterization}
	\begin{flushleft}
		{\bf Input:}
		Privacy budget $\epsilon$, windows size $w$ \\
		{\bf Output:} 
		$\hat{X}_{(i,j)}=\{\hat{x_i},\dots,\hat{x}_t,\dots,\hat{x}_j\}$
	\end{flushleft}
	\begin{algorithmic}[1]
		\State Determine $l$ and $u$;
		\State $\epsilon_w =\epsilon/w$;
		\State Initialize accumulated deviation $D=0$
		\For{each time slot $t$}
		\State Update input value: $x^I_t=x_t+D$
		\State Clipping: $$x^I_t=\begin{cases}
			l, & \text{if } x^I_t < l,\\
			u, & \text{if } x^I_t > u.
		\end{cases}$$
		\State Normalization: $x^I_t=\frac{x^I_t-l}{u-l}$
		\State Perturbation: $x'_t=SW(x^I_t)$
		\State Denormalization: $x'_t=x'_t(u-l)+l$
		\State Calculate deviation: $d_t=x_t-x'_t$
		\State $D=D+d_t$
		\EndFor
		\State $\hat{X}_{(i,j)}=SMA(X'_{(i,j)})$;
		\State \textbf{return} $\hat{X}_{(i,j)}=\{\hat{x_i},\dots,\hat{x}_t,\dots,\hat{x}_j\}$
	\end{algorithmic}
	\label{algorithm_CFBL}
\end{algorithm}

\textbf{The procedure of CAPP.} The Algorithm \ref{algorithm_CFBL} initiates by determining the lower bound $l$ and upper bound $u$ for input values (line 1). The privacy budget per time slot $\epsilon_w$ is calculated (line 2), and an initial accumulated deviation $D$ is set to zero (line 3). For each $x_t$ within the stream, the algorithm first retrieves an input value based on $D$. It then ensures $x^I_t$ remains within the $[l, u]$ bounds by clipping: if $x^I_t$ falls below the lower bound $l$, it is raised to $l$, and if it exceeds the upper bound $u$, it is reduced to $u$ (line 6). Next, the algorithm normalizes the clipped input value to the range [0,1]. This normalization satisfies the input range requirement of the SW algorithm (lines 7-8). Following this, the algorithm employs the SW to perturb $x^I_t$, and the perturbed value $x'_t$ is then denormalized to the original range $[l, u]$ (line 9). It computes the deviation $d_t$ (line 10) and the accumulated deviation $D$ (line 11). Finally, the algorithm executes a smoothing procedure on the perturbed values and returns the estimated stream data (lines 13-14). 

According to Theorem \ref{haha}, CAPP satisfies $\epsilon$-LDP while applying a clipping operation to the input values. This clipping operation reduces the noise scale without compromising privacy guarantees.
\color{black}
\begin{theorem}
	\label{haha}
	Given any two $w$-neighboring streams $X$ and $Y$ and an arbitrary output $S$ from the CAPP algorithm with any clip range $[l,u]$, if each stream value is allocated a privacy budget of $\frac{\epsilon}{w}$, then the CAPP algorithm satisfies $w$-event privacy.
\end{theorem}

\begin{proof}
	For Equation \ref{equen2}, although $x_2 + x_1 - s_1$ is clipped to $[l, u]$, the normalization transformation ($x_{\text{norm}} = \frac{x - l}{u - l}$), which ensures the input is always restricted to $[0,1]$, preserves the privacy guarantees. Since both clipping and normalization are deterministic operations, they neither introduce additional randomness nor alter the privacy budget of the original mechanism. Therefore, the privacy guarantee remains:
	\begin{equation*}
		\frac{P(s_2|x_2 + x_1 - s_1)}{P(s_2|y_2 + y_1 - s_1)} \leq e^{\epsilon / 2}.
	\end{equation*}
	The remaining proof follows the same steps as in Theorem \ref{LDPproof22}.
	
\end{proof}
\color{black}

\textbf{The choice of $l$ and $u$.} The Algorithm \ref{algorithm_CFBL} involves choosing $l$ and $u$, which affect the utility of our method through different types of errors. As more stream values are collected, the range of $x^I_t$ expands. Clipping $x^I_t$ into a fixed range $[l,u]$ presents a trade-off: while a wider range preserves extreme values, it leads to higher sensitivity, requiring larger noise introduction and thus increasing sensitivity error $e_s$; conversely, a narrow range reduces data sensitivity and requires less noise for the same privacy guarantee, but excessive narrowing discards some information, introducing significant discarding error $e_d$.

The specific processes for calculating $e_s$ and $e_d$ are given as follows. We define an error function $T(e_s,e_d)$ to determine the clipping bound:
\begin{equation}
	\setlength{\abovedisplayskip}{1pt}
	\setlength{\belowdisplayskip}{1pt}
	\begin{split}
		T(e_s,e_d)&=e_s-e_d.\\
	\end{split}
	\label{Teped}
\end{equation}

The calculation of $[l,u]$ is based on $T(e_s,e_d)$:
$$\left\{
\begin{array}{rcl}
	l    = 0-T(e_s,e_d),\\
	u =1+T(e_s,e_d).  \\
\end{array} \right. 
$$
%Intuitively, when $\epsilon$ is relatively small, it results in $T(e_s, e_d) > 0$. The error introduced by the perturbation is more significant in this case. When $\epsilon$ is larger, it results in $T(e_s, e_d) < 0$. 

%The specific processes for calculating $e_s$ and $e_d$ are given as follows:

\textit{Sensitivity error $e_s$.} This error quantifies the deviation introduced by clipping. The calculation is defined as:
\begin{equation*}
	\setlength{\abovedisplayskip}{2pt}
	\setlength{\belowdisplayskip}{2pt}
	e_s=e^{x-E(SW(x))}-1,
\end{equation*}
where $x-E(SW(x))$ measures the deviation between the original and expected perturbed values. Using $e^{x-E(SW(x))}-1$ instead of simply $x-E(SW(x))$ to measure error has two advantages. First, it ensures $e_s$ approaches 0 for large $\epsilon$, where sensitivity reduction becomes unnecessary. Second, the exponential function mapping in $e^{x-E(SW(x))}-1$ can amplify even small differences between the original value and perturbed value. Given the unknown distribution of original data, we consider the worst-case scenario where $x = 1$.

\textit{Discarding error $e_d$.} To characterize the discarding error $e_d$, we first establish the probability density function $D_x=x-SW(x)$ to analyze the deviation distribution. Based on this, we define $e_d$ using the standard deviation of $D_x$. Formally:
\begin{equation*}
	\setlength{\abovedisplayskip}{2pt}
	\setlength{\belowdisplayskip}{2pt}
	\begin{split}
		e_d = \sqrt{Var(D_x)}.
	\end{split}
\end{equation*}
Smaller $\epsilon$ leads to larger $Var(D_x)$, indicating more significant deviations, while clipping within a narrow range would result in excessive information loss. Next, we give the calculation process of $Var(D_x)$. First, the probability density function of $G(D_x)$ is:
$$G(D_x)=\left\{
\begin{array}{rcl}
	q,      &      &   \text{if}\quad D_x\in[-1-b+x,-b],\\
	p  &      & \text{if}\quad D_x\in(-b,b),\\
	q  &      & \text{if} \quad D_x\in[b,b+x].\\
\end{array} \right. 
$$
The expectation and variance of this distribution can be calculated, with the results as follows. The expected value of $D_x$ is given by:
\begin{equation*}
	\setlength{\abovedisplayskip}{2pt}
	\setlength{\belowdisplayskip}{2pt}
	E(D_x)=q((1+2b)x-(b+\frac{1}{2})).
\end{equation*}
The expectation of $G(D_x^2) $ is:
\begin{equation*}
	\setlength{\abovedisplayskip}{2pt}
	\setlength{\belowdisplayskip}{2pt}
	E(D_x^2)\! = q\!\frac{\!-\!3b^2\!+\!6bx^2\!-\!6bx\!+\!3b\!+\!3x^2\!-\!3x\!+\!1}{3}\!+\! \frac{2pb^3}{3}.
\end{equation*}
For $e_d$ computation, we similarly consider the worst-case scenario with $x=1$ due to the unknown ground-truth. Then, the variance of $D_x$ is:
\begin{equation*}
	\setlength{\abovedisplayskip}{2pt}
	\setlength{\belowdisplayskip}{2pt}
	\begin{split}
		&Var(D_x) = E(D_x^2) - (E(D_x))^2 \\
		&= \frac{2b^3p}{3} -b^2q^2+b^2q-bq^2+bq-\frac{q^2}{4}+\frac{q}{3}.
	\end{split}
\end{equation*}
\color{black}
\subsection{Discussion on Generalizability}
\textbf{Crowd-level statistics.}
We extend our analysis from individual-level to crowd-level statistics by leveraging the collective information from independently and identically distributed (i.i.d.) users. Specifically, we first estimate individual statistics over a subsequence for each user, then analyze the distribution of these statistics across the population. For example, we estimate the mean values ${M'_1, \dots, M'_n}$ for each user's subsequence and characterize the distribution $D$ that these values follow. We aim to demonstrate that more accurate estimations of $\{M'_1, \dots, M'_n\}$ can lead to a more precise characterization of $D$. Formally, our theoretical framework establishes that accurate individual-level estimations inherently result in accurate crowd-level statistical inferences, as formalized in Theorem \ref{theorem_p}.

\begin{theorem}
	\label{theorem_p}
	Given any user's true feature value $A$ and estimated feature value $\hat{A}$, with the estimation error bounded by $|\hat{A} - A| \leq \beta$ (where $\beta$ is a fixed constant), when the sample size $N$ is sufficiently large, the maximum difference between the empirical distribution function $F_N(x)$ based on the estimated values $\hat{A}_i$ and the true distribution function $F(x)$ does not exceed $\eta$ with probability at least $1-\delta$, i.e.,
	\begin{equation*}
		\setlength{\abovedisplayskip}{2pt}
		\setlength{\belowdisplayskip}{2pt}
		P\left(\sup_x |F_N(x) - F(x)| \leq \eta \right) \geq 1 - \delta,
	\end{equation*}
	where $\eta > \beta$ is the given error bound and $\delta$ is the confidence parameter.
\end{theorem}

\begin{proof}
	We need to analyze the difference between the empirical distribution function $F_N(x) = \frac{1}{N} \sum_{i=1}^N I(\hat{A}_i \leq x)$ constructed from the estimated values $\hat{A}_i$ and the true distribution function $F(x) = P(A \leq x)$.
	
	To facilitate the analysis, we introduce the empirical distribution function based on the true values $A_i$, defined as $F_N^{\text{true}}(x) = \frac{1}{N} \sum_{i=1}^N I(A_i \leq x)$, and apply the triangle inequality to decompose the total error into two parts:
	\begin{equation*}
		\setlength{\abovedisplayskip}{0pt}
		\setlength{\belowdisplayskip}{0pt}
		\begin{split}
			\sup_x |F_N(x) - F(x)| \leq \sup_x |F_N(x) -  \\ F_N^{\text{true}}(x)| + \sup_x |F_N^{\text{true}}(x) - F(x)|.
		\end{split}
	\end{equation*}
	
	We first analyze the estimation error $\sup_x |F_N(x) - F_N^{\text{true}}(x)|$. Since $|\hat{A}_i - A_i| \leq \beta$, for any value of $x$, we have:
	\begin{equation*}
		\setlength{\abovedisplayskip}{0pt}
		\setlength{\belowdisplayskip}{0pt}
		A_i \leq x \Rightarrow \hat{A}_i \leq x + \beta
	\end{equation*}
	and
	\begin{equation*}
		\setlength{\abovedisplayskip}{0pt}
		\setlength{\belowdisplayskip}{0pt}
		\hat{A}_i \leq x \Rightarrow A_i \leq x + \beta
	\end{equation*}
	
	This implies that, for each point $x$:
	\begin{equation*}
		\setlength{\abovedisplayskip}{0pt}
		\setlength{\belowdisplayskip}{0pt}
		F_N^{\text{true}}(x - \beta) \leq F_N(x) \leq F_N^{\text{true}}(x + \beta)
	\end{equation*}
	
	Therefore,
	\begin{align*}
		\setlength{\abovedisplayskip}{2pt}
		\setlength{\belowdisplayskip}{2pt}
		F_N(x) - F_N^{\text{true}}(x) &\leq F_N^{\text{true}}(x + \beta) - F_N^{\text{true}}(x)\\
		F_N^{\text{true}}(x) - F_N(x) &\leq F_N^{\text{true}}(x) - F_N^{\text{true}}(x - \beta)
	\end{align*}
	
	Since the growth rate of an empirical distribution function does not exceed 1, we have:
	\begin{equation*}
		\setlength{\abovedisplayskip}{2pt}
		\setlength{\belowdisplayskip}{2pt}
		\sup_x |F_N(x) - F_N^{\text{true}}(x)| \leq \beta
	\end{equation*}
	
	Next, we examine the convergence error $\sup_x |F_N^{\text{true}}(x) - F(x)|$. According to the Dvoretzky–Kiefer–Wolfowitz (DKW) inequality, for any $\epsilon > 0$:
	\begin{equation*}
		\setlength{\abovedisplayskip}{2pt}
		\setlength{\belowdisplayskip}{2pt}
		P\left(\sup_x |F_N^{\text{true}}(x) - F(x)| > \epsilon \right) \leq 2e^{-2N\epsilon^2}
	\end{equation*}
	
	Taking $\epsilon = \eta - \beta$ and setting $2e^{-2N(\eta-\beta)^2} \leq \delta$, we obtain:
	\begin{equation*}
		\setlength{\abovedisplayskip}{2pt}
		\setlength{\belowdisplayskip}{2pt}
		N \geq \frac{\ln(2/\delta)}{2(\eta-\beta)^2}
	\end{equation*}
	
	When the sample size satisfies the above condition, we have:
	\begin{align*}
		\setlength{\abovedisplayskip}{2pt}
		\setlength{\belowdisplayskip}{2pt}
		P\left(\sup_x |F_N^{\text{true}}(x) - F(x)| \leq \eta - \beta \right) &\geq 1 - \delta
	\end{align*}
	Combining these results, we obtain:
	\begin{align*}
		\setlength{\abovedisplayskip}{2pt}
		\setlength{\belowdisplayskip}{2pt}
		P\left(\sup_x |F_N(x) - F(x)| \leq \beta + (\eta - \beta) \right) &\geq 1 - \delta\\
		P\left(\sup_x |F_N(x) - F(x)| \leq \eta \right) &\geq 1 - \delta
	\end{align*}
	
	This proves that when the sample size $N$ is sufficiently large and satisfies $N \geq \frac{\ln(2/\delta)}{2(\eta-\beta)^2}$, the maximum difference between the empirical distribution function $F_N(x)$ constructed from the estimated values and the true distribution function $F(x)$ does not exceed $\eta$ with probability at least $1-\delta$.
\end{proof}

\textbf{Extension to other mechanisms.} Our methods extend beyond the SW mechanism to other numerical LDP mechanisms, including Laplace~\cite{dwork2006calibrating}, SR~\cite{duchi2018minimax}, and PM~\cite{wang2019collecting}. Applying our approach to these mechanisms requires two key modifications. First, we normalize the original data to $[-1,1]$ to accommodate the input requirements of these mechanisms. Under this setting, the Laplace mechanism adds noise from $Lap(2/\epsilon)$. Accordingly, the clipping process in IPP and APP must also be adjusted to $[-1,1]$. Second, in CAPP, different mechanisms require specific clip intervals $[l,u]$. Due to space limitations, we omit the detailed discussion of these settings.

However, these mechanisms show inferior performance compared to SW. This is primarily because SW provides bounded perturbation results within $(-1/2, 3/2)$, regardless of the privacy budget. In contrast, PM mechanism with privacy budget $\epsilon=0.01$ leads to perturbation in $[-400,400]$. Laplace mechanism generates perturbations well beyond $[-1,1]$ even with small noise. SR mechanism, limited to outputs $\{-1,1\}$, loses substantial temporal information. These limitations explain our focus on the SW mechanism as the primary perturbation mechanism in this paper.

\textbf{Extension to high-dimensional time series data}.
Our scheme can be extended to high-dimensional time series data, such as location information and trajectory data publication. We consider each dimensional time series independently and apply our methods separately. To satisfy the composition theorem of differential privacy, we employ privacy budget splitting and sampling techniques. Specifically, for $d$-dimensional time series data collection, we propose two approaches:
\begin{itemize}
	\item Budget-Split (BS): At any given time slot, users upload data for all dimensions, with each data point allocated a privacy budget of $\epsilon/(dw)$. While this approach involves uploading more data points, it results in higher noise perturbation for each point.
	
	\item Sample-Split (SS): At any given time, users upload information for only one dimension, with each data point allocated a privacy budget of $\epsilon/w$. In this approach, each window can only upload $w/d$ data points, but the noise perturbation for each point is lower.
\end{itemize}
\color{black}

\section{ Perturbation Parameterization Algorithms with Sampling}
\label{Delay}
In this section, we combine sampling and perturbation parameterization (PP) algorithms to help enhance the accuracy of mean statistics while ensuring stream data publication remains satisfactory. Sampling approaches usually select a few values for reporting to increase the per-value privacy budget. However, if the number of values uploaded is too small, it may degrade the utility of stream data publication. To address this, we provide a perturbation parameterization sampling method that determines the optimal number of samples.

\begin{figure}
	\centering
	\includegraphics[width=0.45\textwidth]{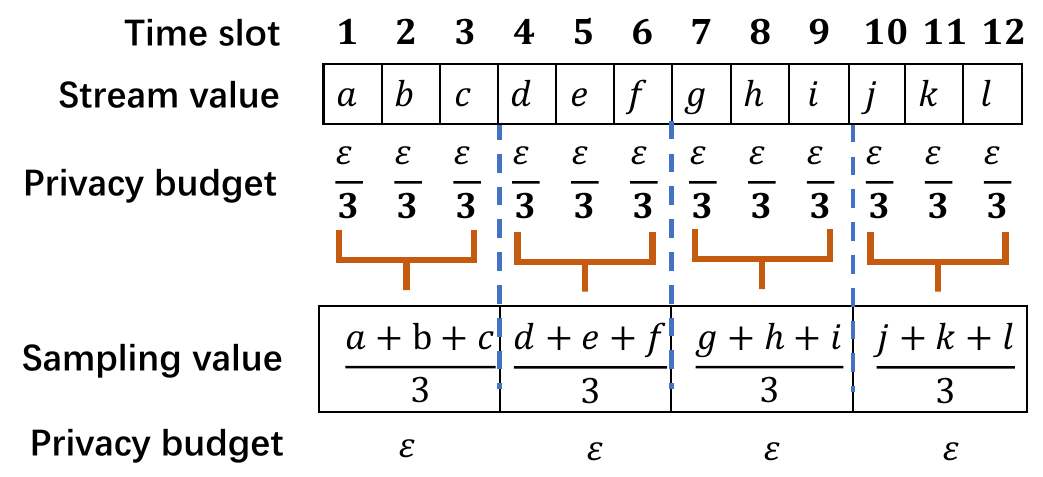}
	\caption{Illustration of the sampling}
	\label{sample}
	%	\vspace{-0.25in}
\end{figure}

\textbf{Sampling Design.} 
Let $n_s$ represent the number of samples for query $X_{(i,j)}$. To ensure $w$-event privacy in the sliding window model: We divide the query interval $[i,j]$ into $n_s$ segments, each containing $\lfloor \frac{j-i+1}{n_s} \rfloor$ time slots\footnote{When $(j-i+1)$ is not divisible by $n_s$, the remaining time slots are assigned to the last segment.}. For each segment, we sample one value at a predetermined position. The sampling positions remain consistent across different windows to maintain privacy guarantees.

A naive approach of simply uploading the value at the predetermined position to the data collector would result in information loss. To maximize the uploaded information while maintaining $w$-event local differential privacy, we upload the mean value of each segment. The user's original sampled values are then:
\begin{equation*}
	\setlength{\abovedisplayskip}{1pt}
	\setlength{\belowdisplayskip}{1pt}
	S=\{s_1,\dots,s_r,\dots,s_{n_s}\}.
\end{equation*}
where $s_r$ is the mean value for the $r$-th segment.	

Let us illustrate the sampling process with an example in Figure \ref{sample}, where $w=3$. Without sampling, each value must be uploaded individually with a privacy budget of $\epsilon/3$. In our sampling approach, we upload one aggregated value for every three consecutive values by computing their mean. This strategy allows us to increase the privacy budget for each uploaded value from $\epsilon/3$ to $\epsilon$. Users then generate and perturb these sampled mean values using our parameterized perturbation algorithms before transmission to the data collector.

After going through our perturbation parameterization algorithms, the perturbed values are:
\begin{equation*}
	\setlength{\abovedisplayskip}{3pt}
	\setlength{\belowdisplayskip}{3pt}
	S'=\{s'_1,\dots,s'_r,\dots,s'_{n_s}\}.
\end{equation*}
To restore the perturbed value for all time slots, we have:
\begin{equation*}
	X'=\{\overbrace{s'_1,\dots,s'_1}^{\lfloor \frac{j-i+1}{n_s} \rfloor \text{ times}},\dots,\overbrace{s'_{n_s},\dots,s'_{n_s}}^{\lfloor \frac{j-i+1}{n_s} \rfloor \text{ times}}\}.
\end{equation*}

\color{black}
\begin{algorithm}
	\caption{Perturbation Parameterization Sampling}
	\begin{flushleft}
		{\bf Input:}
		Privacy budget $\epsilon$, window size $w$, time interval $[i,j]$ \\
		{\bf Output:} 
		$\hat{X}_{(i,j)}=\{\hat{x}_i,\dots,\hat{x}_t,\dots,\hat{x}_j\}$
	\end{flushleft}
	\begin{algorithmic}[1]
		\State Divide the $[i,j]$ into $n_s$ segments
		\State $\gamma$=min$\{\lfloor \frac{j-i+1}{n_s} \rfloor,w\}$, $\epsilon_w =\epsilon/\gamma$, $X'=\emptyset$
		\For{each segment $r$}
		\State Set $x^I_r$ as the mean of current segment
		\State $x'_r=PP(x^I_r)$  \Comment{Same as PP, APP, CAPP}
		\State $X'= X' \sqcup \{x'_r\}^{\lfloor\frac{j-i+1}{n_s}\rfloor}$
		\EndFor
		\State \textbf{return} ${X'}_{(i,j)}=\{{x'}_i,\dots{x'}_t,\dots,{x'}_j\}$
	\end{algorithmic}
	\label{algorithm_sampling}
\end{algorithm}

\textbf{The procedure of PP-S.} By integrating sampling techniques into Perturbation Parameterization methods (IPP, APP, CAPP), we propose the Perturbation Parameterization Sampling (PP-S) algorithm. The Algorithm \ref{algorithm_sampling} begins by dividing the time interval $[i,j]$ into $n_s$ segments (line 1). It then determines the privacy budget per segment $\epsilon_w$ and initializes an empty set $X'$ (line 2). For each segment $r$, the algorithm first calculates the mean value $x^I_r$ of the current segment (line 4), and applies the Perturbation Parameterization (PP) method to perturb $x^I_r$ (line 5). The perturbed value is then replicated to match the segment size and added to set $X'$ (line 6). Finally, the algorithm returns the estimated stream data ${X'}_{(i,j)}$ (line 8).

\color{black}

%According to Theorem \ref{LDPproof22}, CAPP satisfies $\epsilon$-LDP while applying a clipping operation to the input values. This clipping operation reduces the noise scale without compromising privacy guarantees.

\textbf{The choice of $n_s$.}
As the number of samples $n_s$ decreases, each uploaded value gets a larger privacy budget, allowing more accurate mean estimation. However, a smaller $n_s$ compromises our ability to capture the stream characteristics effectively. Our goal is to determine an optimal $n_s$ that balances accurate mean estimation with effective stream data publication. 

Specifically, we will then utilize the distribution of the variance \cite{pranklin1974introduction} of these $n_s$ values to help determine the final size of $n_s$. Given $w$ with a total privacy budget of $\epsilon$, let $Var(n_s, \epsilon)$ denote the variance for the sample variance after sampling $n_s$ times. The reason for using the \textbf{variance of the sample variance}, rather than the \textbf{sample variance} itself, is that it better reflects the variability and instability of the sample variance statistic under repeated random sampling. By multiplying the additional factor $n_s$, we amplify the impact of sampling size, and thus our objective function is:
\begin{equation}
	\setlength{\abovedisplayskip}{1pt}
	\setlength{\belowdisplayskip}{1pt}
	\operatorname{argmin}_{n_s} n_s Var(n_s,\epsilon).
	\label{var}
\end{equation}

Since $n_s\in\{1,\dots,j-i+1\}$ is an integer, we can only enumerate all possible values of $n_s$ and find the one that minimizes the objective function.

Next, we detail the calculation of $Var(n_s,\epsilon)$. According to \cite{pranklin1974introduction}, $SW(x)$ is independent and identically distributed, then we have:
\begin{equation}
	\setlength{\abovedisplayskip}{1pt}
	\setlength{\belowdisplayskip}{1pt}
	Var(n_s,\epsilon)= \frac{1}{n_s}\left(\mu_4 -  \frac{\sigma^{2}(n_s-3)}{n_s-1}\right),
	\label{var_ga}
\end{equation}
where $\sigma^{2}$ is the corresponding variance, $\mu_4$ is the corresponding fourth central moment, which are all parameters related to $\epsilon$. Due to unknown stream data, we examine $SW(x)$ at $x=1$ (maximum variance case) and we obtain the expectation $\mu$ as follows:
\begin{equation*}
	\setlength{\abovedisplayskip}{1pt}
	\setlength{\belowdisplayskip}{1pt}
	\begin{split}
		\mu & = \int_{-b}^{1+b} SW(x) \, dx \\
		&= \frac{q(2b - 4bx + 1)}{2} + 2bpx = 2b(p-q)x+qb+\frac{q}{2}\\
		& =2bp - bq +\frac{q}{2}.
	\end{split}
\end{equation*}
Next, we derive the variance $\sigma^2$, which is computed by the following equation:
\begin{equation*}
	\setlength{\abovedisplayskip}{1pt}
	\setlength{\belowdisplayskip}{1pt}
	\begin{split}		
		\sigma^{2} &= \int_{-b}^{1+b} (x - \mu)^2 SW(x) \, dx  \\
		&= \frac{2b^3p}{3} - \frac{b^3q}{3} - 4b^2p^2 + 4b^2pq - b^2q^2 + b^2q \\&- 2bpq + 2bp + bq^2 - \frac{2bq}{3} - \frac{q^2}{4} + \frac{q}{3}.
	\end{split}  
\end{equation*}
Finally, we determine the fourth central moment \( \mu_4 \), given by the equation:
\begin{equation*}
	\setlength{\abovedisplayskip}{2pt}
	\setlength{\belowdisplayskip}{2pt}
	\begin{split}
		\mu_4 & = \int_{-b}^{1+b} (x - \mu)^4 SW(x) \, dx \\
		&=\frac{q}{5} + 2bp - bq - q\mu + 4b^3p + \frac{2b^5p}{5} + 2b^2q \\
		&- 2b^3q + b^4q + 2q\mu^2 - 2q\mu^3 + q\mu^4 + 12bp\mu^2 - 8bp\mu^3 \\
		&- 8b^3p\mu + 2bp\mu^4 - 6bq\mu^2 - 6b^2q\mu + 4bq\mu^3 + 4b^3q\mu \\&+ 4b^3p\mu^2 + 6b^2q\mu^2 - 8bp\mu + 4bq\mu.
		%   &= q/5 + 2*S*p - S*q - q*u + 4*S^3*p + (2*S^5*p)/5 \\& + 2*S^2*q - 2*S^3*q + S^4*q + 2*q*u^2 \\ &- 2*q*u^3 + q*u^4 + 4*S^3*p*u^2 + 6*S^2*q*u^2 - 8*S*p*u \\&+ 4*S*q*u + 12*S*p*u^2 - 8*S*p*u^3 - 8*S^3*p*u + \\& 2*S*p*u^4 - 6*S*q*u^2 - 6*S^2*q*u + 4*S*q*u^3 + 4*S^3*q*u
	\end{split}
\end{equation*}
\color{black}

\textit{Guidelines for selecting the value of $n_s$}. We provide heuristic guidelines for choosing $n_s$ based on the distributional properties of the data. When $n_s$ increases monotonically in Equation \ref{var}, $Var(n_s,\epsilon)$ must eventually exhibit a decreasing pattern to achieve its minimum. According to the definition of sample variance in Equation \ref{var_ga}, for relatively small values of $n_s$, the objective function is primarily determined by the fourth moment $\mu_4$, as the convergence rate of the sample fourth moment substantially exceeds that of $1/n$.

For heavy-tailed distributions (e.g., Cauchy distribution), $Var(n_s,\epsilon)$ tends to grow without bound as $n_s$ increases. In such scenarios, selecting a relatively small $n_s$ is recommended to prevent the potential explosion of $Var(n_s,\epsilon)$.

For light-tailed distributions (e.g., normal and uniform distributions), extreme values have limited influence on the fourth moment. Consequently, the fourth moment typically exists and $Var(n_s,\epsilon)$ demonstrates rapid convergence with increasing $n_s$. Although the behavior of $Var(n_s,\epsilon)$ may be complex for small $n_s$, an optimal $n_s$ that minimizes Equation \ref{var_ga} typically occurs near the point where $Var(n_s,\epsilon)$ begins to decrease and stabilize. In these cases, selecting a moderate value of $n_s$ represents a robust choice.
\color{black}

Theorem \ref{LDPproof2} proves that the Perturbation Parameterization algorithm with sampling satisfies $w$-event LDP.
\begin{theorem}
	Let $X$ and $Y$ be any two $w$-neighboring subsequences, and let $S'$ be an arbitrary output from the Perturbation Parameterization algorithm with sampling (PP-S), where PP-S represents IPP, APP, or CAPP. For a total of $n_w$ sampled values in each window, with each sampled value allocated a privacy budget of $\frac{\epsilon}{n_w}$, the PP-S mechanism achieves $w$-event local differential privacy.
	\label{LDPproof2}
\end{theorem}

\begin{proof}
	Define $F = (f_1, \dots, f_{n_w})$ as the sampling result from the stream $X = (x_1, x_2, \dots, x_w)$, and $G = (g_1, \dots, g_{n_w})$ as the sampling result from the stream $Y = (y_1, y_2, \dots, y_w)$. Let $S' = (s'_1, s'_2, \dots, s'_{n_w})$ denote the perturbed output sequence and $S^r$ its reverse sequence. Applying these definitions, Equation \ref{pf1} can be rewritten as:
	\begin{equation*}
		\begin{split}
			&\frac{P\{PP\!-\!S(f_1,f_2,...,f_{n_w})=S'\}}{P\{PP\!-\!S(g_1,g_2,...,g_{n_w})=S'\}}\\
			=&\frac{P\{PP\!-\!S(f_{n_w},f_{n_w-1},...,f_1)=S^r\}}{P\{PP\!-\!S(g_{n_w},g_{n_w-1},...,g_1)=S^r\}}\\
			=&\frac{P(s'_{n_w}|f^I_{n_w})\dots P(s'_2|f^I_2)P(s'_1|f_1)}{P(s'_{n_w}|g^I_{n_w})\dots P(s'_2|g^I_2)P(s'_1|g_1)}\\
			\leq&(\frac{p}{q})^{n_w}\leq (e^{\epsilon/{n_w}})^{n_w}
			\leq e^{\epsilon}.
		\end{split}
	\end{equation*}
	\label{proofof1LBl}
\end{proof}

\color{black}
\begin{figure*}[th]
	\hspace{0.25in}
	%	\vspace{-0.06in}
	{
		\begin{minipage}{15cm}
			\centering
			\includegraphics[scale=0.8]{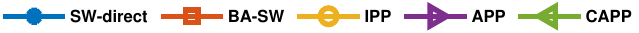}
		\end{minipage}
	}
	\\
	\vspace{-0.12in}
	\centering
	\subfigure[\textbf{C6H6},$w=10$.]{
		\begin{minipage}[t]{0.24\linewidth}
			\centering
			\includegraphics[width=1\textwidth]{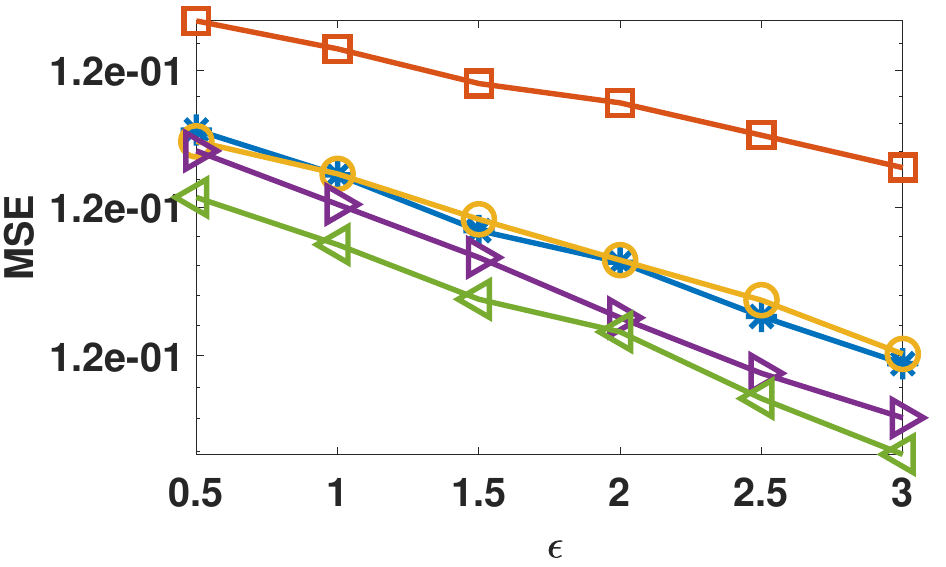}
			%\caption{fig1}
		\end{minipage}%
	}%
	\subfigure[\textbf{Volume},$w=10$.]{
		\begin{minipage}[t]{0.24\linewidth}
			\centering
			\includegraphics[width=1\textwidth]{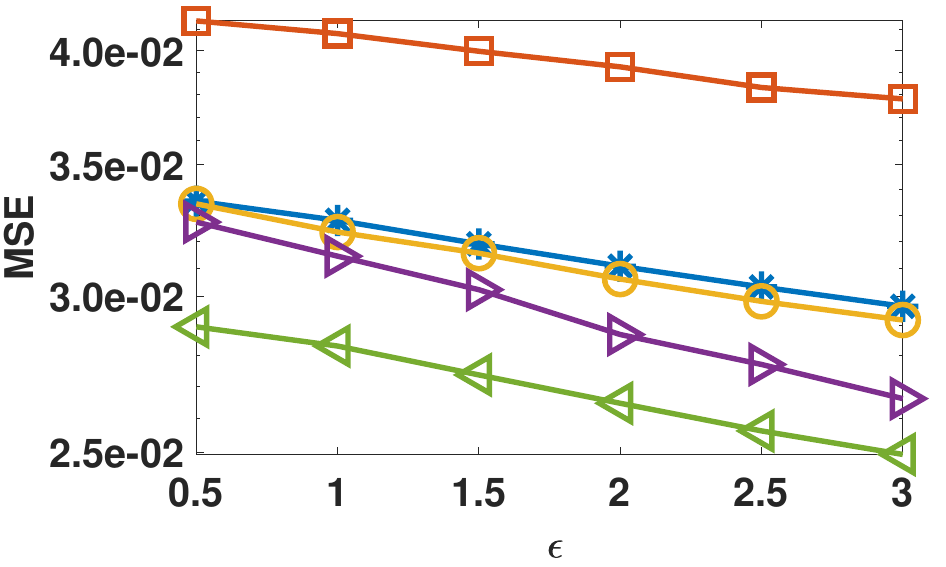}
			%\caption{fig2}
		\end{minipage}%
	}%
	\subfigure[\textbf{Taxi},$w=10$.]{
		\begin{minipage}[t]{0.24\linewidth}
			\centering
			\includegraphics[width=1\textwidth]{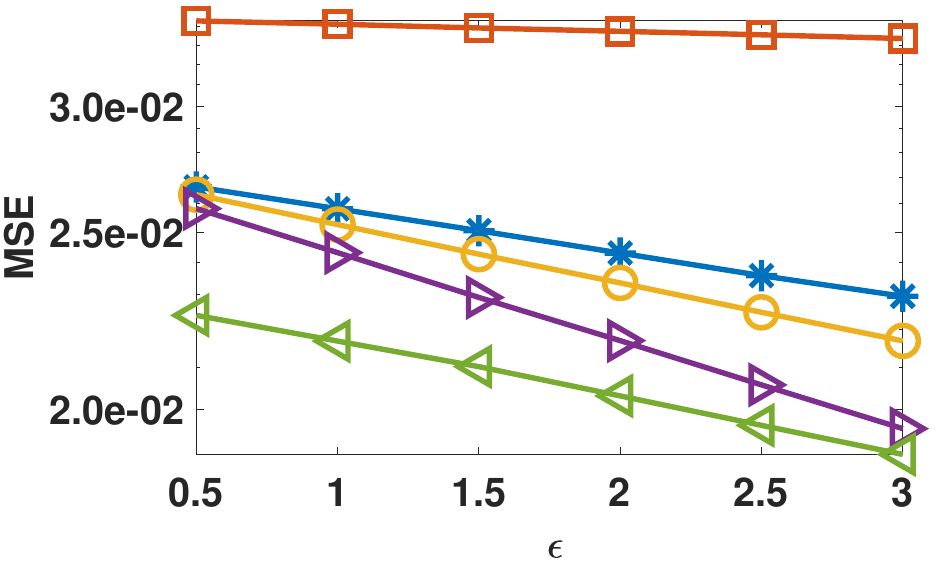}
			%\caption{fig2}
		\end{minipage}
	}%
	\subfigure[\textbf{Power},$w=10$.]{
		\begin{minipage}[t]{0.24\linewidth}
			\centering
			\includegraphics[width=1\textwidth]{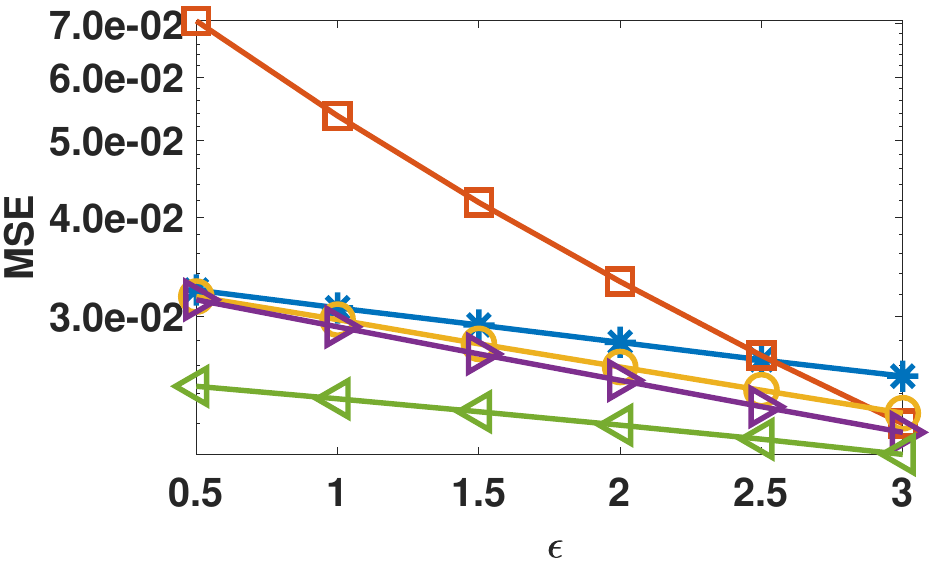}
			%\caption{fig2}
		\end{minipage}
	}%
	\vspace{-0.12in}
	\centering
	\subfigure[\textbf{C6H6},$w=30$.]{
		\begin{minipage}[t]{0.24\linewidth}
			\centering
			\includegraphics[width=1\textwidth]{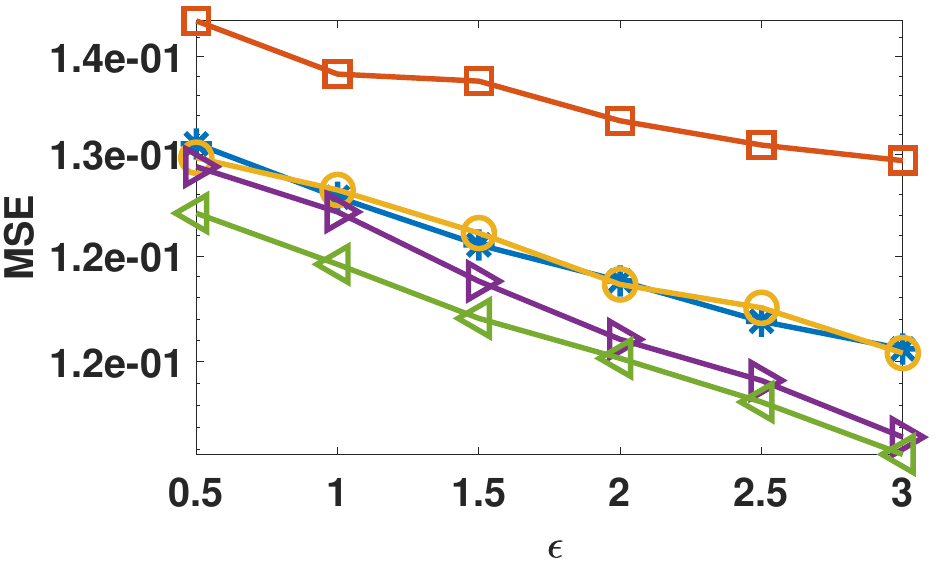}
			%\caption{fig1}
		\end{minipage}%
	}%
	\subfigure[\textbf{Volume},$w=30$.]{
		\begin{minipage}[t]{0.24\linewidth}
			\centering
			\includegraphics[width=1\textwidth]{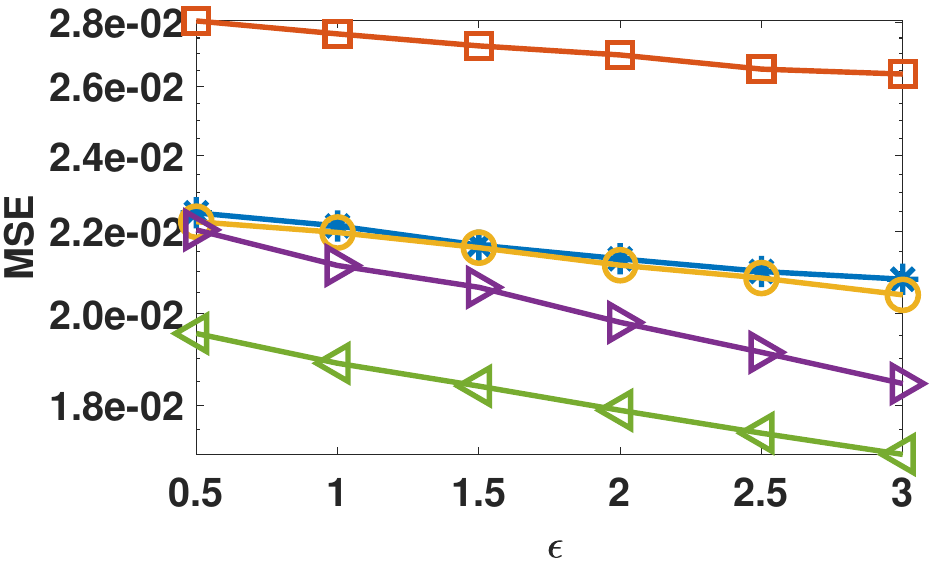}
			%\caption{fig2}
		\end{minipage}%
	}%
	\subfigure[\textbf{Taxi},$w=30$.]{
		\begin{minipage}[t]{0.24\linewidth}
			\centering
			\includegraphics[width=1\textwidth]{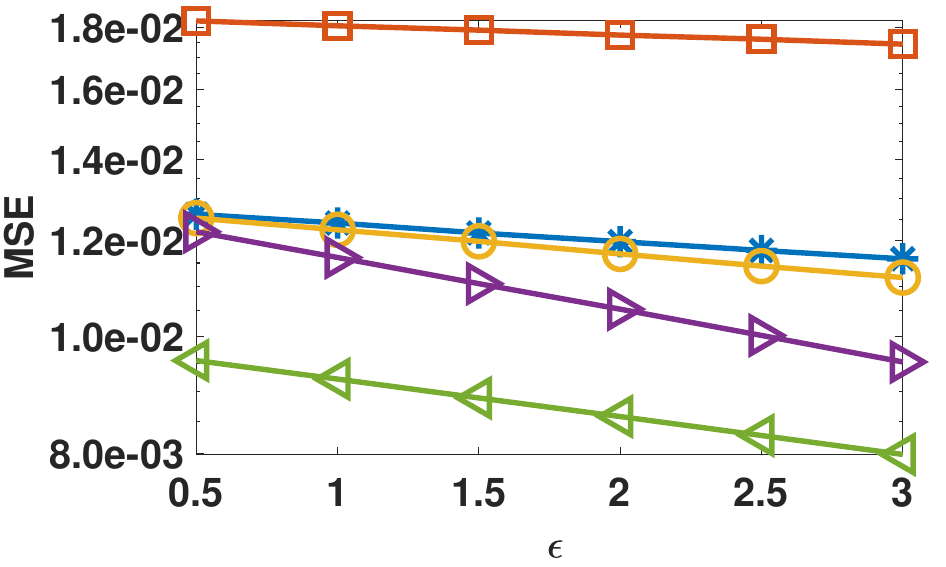}
			%\caption{fig2}
		\end{minipage}
	}%
	\subfigure[\textbf{Power},$w=30$.]{
		\begin{minipage}[t]{0.24\linewidth}
			\centering
			\includegraphics[width=1\textwidth]{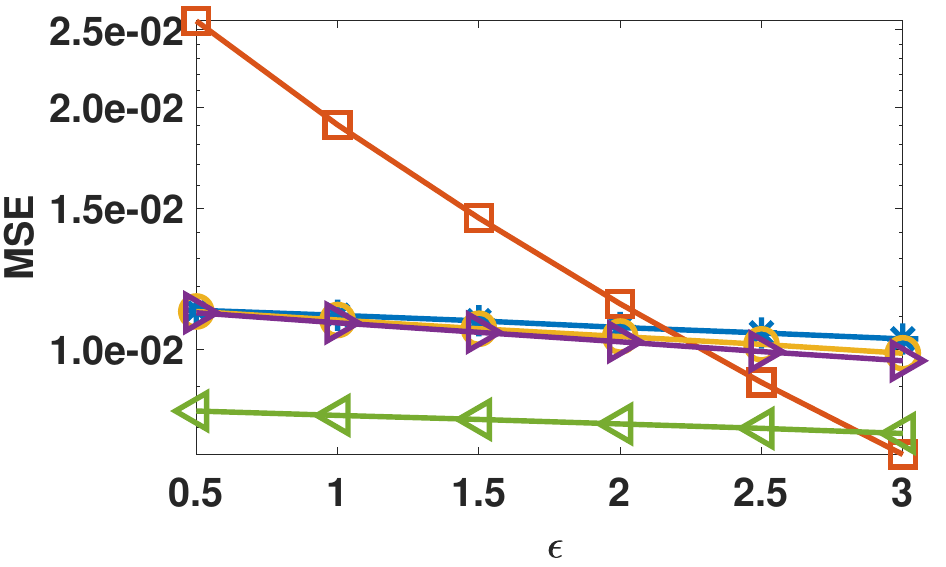}
			%\caption{fig2}
		\end{minipage}
	}%
	\vspace{-0.12in}
	\centering
	\subfigure[\textbf{C6H6},$w=50$.]{
		\begin{minipage}[t]{0.24\linewidth}
			\centering
			\includegraphics[width=1\textwidth]{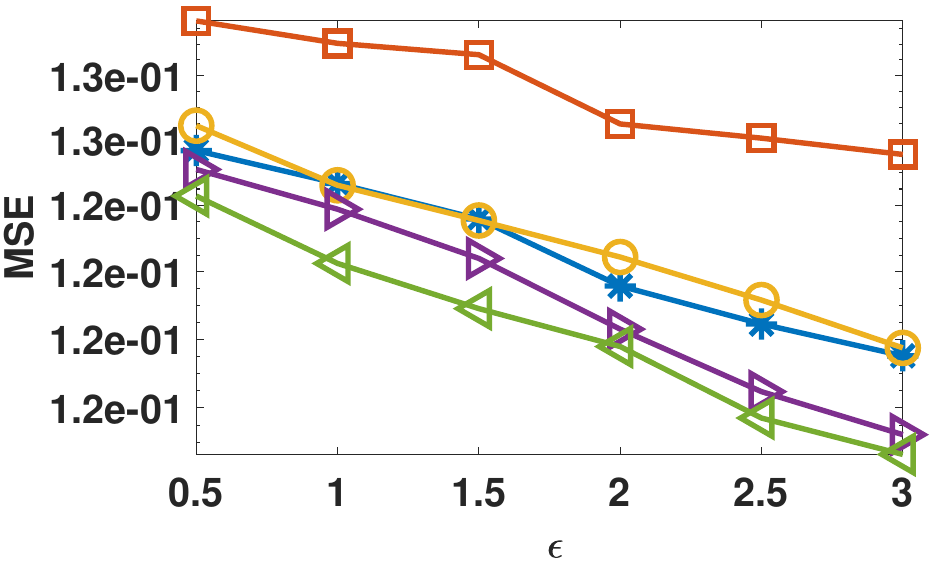}
			%\caption{fig1}
		\end{minipage}%
	}%
	\subfigure[\textbf{Volume},$w=50$.]{
		\begin{minipage}[t]{0.24\linewidth}
			\centering
			\includegraphics[width=1\textwidth]{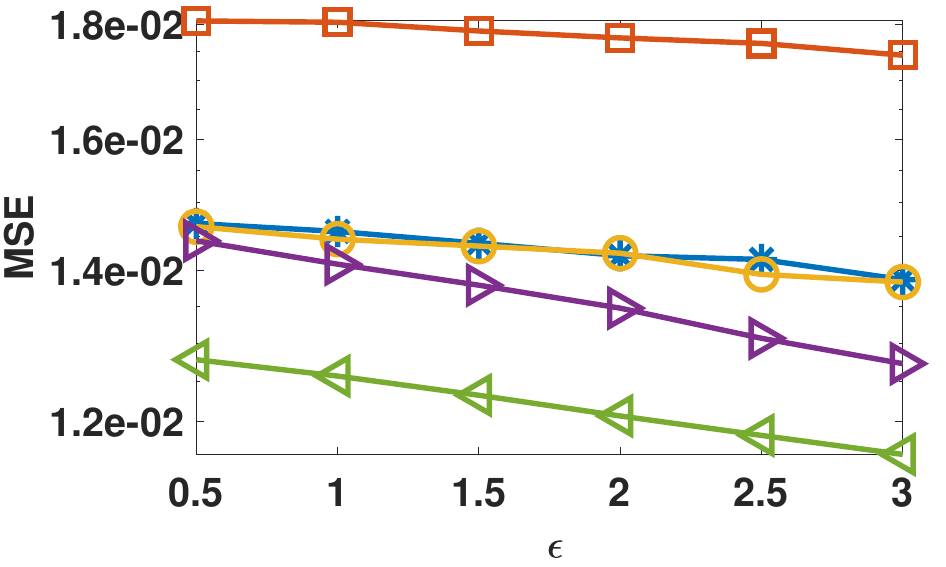}
			%\caption{fig2}
		\end{minipage}%
	}%
	\subfigure[\textbf{Taxi},$w=50$.]{
		\begin{minipage}[t]{0.24\linewidth}
			\centering
			\includegraphics[width=1\textwidth]{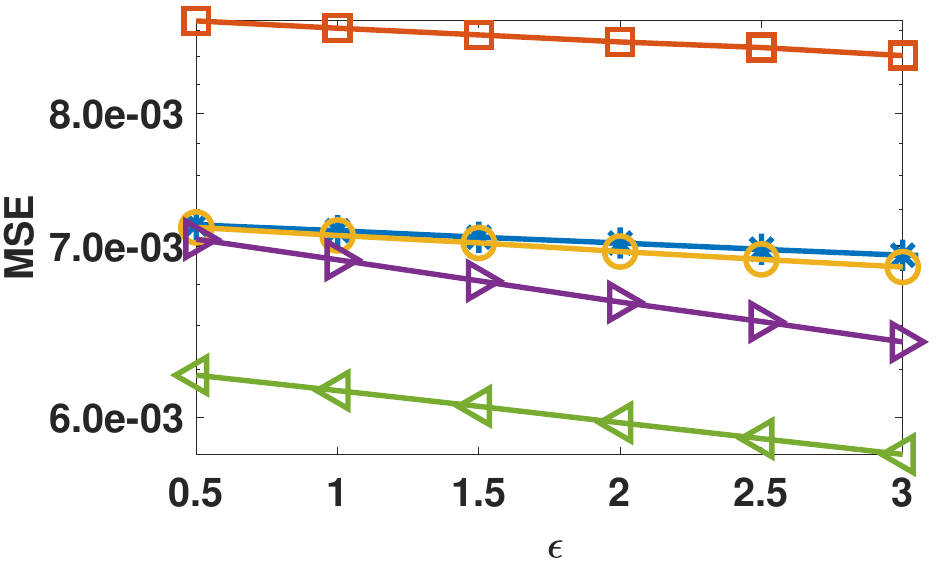}
			%\caption{fig2}
		\end{minipage}
	}%
	\subfigure[\textbf{Power},$w=50$.]{
		\begin{minipage}[t]{0.24\linewidth}
			\centering
			\includegraphics[width=1\textwidth]{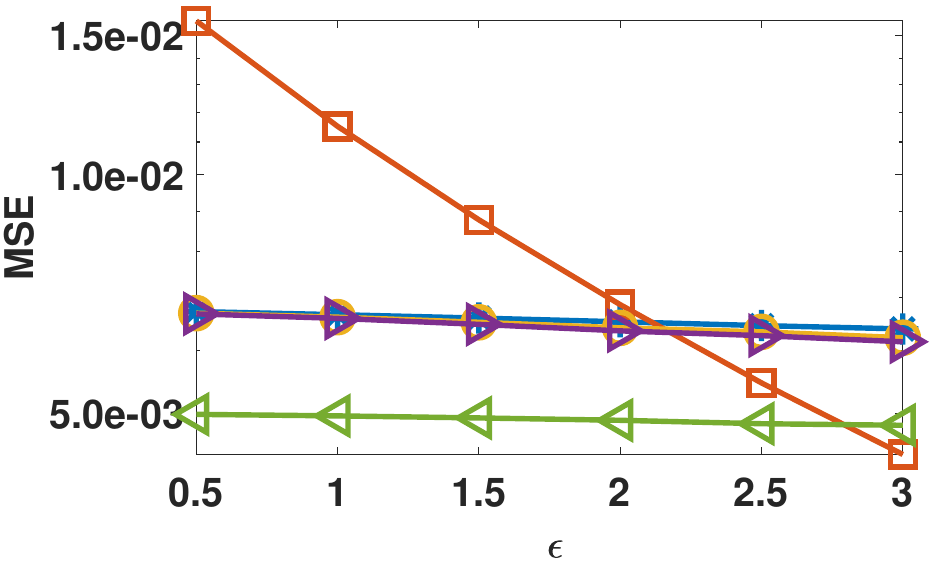}
			%\caption{fig2}
		\end{minipage}
	}%
	%	\vspace{-0.12in}
	\caption{MSE comparison w.r.t. $\epsilon$ for perturbation parameterization based algorithms vs SW-direct}
	\label{MeanEstimation1}
	%	\vspace{-0.22in}
\end{figure*}

\section{Experimental Results}
\label{exp}
In this section, we conduct experiment to validate the effectiveness of our proposed solutions. 

\subsection{Experimental Setup}
The experiments are conducted on a PC equipped with an AMD Ryzen 7 2700X eight-core processor, 64GB RAM, and Windows 10, using MATLAB R2019b. The experiments are conducted 100 times and the results were subsequently averaged. All datasets and code are available online\footnote{https://github.com/RONGDUGithub/CAPP}.

\subsubsection{Comparison algorithms}
\textcolor{black}{We categorize the compared algorithms into two sets: those without sampling and those with sampling. For the first set, we analyze several approaches, including a naive algorithm (\textbf{SW-direct}) that directly applies the SW perturbation method to each value, and ToPL~\cite{wang2021continuous}, a state-of-the-art method for mean estimation from subsequences. We also compare BA-SW~\cite{ren2022ldp,kellaris2014differentially}, which integrates budget absorption with the SW mechanism to preserve the privacy budget by skipping the transmission of minimally changed data points. Our proposed algorithms in this set include IPP, APP, and CAPP. For the second set, involving sampling, we compare a naive approach (\textbf{Sampling}) that samples stream values before applying the SW mechanism, alongside our proposed methods APP-S and CAPP-S.}

As for the smoothing step, it is necessary to determine a smoothing window size. A larger window size offers advantages in estimating the mean; however, it may result in unfavorable outcomes for stream data publication. In our experimental setting, we choose the smoothing window size of 3.

%$n$, the number of users. As a significant amount of noise is added to each individual user’s data in terms of LDP, a large population would be required to effectively remove the noise and reveal the true top-$k$ items. 

%\textbf{Dataset.}
%\textbf{Steam Video Games}\footnote{https://www.kaggle.com/datasets/tamber/steam-video-games?resource=download}.

\subsubsection{Utility metrics}
\color{black}
We classify our performance evaluations into three categories: (a) mean estimation with Mean Squared Error (MSE) assessment \cite{lehmann2006theory}, (b) direct stream data release evaluated by cosine distance \cite{singhal2001modern}, and (c) distribution analysis of subsequence means using Wasserstein Distance \cite{Ludger1985The}. 
\color{black}

\textbf{MSE.} The Mean Squared Error, denoted as $\text{MSE}$, is a commonly used metric to measure the average squared difference between predicted values $\hat{y}_i$ and true values $y_i$ in regression tasks. Given a set of $n$ pairs of predicted and true values, the MSE is calculated as:
\begin{equation*}
	MSE = \frac{1}{n} \sum_{i=1}^{n} (\hat{y}_i - y_i)^2.
\end{equation*}

%The MSE quantifies the average magnitude of the squared errors, providing a measure of the overall goodness-of-fit of the regression model. The MSE value is non-negative, with lower values indicating better model performance. In our experiments, we select 50 groups of subsequences and measure the MSE of their mean values.

\textbf{Cosine distance.} The cosine distance \cite{singhal2001modern}, denoted as $d_{\text{cos}}$, measures the dissimilarity between two vectors in a vector space. Given two vectors $\mathbf{u}$ and $\mathbf{v}$, the cosine distance is calculated as:
\begin{equation*}
	\setlength{\abovedisplayskip}{2pt}
	\setlength{\belowdisplayskip}{2pt}
	d_{\text{cos}}(\mathbf{u}, \mathbf{v}) = 1 - \frac{\mathbf{u} \cdot \mathbf{v}}{|\mathbf{u}| |\mathbf{v}|}.
\end{equation*}

where $d_{cos}\cdot$ represents the dot product of the vectors, and $|\cdot|$ denotes the Euclidean norm. If the cosine distance is high (e.g., close to 1), indicating that the vectors are dissimilar.

\color{black}
\textbf{Wasserstein Distance.} Also known as Earth Mover's Distance (EMD), it measures the minimum cost of transforming one probability distribution into another. The Wasserstein distance is computed as the sum of absolute differences between two empirical CDFs $F$ and $G$, as shown below:
\begin{equation*}
	\setlength{\abovedisplayskip}{2pt}
	\setlength{\belowdisplayskip}{2pt}
	W(F,G) = \sum_{i=1}^{n} |F_i - G_i|.
\end{equation*}
A smaller Wasserstein distance indicates higher similarity between the distributions.
\subsubsection{Datasets}
The experiments are conducted on four real-world datasets. The \textbf{Volume} and \textbf{C6H6} datasets each consist of a single data stream from a single user, while the \textbf{Taxi} and \textbf{Power} datasets contain multiple data streams from multiple users. For the latter two datasets, we perform both user-level and crowd-level statistical analysis.
\color{black}

\textbf{Volume\footnote{https://archive.ics.uci.edu/dataset/492/metro+interstate+traffic+volume}.} The dataset comprises hourly measurements of westbound traffic volume from the Minnesota Department of Transportation (MNDoT) Automatic Traffic Recorder (ATR) station 301, strategically situated midway between Minneapolis and St. Paul along Interstate 94. The data stream contains a total of 48204 valid entries. 

\textbf{C6H6\footnote{https://archive.ics.uci.edu/dataset/360/air+quality}.} The dataset encompasses 9,358 instances of hourly averaged data collected by a suite of five metal oxide chemical sensors integrated into an Air Quality Chemical Multisensor Device, spanning from March 2004 to February 2005. We focus on the subset of data related to the fluctuations in benzene concentration levels. 

\textbf{Taxi\footnote{https://www.microsoft.com/en-us/research/publication/t-drive-trajectory-data-sample/}.} The dataset captures the real-time trajectories of 10,357 taxis in Beijing, recorded from February 2 to February 8, 2008. We focused on extracting the latitude information for each taxi at 1307 specific timestamps from 1500 drivers.

\textbf{Power\footnote{https://www.cs.ucr.edu/\%20eamonn/time\%20series\%20data/}.} This dataset is sourced from the UCR Time Series Data Mining Archive. It contains the power usage data of 25,562 electrical devices, with each time series consisting of 96 stream values.

\subsection{Overall Results for Perturbation Parameterization Algorithms}
\subsubsection{The results for mean estimation} 
\mbox{}\par%\vspace{\baselineskip}
\textbf{Comparison of SW-based algorithms and ToPL.}  In Table \ref{tableresult}, our comparative analysis primarily focuses on two categories of algorithms. The first category encompasses those based on the SW mechanism, which includes the SW-direct, IPP, and APP algorithms. The second category is the ToPL algorithm \cite{wang2021continuous}, which first uses the SW algorithm to remove outlier data and obtain a reasonable data range, then perturbs the data using the HM mechanism \cite{wang2019collecting}. 

In Table \ref{tableresult}, we find that the MSE of ToPL is more than 100 times larger than others for mean estimation. This is because the perturbation threshold mechanism of SW is much smaller than that of HM. For example, when the privacy budget per window is 1, and there are 20 time slots, the privacy budget allocated to each time slot is 0.05. The perturbation threshold for SW is mapped from [0,1] to [-0.4836, 1.4836]. In contrast, the HM perturbation threshold is mapped from [-1,1] to [-80,80]. Thus, it is evident that ToPL incurs a larger error with a smaller privacy budget. Furthermore, the order of magnitude of this error increases exponentially as epsilon decreases. Due to the significant discrepancy of ToPL's results compared to PP methods, we will not include them in our subsequent experimental result figures.

% \usepackage{multirow}
\begin{table}[]
	\caption{Results for ToPl with SW-based algorithms}
	\centering % This command centers the table
	\begin{tabular}{|c|c|c|c|c|c|}
		\hline
		MSE                                                                          & $w$ & \multicolumn{1}{c|}{SW-direct} & \multicolumn{1}{c|}{IPP} & \multicolumn{1}{c|}{APP} & \multicolumn{1}{c|}{ToPL} \\ \hline
		\multirow{3}{*}{\begin{tabular}[c]{@{}c@{}}C6H6\\ $\epsilon=1$\end{tabular}}  & 20  & 0.131                     & 0.131                   & 0.129                   & 25.214                    \\ \cline{2-6} 
		& 40  & 0.125                     & 0.126                   & 0.125                   & 51.613                    \\ \cline{2-6} 
		& 60  & 0.124                     & 0.124                   & 0.123                   & 80.070                    \\ \hline
		\multirow{3}{*}{\begin{tabular}[c]{@{}c@{}}Taxi\\ $\epsilon=1$\end{tabular}} & 20  & 1.28E-04                  & 1.25E-04                & 1.19E-04                & 0.043                     \\ \cline{2-6} 
		& 40  & 4.9E-05                   & 4.8E-05                 & 4.6E-05                 & 0.132                     \\ \cline{2-6} 
		& 60  & 3.7E-05                   & 3.6E-05                 & 3.5E-05                 & 0.221                     \\ \hline
	\end{tabular}
	\vspace{-0.22in}
	%	\caption{Your caption here.} % Add a caption if needed
	\label{tableresult} % Label your table for reference if needed
\end{table}
\begin{figure*}[ht]
	\hspace{0.25in}
	%\vspace{-0.06in}
	{
		\begin{minipage}{15cm}
			\centering
			\includegraphics[scale=0.7]{legend_nosample.pdf}
		\end{minipage}
	}
	\\
	%\vspace{-0.12in}
	\centering
	\subfigure[\textbf{C6H6},$w=10$.]{
		\begin{minipage}[t]{0.24\linewidth}
			\centering
			\includegraphics[width=1\textwidth]{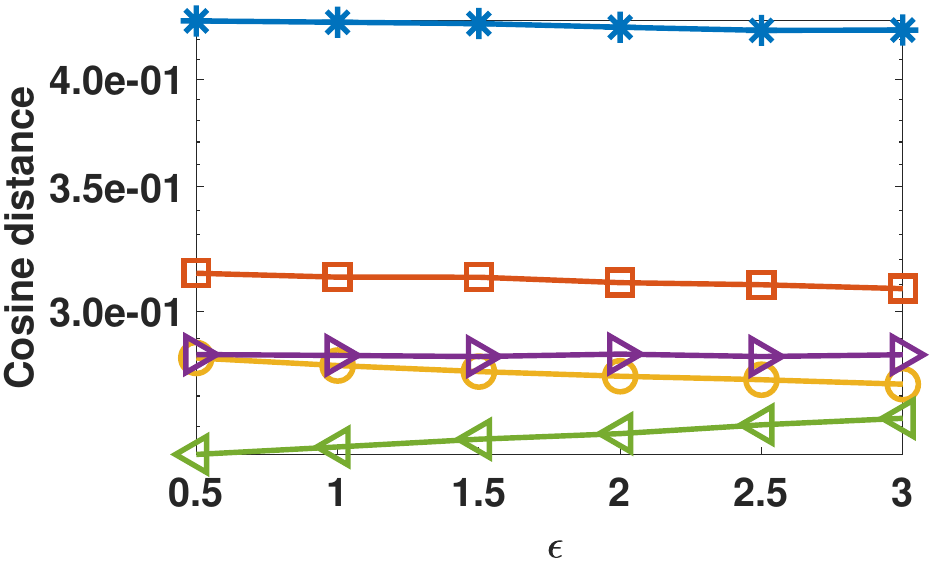}
			%\caption{fig1}
		\end{minipage}%
	}%
	\subfigure[\textbf{Volume},$w=10$.]{
		\begin{minipage}[t]{0.24\linewidth}
			\centering
			\includegraphics[width=1\textwidth]{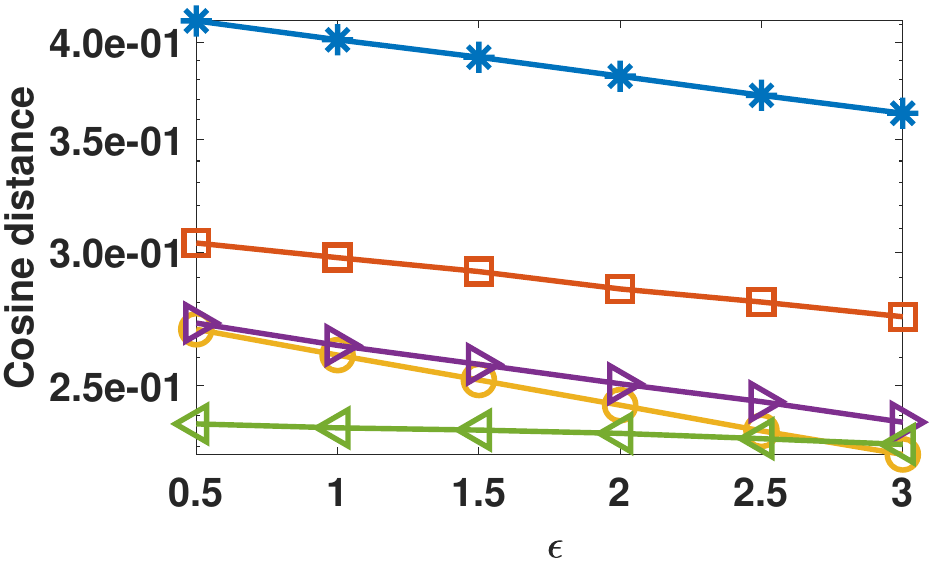}
			%\caption{fig2}
		\end{minipage}%
	}%
	\subfigure[\textbf{Taxi},$w=10$.]{
		\begin{minipage}[t]{0.24\linewidth}
			\centering
			\includegraphics[width=1\textwidth]{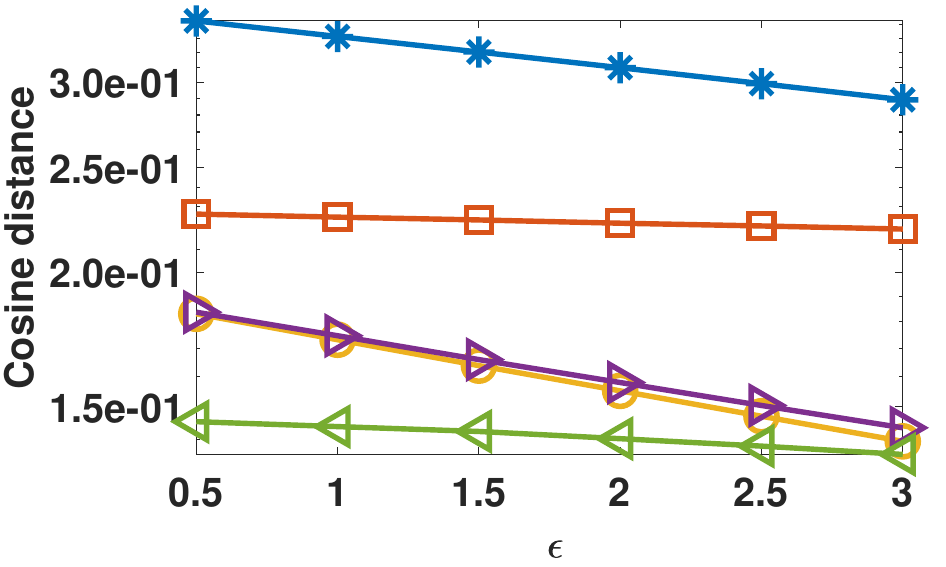}
			%\caption{fig2}
		\end{minipage}
	}%
	\subfigure[\textbf{Power},$w=10$.]{
		\begin{minipage}[t]{0.24\linewidth}
			\centering
			\includegraphics[width=1\textwidth]{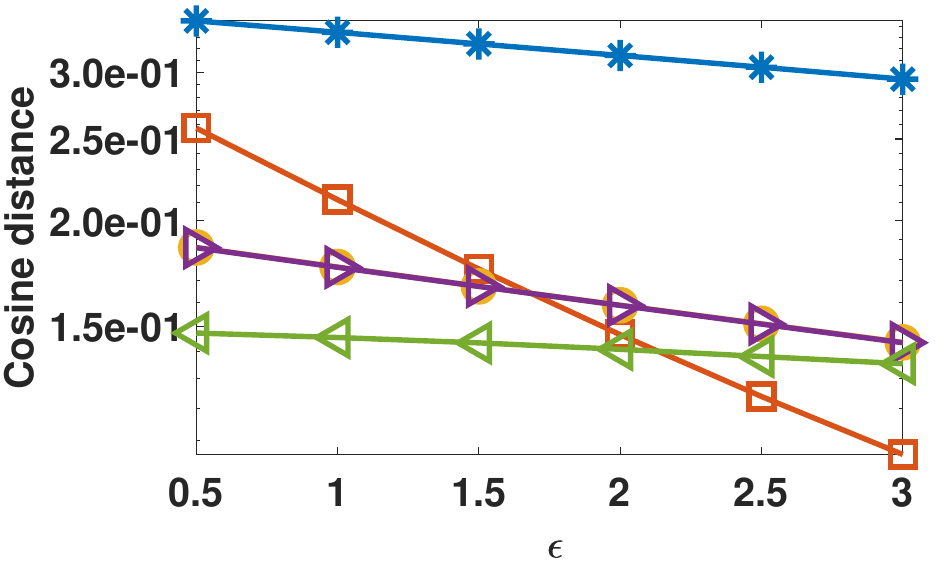}
			%\caption{fig2}
		\end{minipage}
	}%
	\vspace{-0.06in}
	\centering
	\subfigure[\textbf{C6H6},$w=30$.]{
		\begin{minipage}[t]{0.24\linewidth}
			\centering
			\includegraphics[width=1\textwidth]{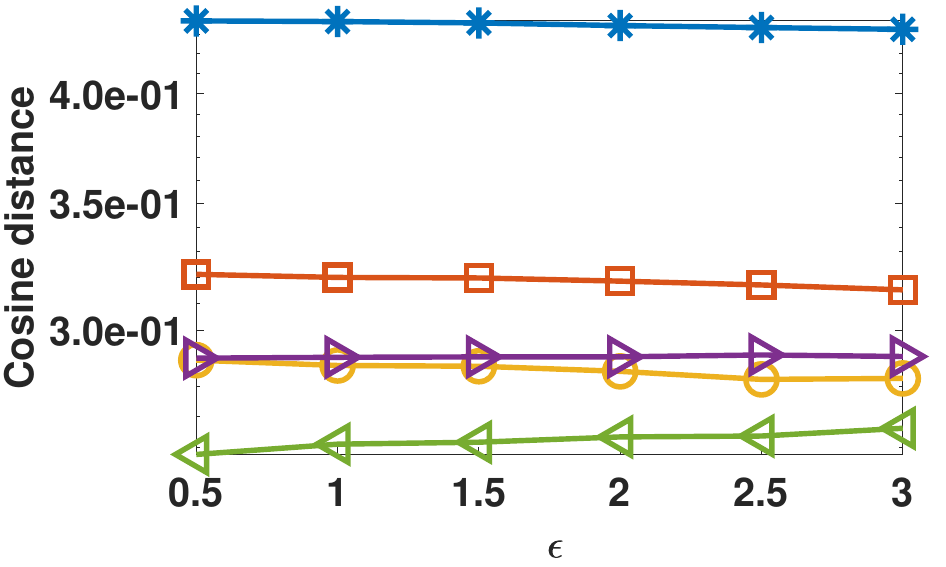}
			%\caption{fig1}
		\end{minipage}%
	}%
	\subfigure[\textbf{Volume},$w=30$.]{
		\begin{minipage}[t]{0.24\linewidth}
			\centering
			\includegraphics[width=1\textwidth]{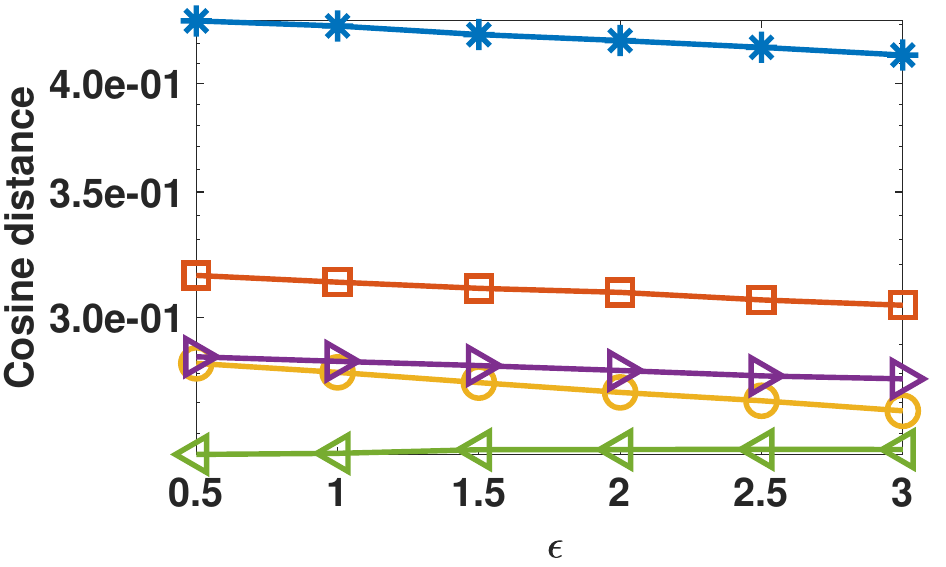}
			%\caption{fig2}
		\end{minipage}%
	}%
	\subfigure[\textbf{Taxi},$w=30$.]{
		\begin{minipage}[t]{0.24\linewidth}
			\centering
			\includegraphics[width=1\textwidth]{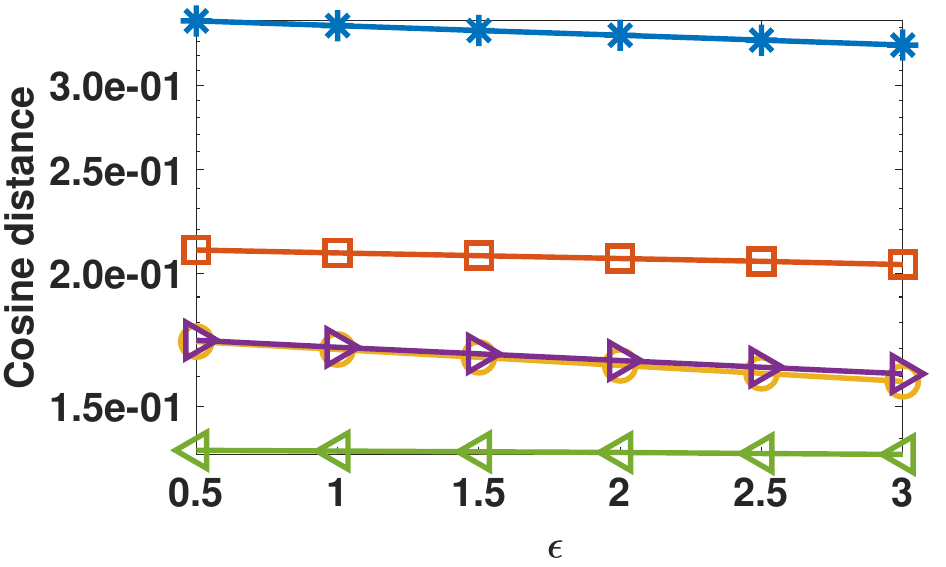}
			%\caption{fig2}
		\end{minipage}
	}%
	\subfigure[\textbf{Power},$w=30$.]{
		\begin{minipage}[t]{0.24\linewidth}
			\centering
			\includegraphics[width=1\textwidth]{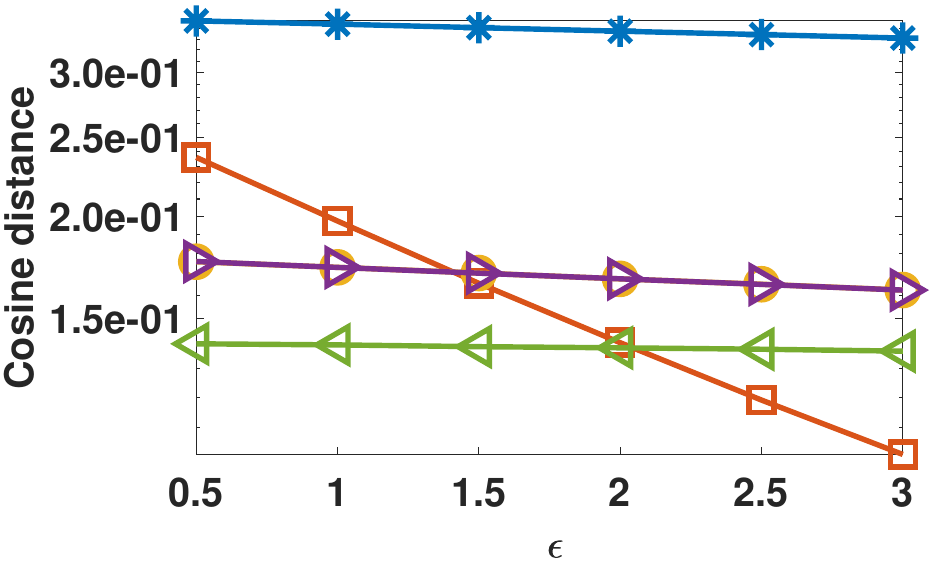}
			%\caption{fig2}
		\end{minipage}
	}%
	\vspace{-0.06in}
	\centering
	\subfigure[\textbf{C6H6},$w=50$.]{
		\begin{minipage}[t]{0.24\linewidth}
			\centering
			\includegraphics[width=1\textwidth]{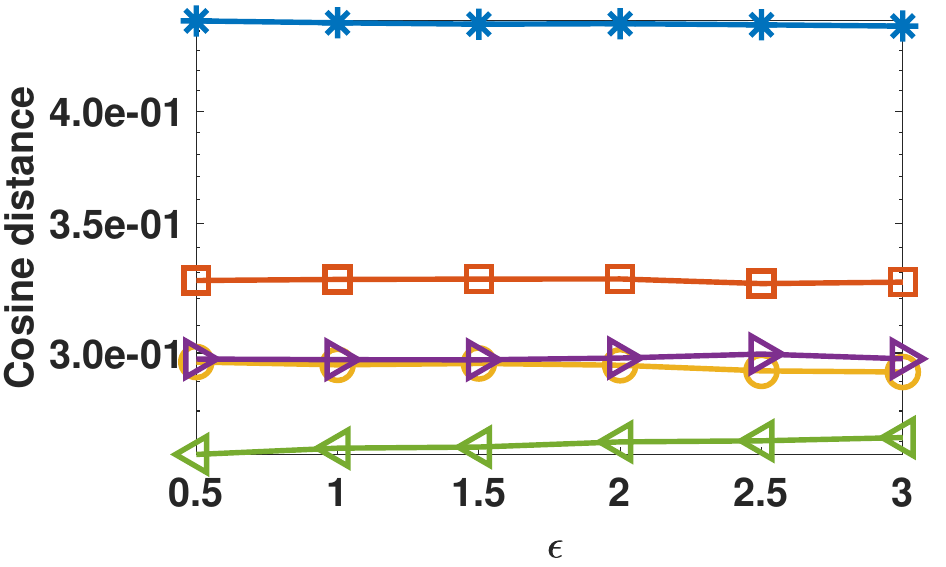}
			%\caption{fig1}
		\end{minipage}%
	}%
	\subfigure[\textbf{Volume},$w=50$.]{
		\begin{minipage}[t]{0.24\linewidth}
			\centering
			\includegraphics[width=1\textwidth]{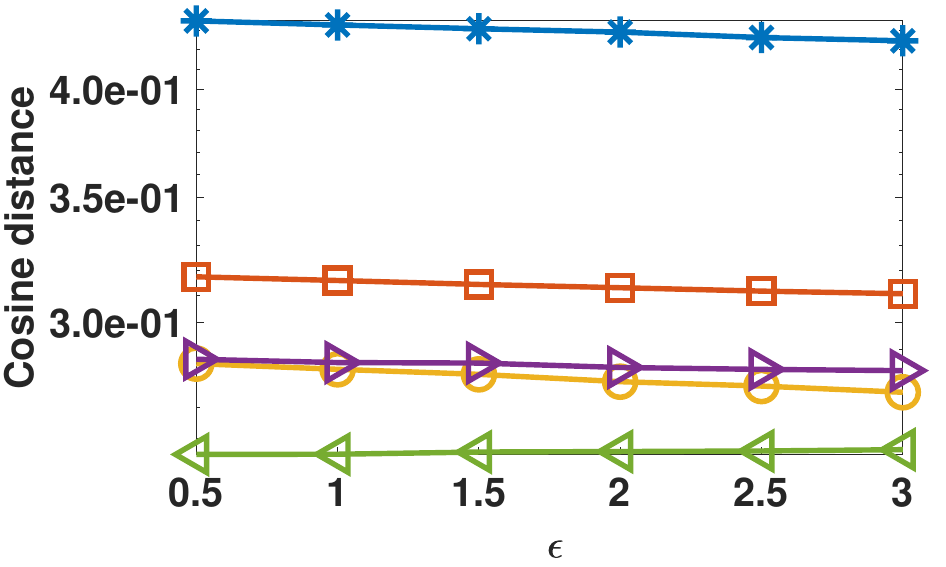}
			%\caption{fig2}
		\end{minipage}%
	}%
	\subfigure[\textbf{Taxi},$w=50$.]{
		\begin{minipage}[t]{0.24\linewidth}
			\centering
			\includegraphics[width=1\textwidth]{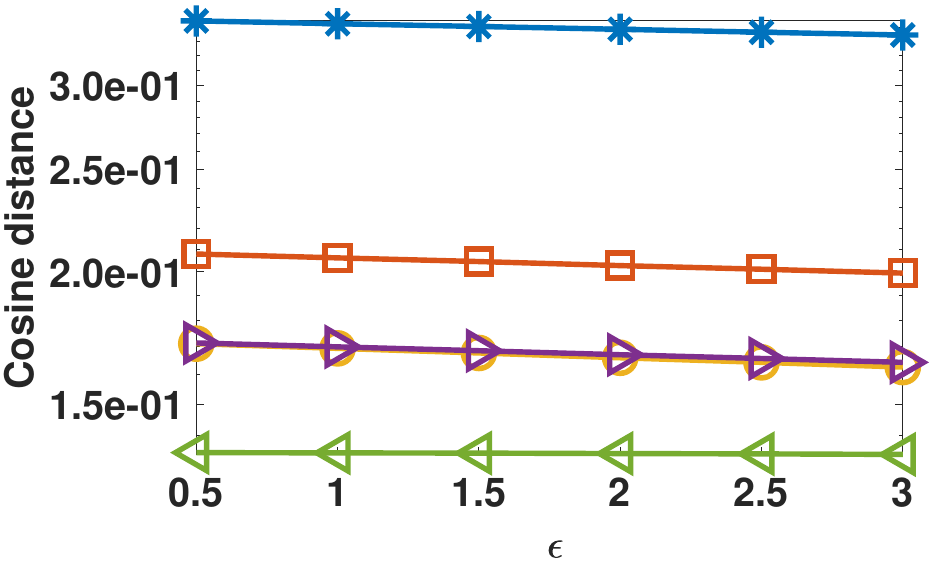}
			%\caption{fig2}
		\end{minipage}
	}%
	\subfigure[\textbf{Power},$w=50$.]{
		\begin{minipage}[t]{0.24\linewidth}
			\centering
			\includegraphics[width=1\textwidth]{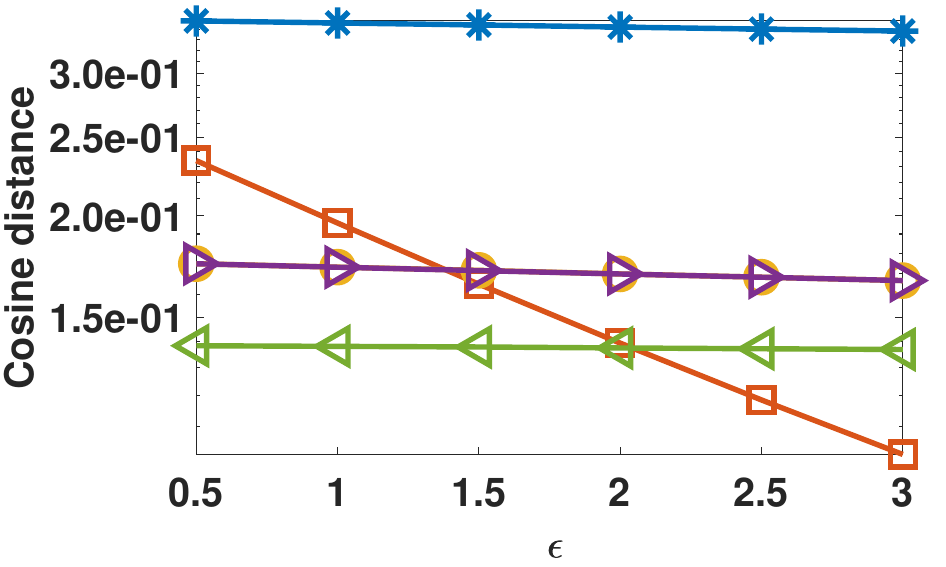}
			%\caption{fig2}
		\end{minipage}
	}%
	\vspace{-0.12in}
	\caption{Cosine distance comparison w.r.t. $\epsilon$ for perturbation parameterization algorithms vs SW-direct}
	\label{FrequencyEstimation1}
	\vspace{-0.22in}
\end{figure*}
\color{black}
\textbf{Comparison on PP algorithms and SW-direct \& BA-SW.}
In Figure \ref{MeanEstimation1}, we compare MSE for mean estimation among the SW-direct, BA-SW and our PP algorithms (IPP, APP, and CAPP). We compute the average over 50 randomly sampled time subsequences with length $w$.

We observe that in most subfigures, the MSE of BA-SW is the largest compared to our perturbation parameterization algorithms, with SW-direct performing as the second worst.
\color{black}The utility of APP is better than that of IPP, as APP leverages more perturbation information from the stream values, resulting in smaller errors in mean estimation, as demonstrated by Lemma \ref{lemmaforFLB}. Notably, CAPP achieves the best experimental results, which can be attributed to selecting a more suitable range for the input values. Under extremely limited privacy budgets, we reduce the mechanism's sensitivity by clipping the input values to a smaller range and normalizing them, thereby producing perturbed values with better distinguishability.
This superior performance is consistent across different window sizes $w$, as shown in Figures \ref{MeanEstimation1} (b), (f), and (j).

\color{black}
It is worth noting that in Figures \ref{MeanEstimation1} (d), (h), and (l), the BA-SW algorithm performs best on the \textbf{Power} dataset when $\epsilon$ is relatively large. This is because many subsequences in the \textbf{Power} dataset are entirely composed of a unique constant value. BA-SW effectively identifies points with minimal variation and directly reuses previous values. By conserving the privacy budget, it introduces less noise, thereby improving the utility of subsequent uploads. 

\color{black}
\subsubsection{The results for stream data publication}
In Figure \ref{FrequencyEstimation1}, the basic setup is the same as that in Figure \ref{MeanEstimation1}. We then compare the cosine distance for different algorithms, where a smaller cosine distance indicates that the estimated data stream is closer to the ground truth.

We can see that the SW-direct's cosine distance is the largest compared to our perturbation parameterization algorithms, indicating the worst performance. This is because our algorithms utilize the perturbation of previous data to reduce errors. We can see that CAPP demonstrates the best performance among the perturbation parameterization algorithms. This is because the method uses perturbation information and sets an appropriate $[l,u]$, further expanding the advantage. It is noteworthy that for \textbf{C6H6}, the utility of CAPP increases when $\epsilon$ increases. This is because the design of ${l,u}$ considers the worst case scenario, which may lead to improper selection of $T_{(e_p,e_d)}$ in Equation \ref{Teped} for different datasets. However, we can still see that CAPP achieves better utility than IPP and APP. In addition, we observe that IPP is slightly better than APP for stream data publication. This is because APP is influenced by multiple data points, which may result in poor estimates of the true characteristics of the data stream. In fact, APP is more suitable for mean estimation, as can be seen in Figure \ref{MeanEstimation1}. In contrast, IPP preserves the information of individual data points better and thus achieves better utility than APP for stream data publication.

Similarly, our approaches maintain their superior performance across different datasets under the same parameter settings, as shown in Figures \ref{FrequencyEstimation1} (a)-(d), (e)-(h), and (i)-(l). Observing the varied $w$ values specific to each dataset, as illustrated in Figures \ref{FrequencyEstimation1} (a), (e), and (i), it becomes clear that our perturbation parameterization algorithms sustain commendable performance across different $w$ configurations. This demonstrates the robustness of our perturbation parameterization algorithms.

\color{black}

\subsection{Overall Results for Sampling}
\subsubsection{The results for mean estimation}
In this group of experiments, we compare different parameters across datasets, including the query length $q$, window size $w$, and privacy budget $\epsilon$. To ensure the stability of mean estimation results, we randomly select 50 subsequences with length $q$ and average them.

As shown in Figures \ref{samplemean}, when $q$ and $w$ are fixed, the MSE decreases for all algorithms as the privacy budget increases. Among them, sampling performs the worst because it only uploads information from sampled points. However, APP-S and CAPP-S outperform other non-sampling methods. This experiment demonstrates that our perturbation parameterization algorithms retain their advantage even with sampling, showcasing their superiority across different datasets and various $w$ and $q$ combinations.

\begin{figure*}[th]
	%	\vspace{-0.1in}
	\hspace{0.25in}
	{
		\begin{minipage}{15cm}
			\centering
			\includegraphics[scale=0.7]{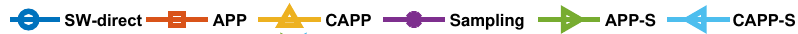}
		\end{minipage}
	}
	\\
	%	\vspace{-0.12in}
	\centering
	\subfigure[\textbf{Volume},$w=20,q=10$.]{
		\begin{minipage}[t]{0.24\linewidth}
			\centering
			\includegraphics[width=1\textwidth]{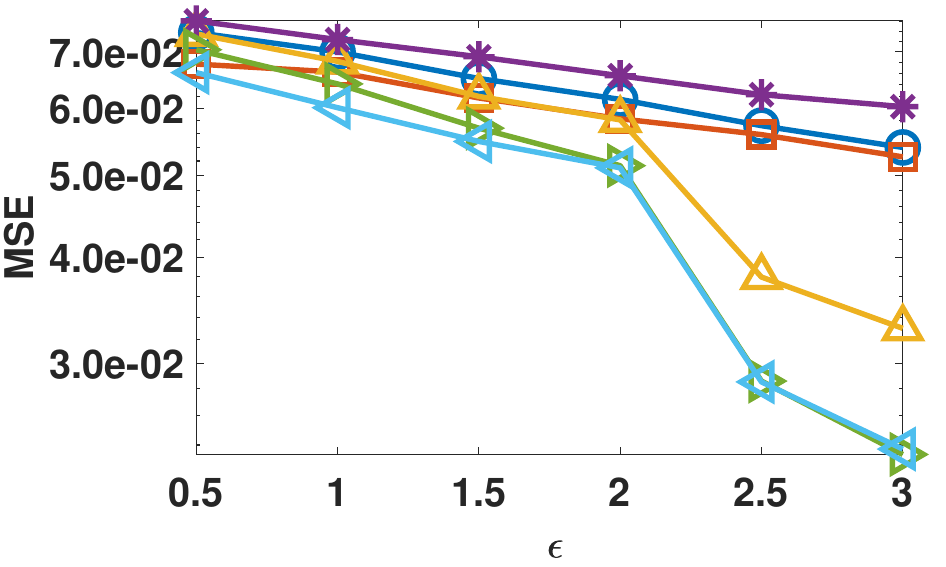}
			%\caption{fig1}
		\end{minipage}%
	}%
	\subfigure[\textbf{Volume},$w=30,q=10$.]{
		\begin{minipage}[t]{0.24\linewidth}
			\centering
			\includegraphics[width=1\textwidth]{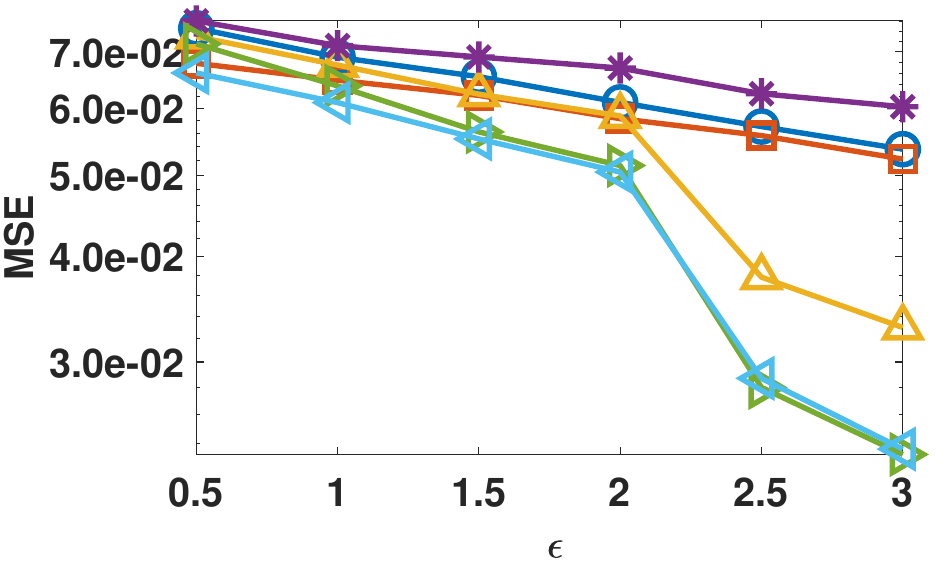}
			%\caption{fig2}
		\end{minipage}%
	}%
	\subfigure[\textbf{Volume},$w=30,q=20$.]{
		\begin{minipage}[t]{0.24\linewidth}
			\centering
			\includegraphics[width=1\textwidth]{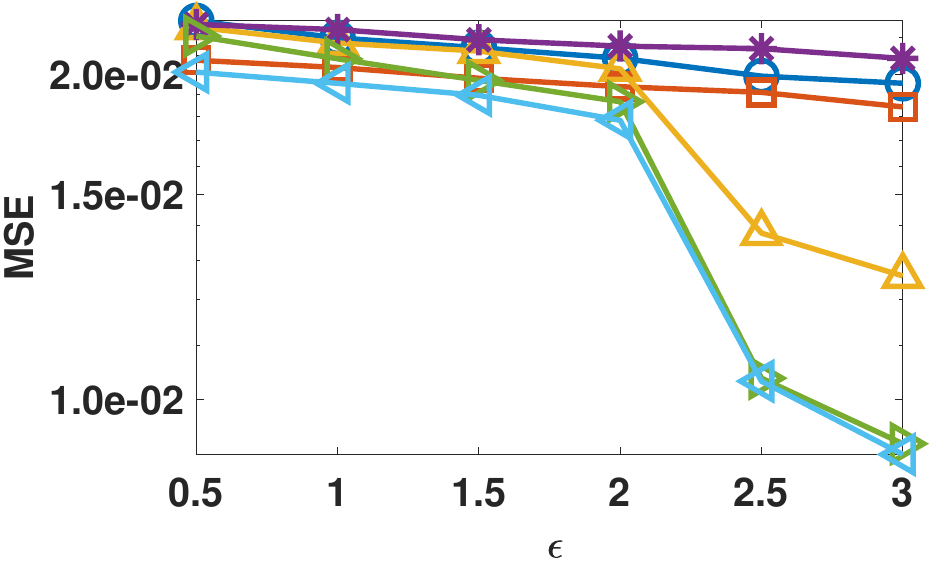}
			%\caption{fig2}
		\end{minipage}
	}%
	\subfigure[\textbf{Volume},,$w=30,q=40$.]{
		\begin{minipage}[t]{0.24\linewidth}
			\centering
			\includegraphics[width=1\textwidth]{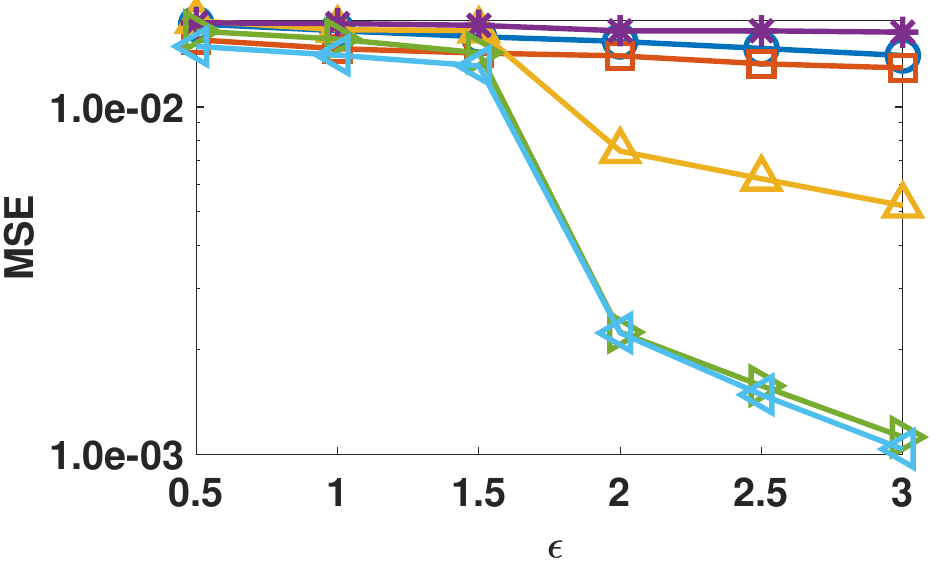}
			%\caption{fig2}
		\end{minipage}
	}%
	\vspace{-0.05in}
	\\
	\subfigure[\textbf{Volume},$w=20,q=30$.]{
		\begin{minipage}[t]{0.24\linewidth}
			\centering
			\includegraphics[width=1\textwidth]{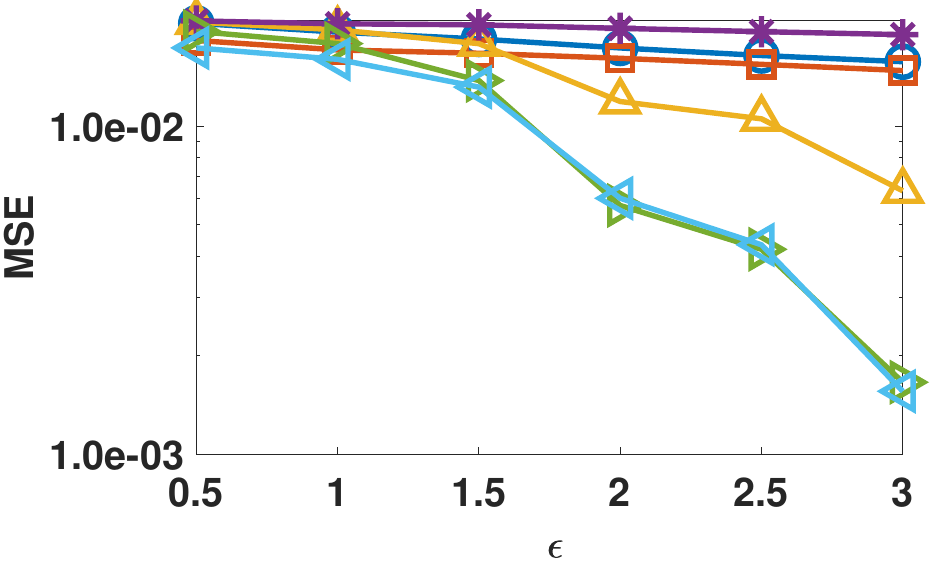}
			%\caption{fig1}
		\end{minipage}%
	}%
	\subfigure[\textbf{C6H6},$w=20,q=30$.]{
		\begin{minipage}[t]{0.24\linewidth}
			\centering
			\includegraphics[width=1\textwidth]{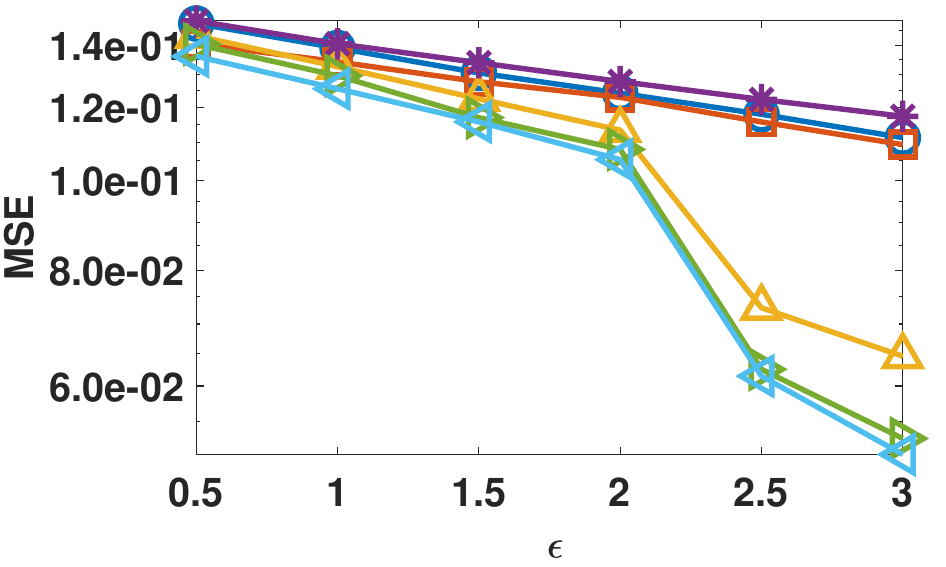}
			%\caption{fig2}
		\end{minipage}%
	}%
	\subfigure[\textbf{Power},$w=20,q=30$.]{
		\begin{minipage}[t]{0.24\linewidth}
			\centering
			\includegraphics[width=1\textwidth]{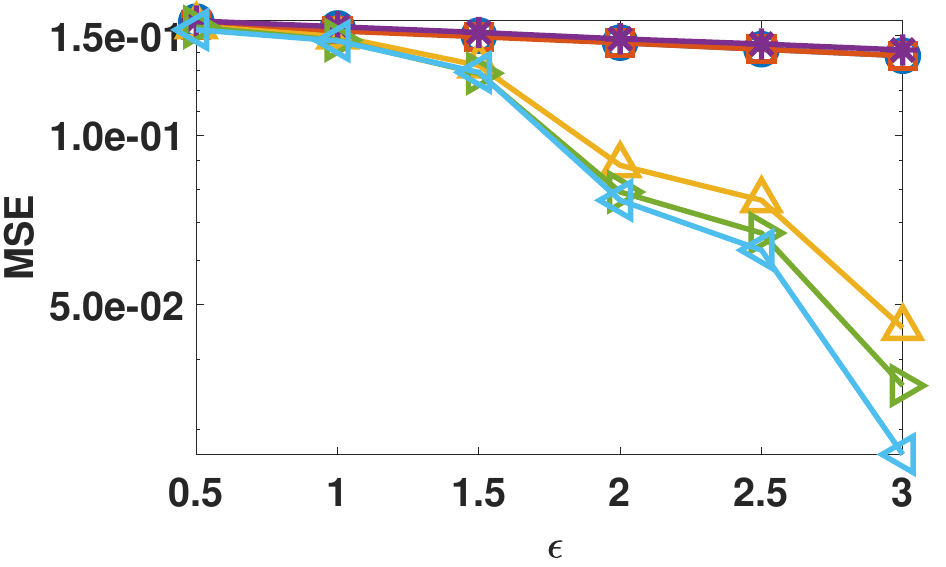}
			%\caption{fig2}
		\end{minipage}
	}%
	\subfigure[\textbf{Taxi},$w=20,q=30$.]{
		\begin{minipage}[t]{0.24\linewidth}
			\centering
			\includegraphics[width=1\textwidth]{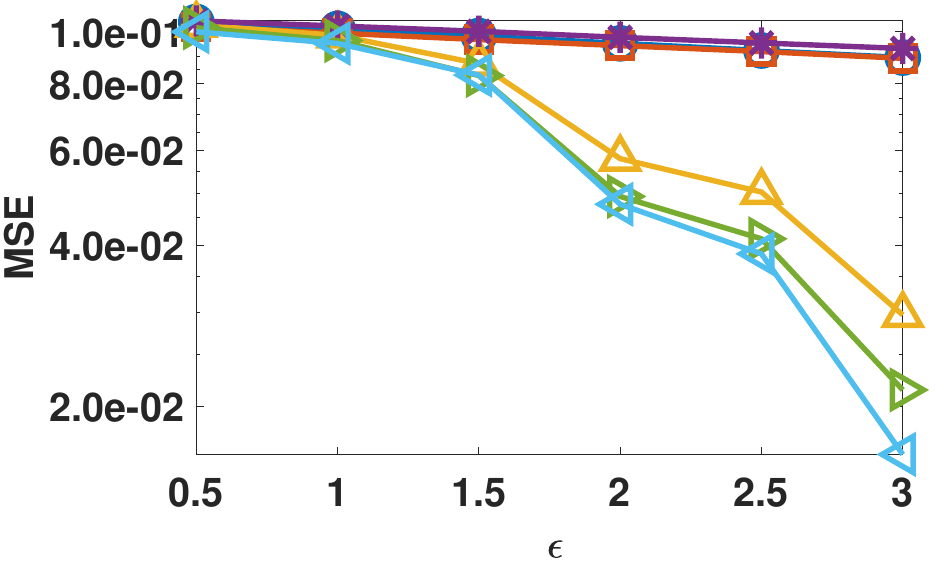}
			%\caption{fig2}
		\end{minipage}
	}%
	\vspace{-0.05in}
	\caption{MSE comparison w.r.t. $\epsilon$ for sampling-based algorithms vs non-sampling algorithms}
	\label{samplemean}
	\vspace{-0.1in}
\end{figure*}

\subsubsection{The results for stream data publication}
In Figure \ref{samplefrequency}, we evaluate the performance of stream data publication by comparing the cosine distance between sampling and non-sampling algorithms. The results reveal different patterns from those observed in mean estimation (Figure \ref{samplemean}). While sampling algorithms demonstrate significant advantages in mean estimation, they achieve similar performance levels in stream data publication. This pattern difference can be attributed to the reduced number of collecting values in each window due to sampling, which affects the two tasks differently: it helps reduce noise in mean estimation while still maintaining sufficient information for stream data publication. As illustrated in both figures, sampling-based algorithms perform well in mean estimation. While their performance in stream data publication is not as strong as CAPP, they still outperform APP.

\begin{figure*}[t]
	\hspace{0.25in}
	%	\vspace{-0.06in}
	{
		\begin{minipage}{10cm}
			\centering
			\includegraphics[scale=0.7]{legend_sampling.pdf}
		\end{minipage}
	}
	\\
	%	\vspace{-0.12in}
	\centering
	\subfigure[\textbf{Volume},$w=20,q=10$.]{
		\begin{minipage}[t]{0.24\linewidth}
			\centering
			\includegraphics[width=1\textwidth]{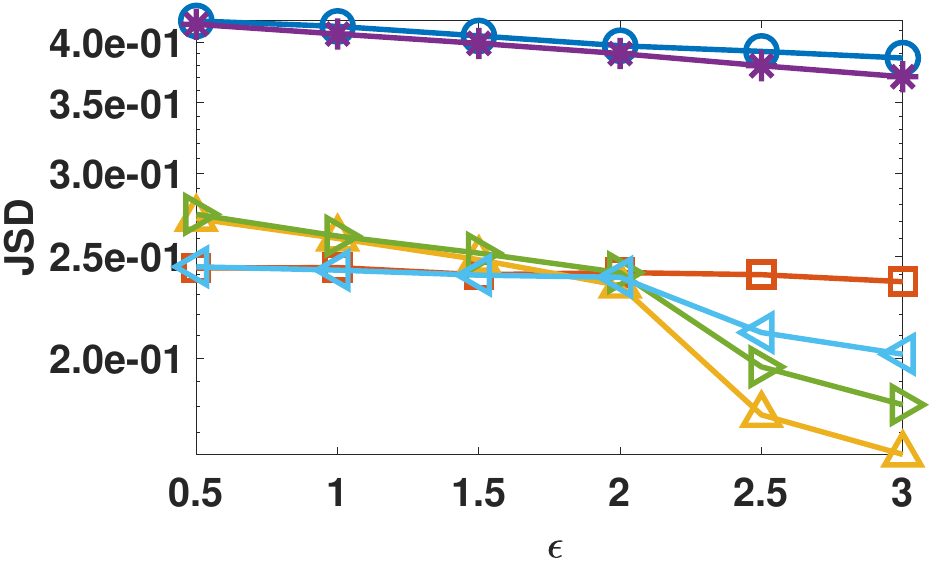}
			%\caption{fig1}
		\end{minipage}%
	}%
	\subfigure[\textbf{Volume},$w=30,q=10$.]{
		\begin{minipage}[t]{0.24\linewidth}
			\centering
			\includegraphics[width=1\textwidth]{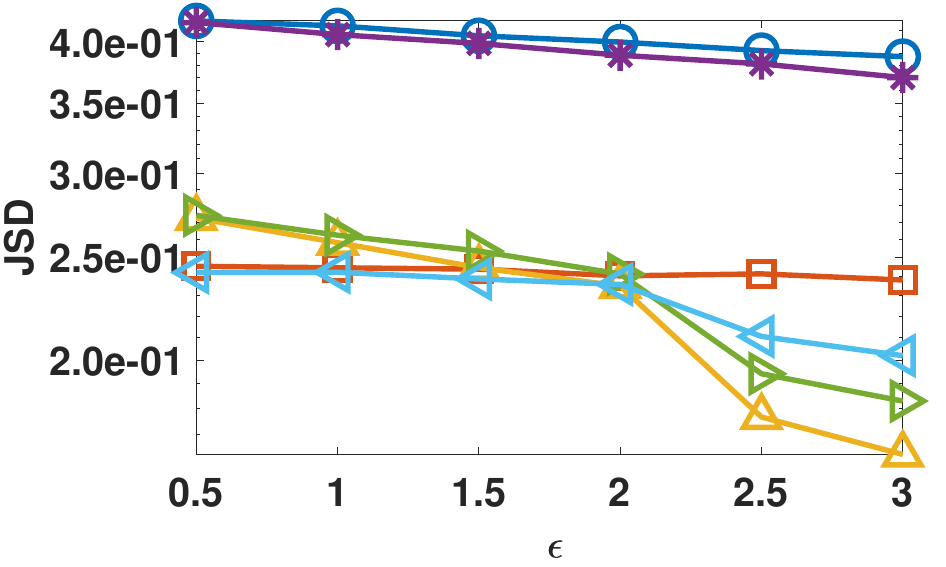}
			%\caption{fig2}
		\end{minipage}%
	}%
	\subfigure[\textbf{Volume},$w=30,q=20$.]{
		\begin{minipage}[t]{0.24\linewidth}
			\centering
			\includegraphics[width=1\textwidth]{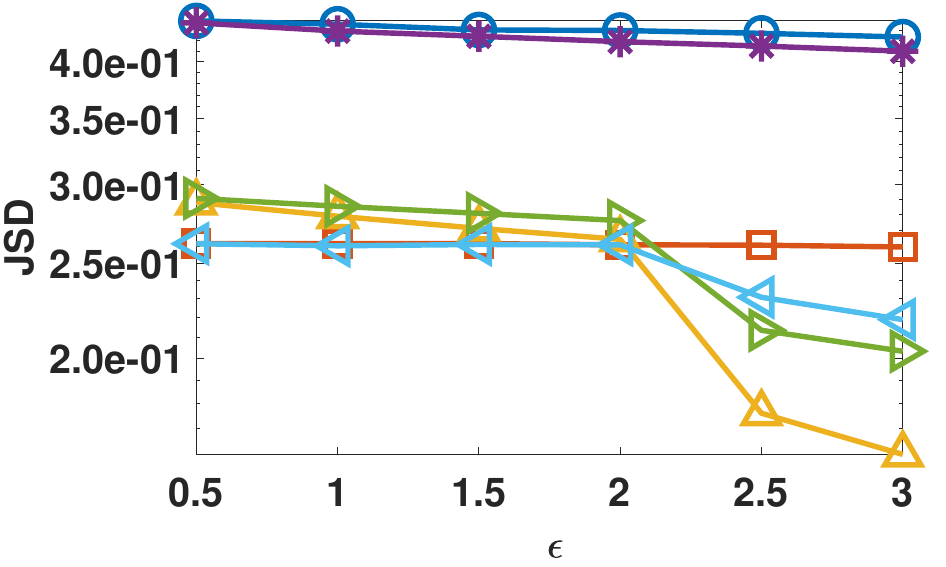}
			%\caption{fig2}
		\end{minipage}
	}%
	\subfigure[\textbf{Volume},,$w=30,q=40$.]{
		\begin{minipage}[t]{0.24\linewidth}
			\centering
			\includegraphics[width=1\textwidth]{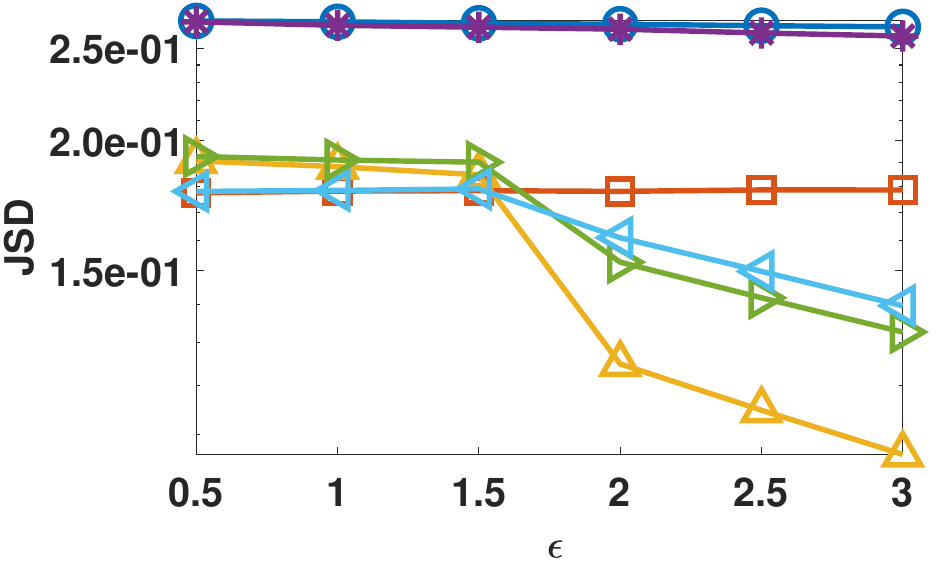}
			%\caption{fig2}
		\end{minipage}
	}%
	\\
	%		\vspace{-0.05in}
	\subfigure[\textbf{Volume},$w=20,q=30$.]{
		\begin{minipage}[t]{0.24\linewidth}
			\centering
			\includegraphics[width=1\textwidth]{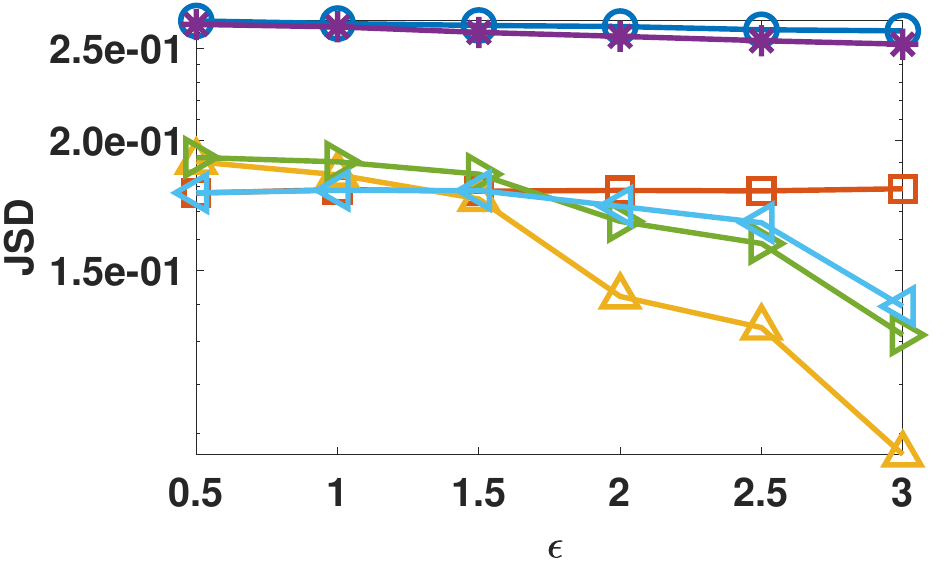}
			%\caption{fig1}
		\end{minipage}%
	}%
	\subfigure[\textbf{C6H6},$w=20,q=30$.]{
		\begin{minipage}[t]{0.24\linewidth}
			\centering
			\includegraphics[width=1\textwidth]{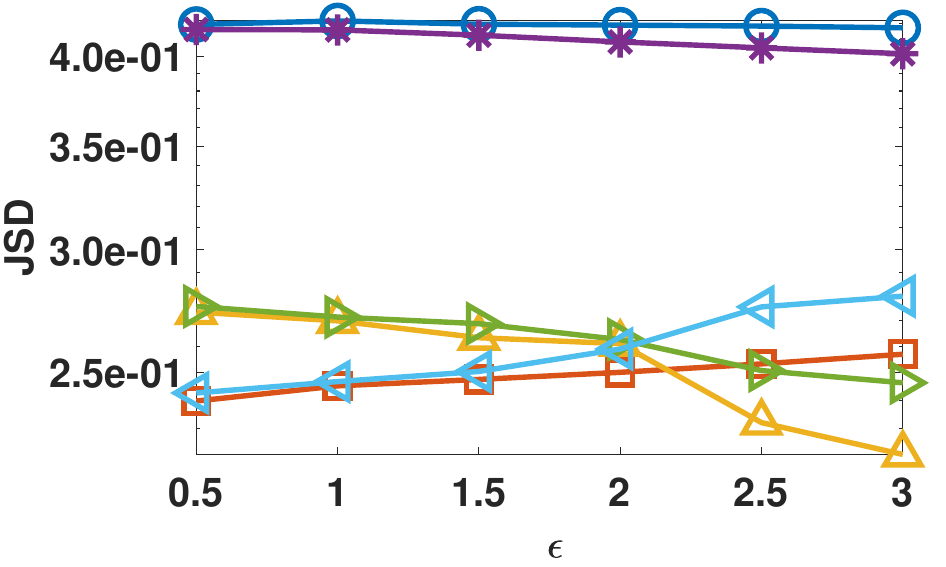}
			%\caption{fig2}
		\end{minipage}%
	}%
	\subfigure[\textbf{Power},$w=20,q=30$.]{
		\begin{minipage}[t]{0.24\linewidth}
			\centering
			\includegraphics[width=1\textwidth]{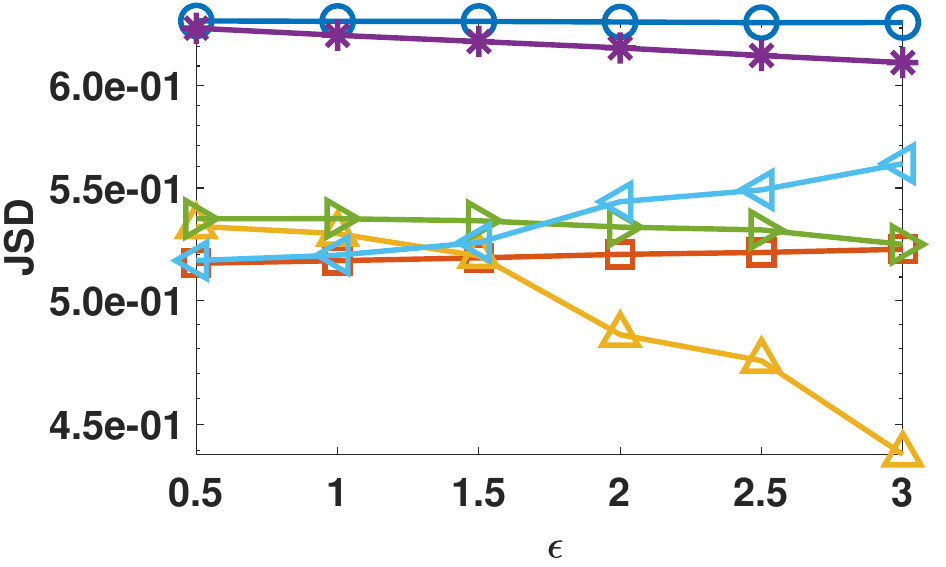}
			%\caption{fig2}
		\end{minipage}
	}%
	\subfigure[\textbf{Taxi},$w=20,q=30$.]{
		\begin{minipage}[t]{0.24\linewidth}
			\centering
			\includegraphics[width=1\textwidth]{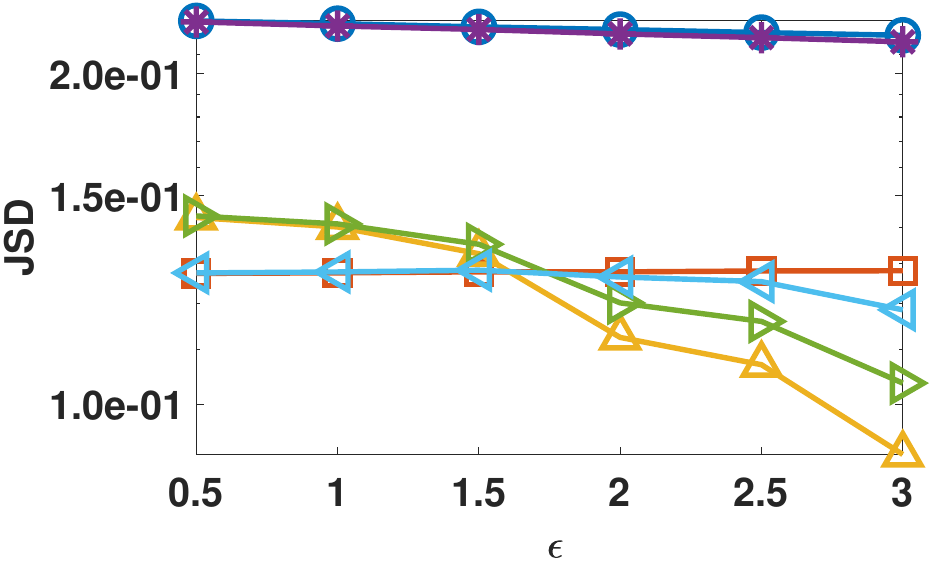}
			%\caption{fig2}
		\end{minipage}
	}%
	\vspace{-0.05in}
	\caption{Cosine Distance w.r.t. $\epsilon$ for sampling-based algorithms vs non-sampling algorithms}
	\label{samplefrequency}
	\vspace{-0.1in}
\end{figure*}

\color{black}
\subsection{Discussion on Generalizability}

\subsubsection{The results for crowd-level statistics}  
We extend our analysis from individual-level statistics to crowd-level statistics, as shown in Figure \ref{wsdis_1}. The evaluation metric used is the Wasserstein distance between two distributions: the estimated distribution and the true distribution of means calculated from population subsequences. Figures \ref{wsdis_1} (a)-(d) present the results without sampling techniques. Among the methods, BA-SW performs the worst, followed by SW-direct, whereas CAPP outperforms APP and IPP. With sampling techniques (Figures \ref{wsdis_1} (e)-(h)), CAPP-S achieves the best performance, followed by APP-S, highlighting the synergy between our proposed methods and sampling. These results indicate that our methods not only provide accurate user-level time series statistics but also excel in estimating crowd-level statistics.

\begin{figure*}[t]
	\hspace{0.25in}
	%	\vspace{-0.06in}
	{
		\begin{minipage}{14cm}
			\centering
			\includegraphics[scale=0.7]{legend_nosample.pdf}
		\end{minipage}
	}
	\\
	%	\vspace{-0.12in}
	\centering
	\subfigure[\textbf{Taxi}, $w=q=10$.]{
		\begin{minipage}[t]{0.24\linewidth}
			\centering
			\includegraphics[width=1\textwidth]{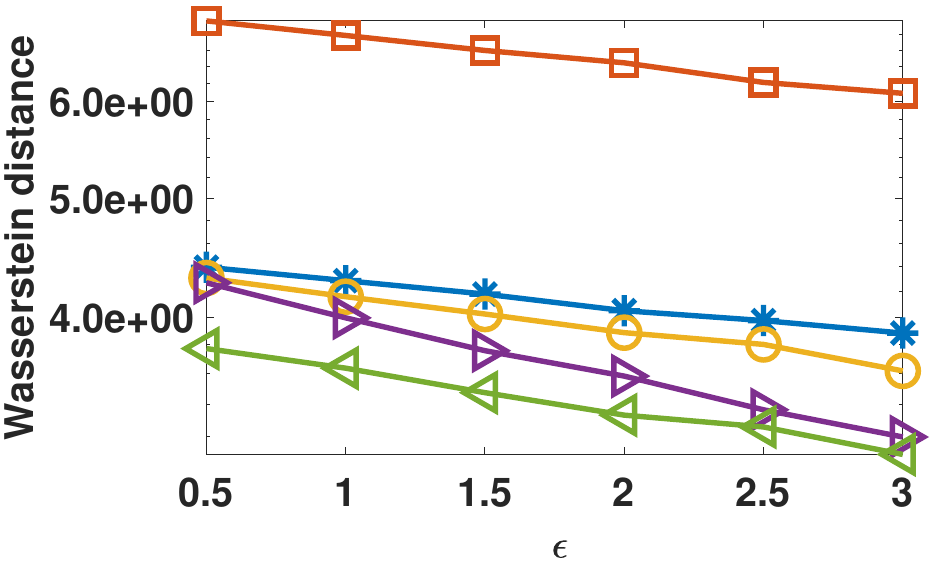}
			%\caption{fig1}
		\end{minipage}%
	}%
	\subfigure[\textbf{Taxi}, $w=q=30$.]{
		\begin{minipage}[t]{0.24\linewidth}
			\centering
			\includegraphics[width=1\textwidth]{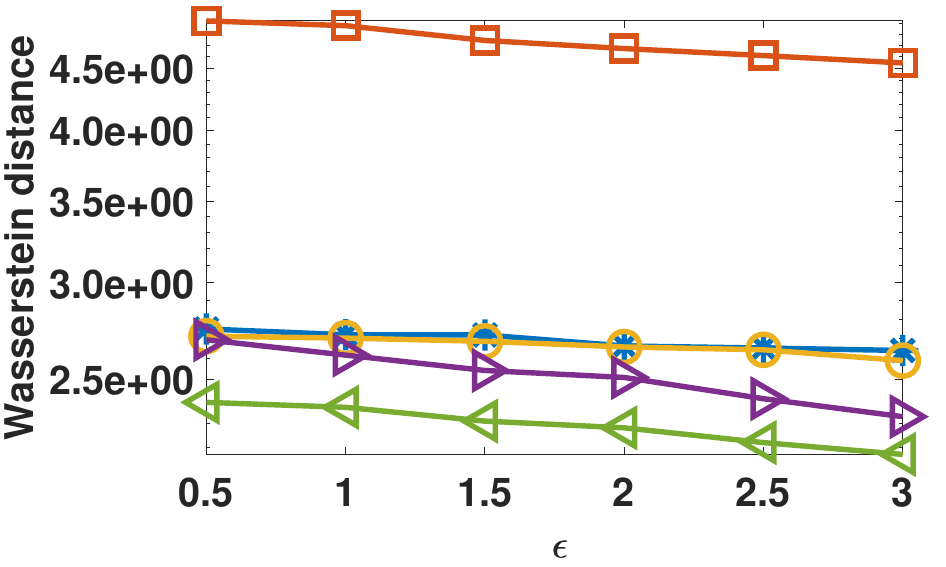}
			%\caption{fig2}
		\end{minipage}%
	}%
	\subfigure[\textbf{Power}, $w=q=10$.]{
		\begin{minipage}[t]{0.24\linewidth}
			\centering
			\includegraphics[width=1\textwidth]{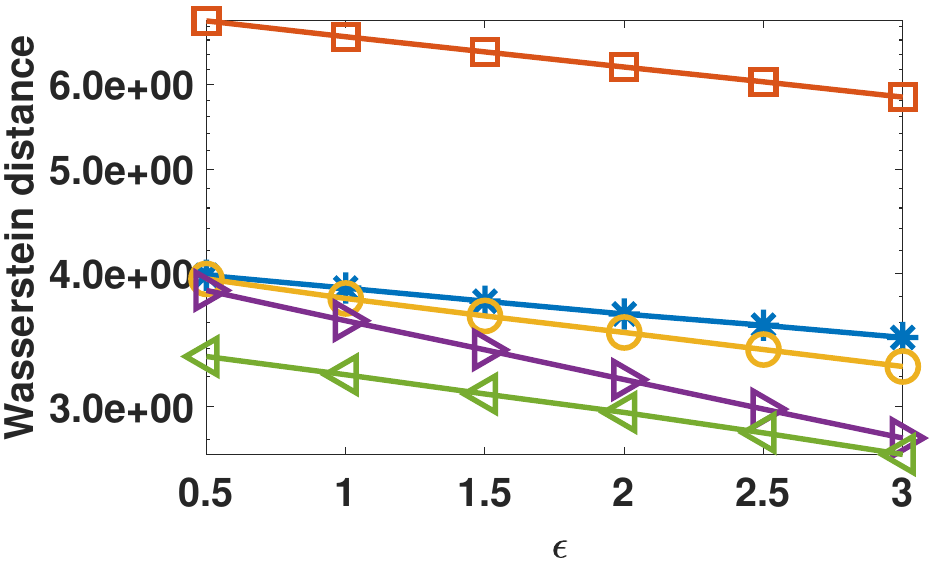}
			%\caption{fig2}
		\end{minipage}
	}%
	\subfigure[\textbf{Power}, $w=q=30$.]{
		\begin{minipage}[t]{0.24\linewidth}
			\centering
			\includegraphics[width=1\textwidth]{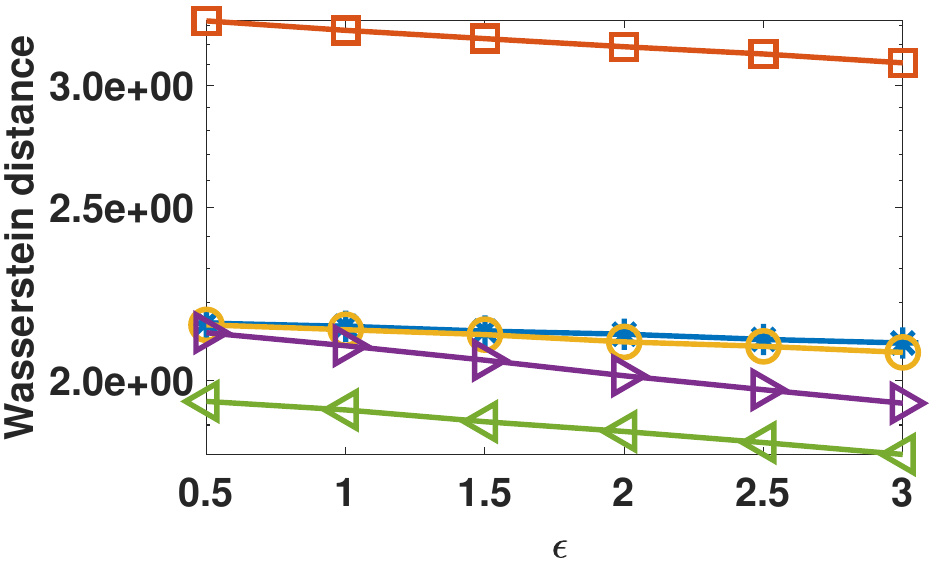}
			%\caption{fig2}
		\end{minipage}
	}%
	%	\vspace{-0.12in}
	\\
	{
		\begin{minipage}{10cm}
			\centering
			\includegraphics[scale=0.7]{legend_sampling.pdf}
		\end{minipage}
	}
	\\
	\centering
	\subfigure[\textbf{Taxi}, $w=20,q=10$.]{
		\begin{minipage}[t]{0.24\linewidth}
			\centering
			\includegraphics[width=1\textwidth]{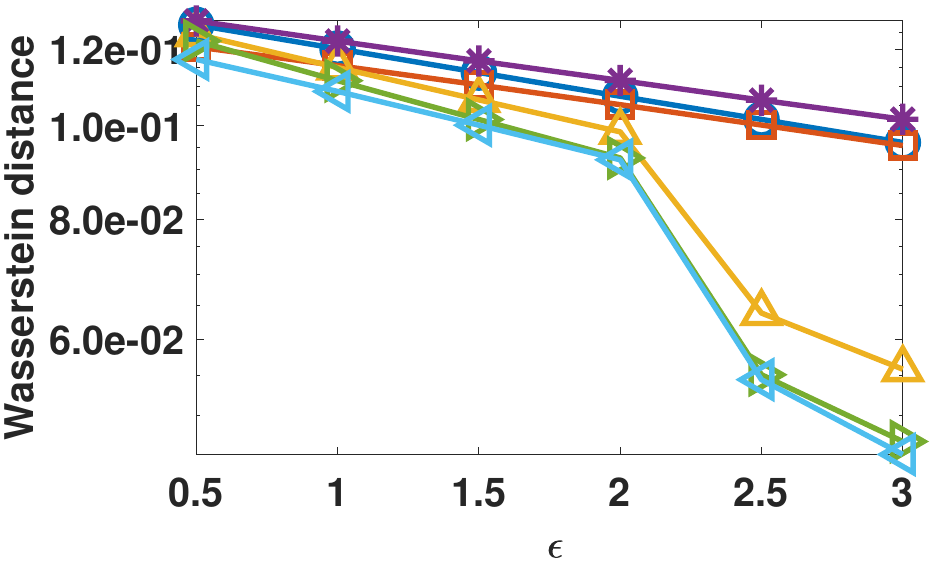}
			%\caption{fig1}
		\end{minipage}%
	}%
	\subfigure[\textbf{Taxi}, $w=20,q=30$.]{
		\begin{minipage}[t]{0.24\linewidth}
			\centering
			\includegraphics[width=1\textwidth]{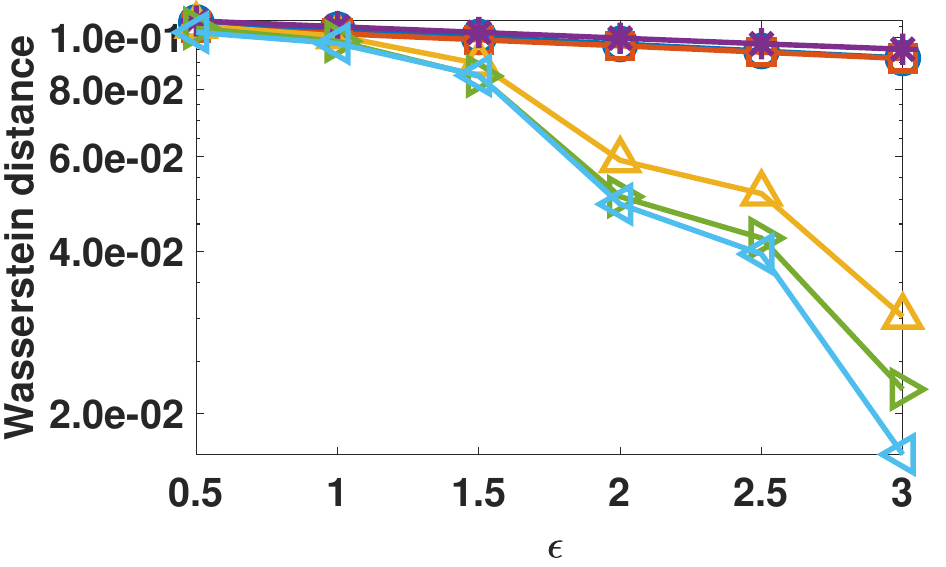}
			%\caption{fig2}
		\end{minipage}%
	}%
	\subfigure[\textbf{Taxi}, $w=30,q=10$.]{
		\begin{minipage}[t]{0.24\linewidth}
			\centering
			\includegraphics[width=1\textwidth]{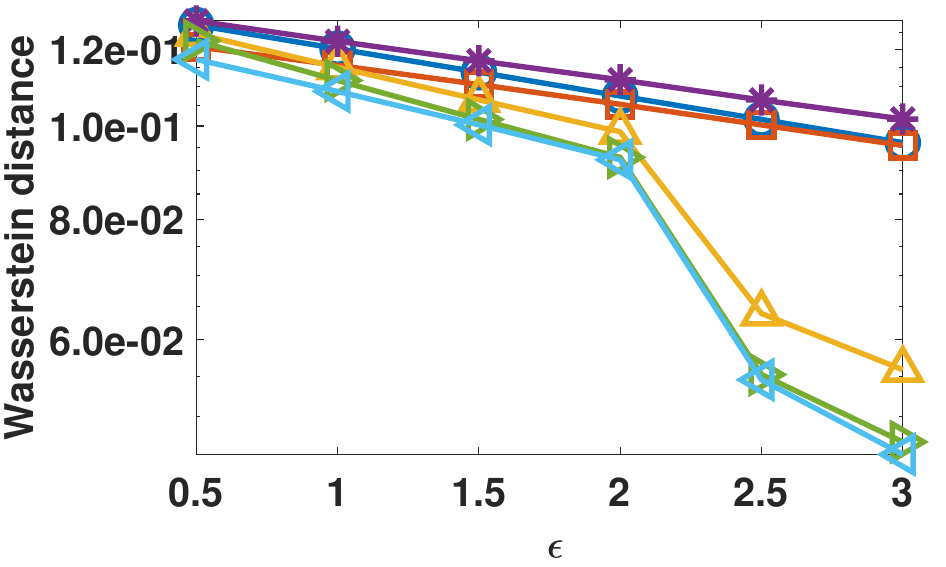}
			%\caption{fig2}
		\end{minipage}
	}%
	\subfigure[\textbf{Taxi}, $w=30,q=40$.]{
		\begin{minipage}[t]{0.24\linewidth}
			\centering
			\includegraphics[width=1\textwidth]{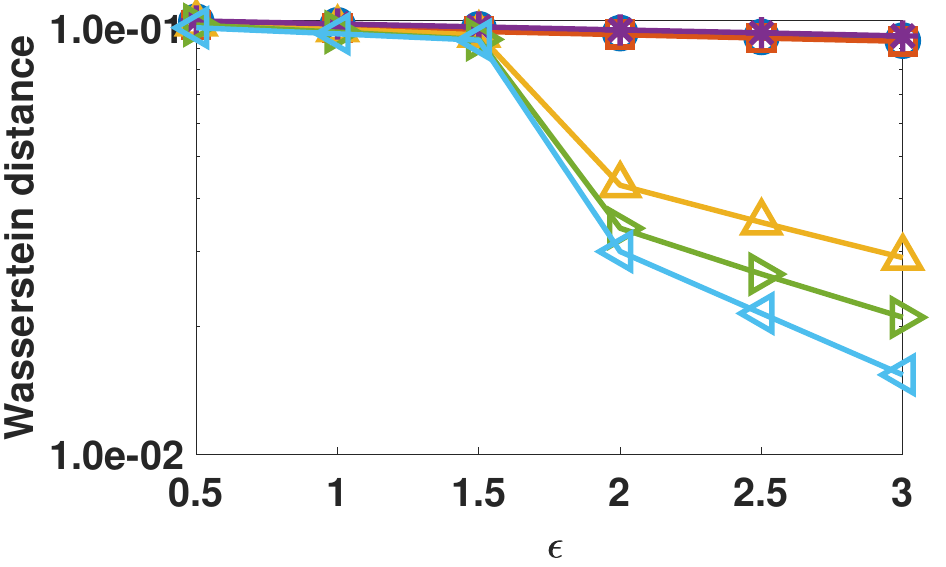}
			%\caption{fig2}
		\end{minipage}
	}%
	\vspace{-0.05in}
	\caption{Wasserstein distance comparison of user mean distributions w.r.t. $\epsilon$}
	\label{wsdis_1}
	\vspace{-0.1in}
\end{figure*}

\subsubsection{Generalizability to different mechanisms}  
We replace the SW mechanism with alternative numerical mechanisms (Laplace, SR, and PM), and the results are presented in Figure \ref{othermechanism}. The results demonstrate that our APP scheme consistently enhances performance across all mechanisms in both mean estimation and stream publication tasks. Notably, the SW mechanism demonstrates superior performance compared to all other mechanisms. This advantage arises from its bounded perturbation range of ($-\frac{1}{2}, \frac{1}{2}$), which preserves a substantial amount of information during the clipping process, irrespective of the $\epsilon$ value. In contrast, other mechanisms either sacrifice excessive information due to their broad perturbation domains or inherently discard significant data characteristics during their perturbation processes, as thoroughly analyzed in Section IV.C.
\begin{figure*}[t]
	\hspace{-0.9in}
	%	\vspace{-0.06in}
	{
		\begin{minipage}{10cm}
			\centering
			\includegraphics[scale=0.75]{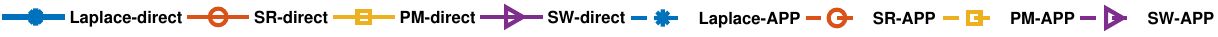}
		\end{minipage}
	}
	\\
	%	\vspace{-0.12in}
	\centering
	\subfigure[\textbf{C6H6}, MSE.]{
		\begin{minipage}[t]{0.24\linewidth}
			\centering
			\includegraphics[width=1\textwidth]{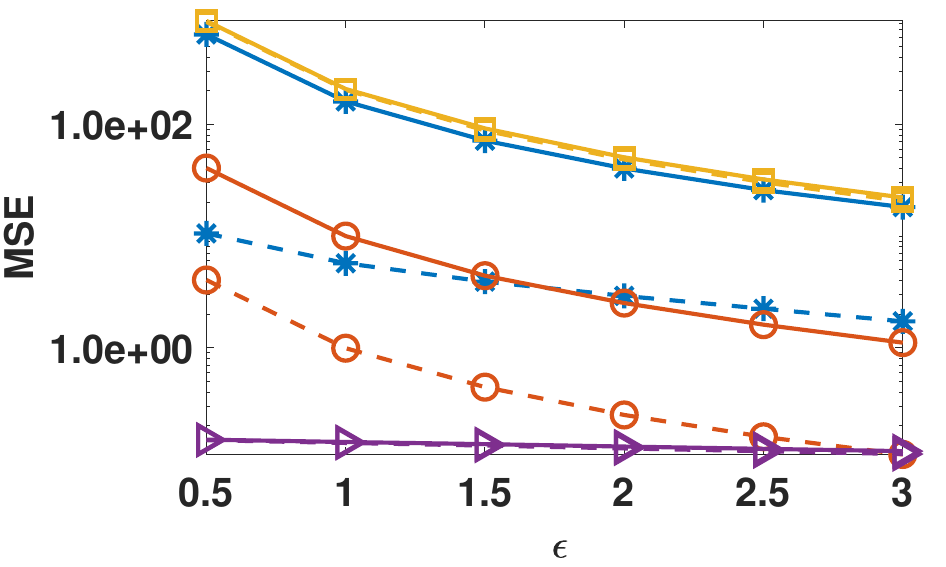}
			%\caption{fig1}
		\end{minipage}%
	}%
	\subfigure[\textbf{Volume}, MSE.]{
		\begin{minipage}[t]{0.24\linewidth}
			\centering
			\includegraphics[width=1\textwidth]{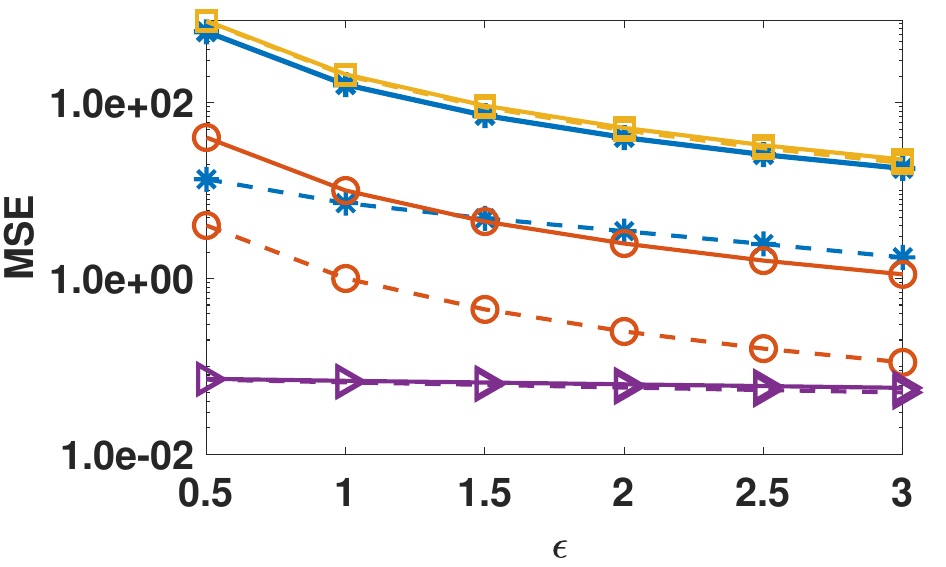}
			%\caption{fig2}
		\end{minipage}%
	}%
	\subfigure[\textbf{C6H6}, JSD.]{
		\begin{minipage}[t]{0.24\linewidth}
			\centering
			\includegraphics[width=1\textwidth]{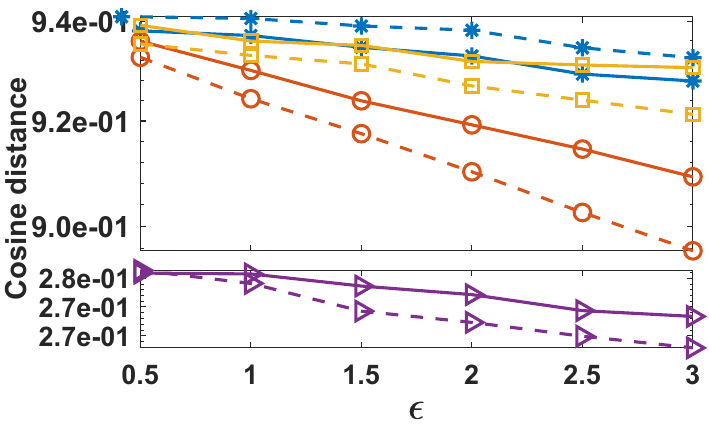}
			%\caption{fig2}
		\end{minipage}
	}%
	\subfigure[\textbf{Volume}, JSD.]{
		\begin{minipage}[t]{0.24\linewidth}
			\centering
			\includegraphics[width=1\textwidth]{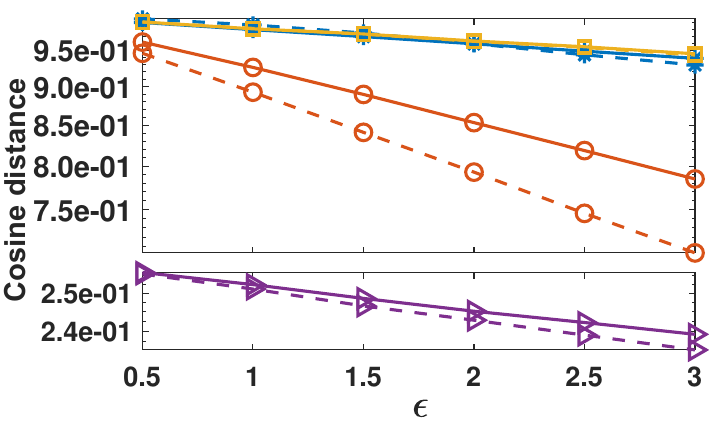}
			%\caption{fig2}
		\end{minipage}
	}%
	\\
	\caption{Performance comparison of different LDP mechanisms with APP (Laplace, SR, PM, and SW)}
	\label{othermechanism}
\end{figure*}

\subsubsection{The results for high-dimensional time series}  
We extend our method to the analysis of high-dimensional time series. To evaluate its performance, we generate two datasets with $d=5$ and $d=12$ dimensions, where each dimension follows a sinusoidal function with varying frequency parameters. We process each dimension of the time series independently, applying privacy-preserving mechanisms separately to each one. To address the privacy budget constraints in this high-dimensional setting, we implement two strategies: Budget-Split (BS) and Sample-Split (SS). Figure \ref{high_dim} indicates that, while SW-SS and SW-BS perform the worst within their respective strategies, our APP and CAPP schemes significantly enhance the performance of both BS and SS approaches. However, the BS strategies (SW-BS, APP-BS, and CAPP-BS) outperform the SS strategies (SW-SS, APP-SS, and CAPP-SS), primarily because the SS approach suffers from reduced effectiveness caused by the limited number of data points per window introduced by sampling. This negative impact of sampling outweighs the disadvantages associated with splitting the privacy budget.

\begin{figure*}[t]
	\hspace{0.25in}
	%	\vspace{-0.06in}
	{
		\begin{minipage}{14cm}
			\centering
			\includegraphics[scale=0.7]{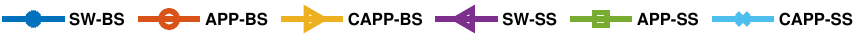}
		\end{minipage}
	}
	\\
	%	\vspace{-0.12in}
	\centering
	\subfigure[\textbf{Sin-data}, MSE, $d=5$.]{
		\begin{minipage}[t]{0.24\linewidth}
			\centering
			\includegraphics[width=1\textwidth]{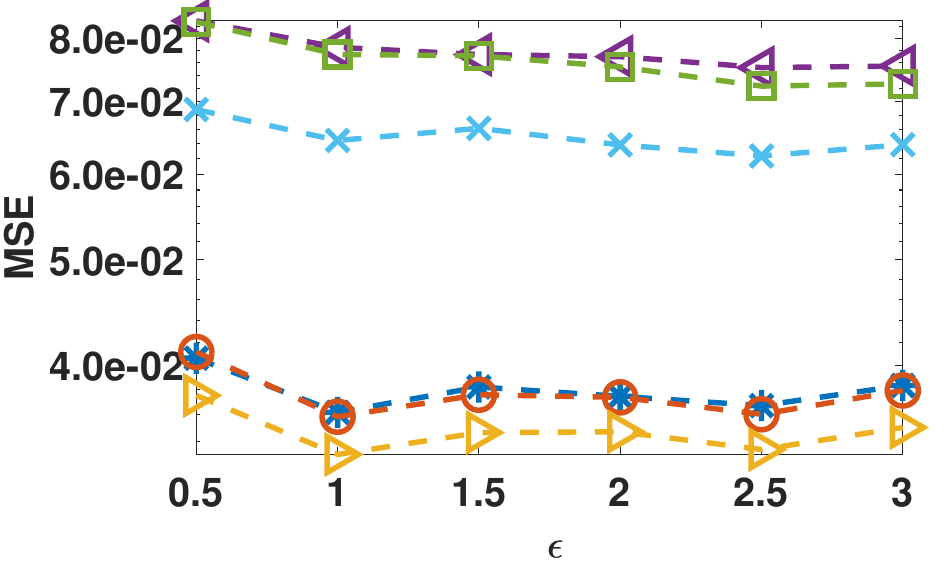}
			%\caption{fig1}
		\end{minipage}%
	}%
	\subfigure[\textbf{Sin-data}, MSE, $d=10$.]{
		\begin{minipage}[t]{0.24\linewidth}
			\centering
			\includegraphics[width=1\textwidth]{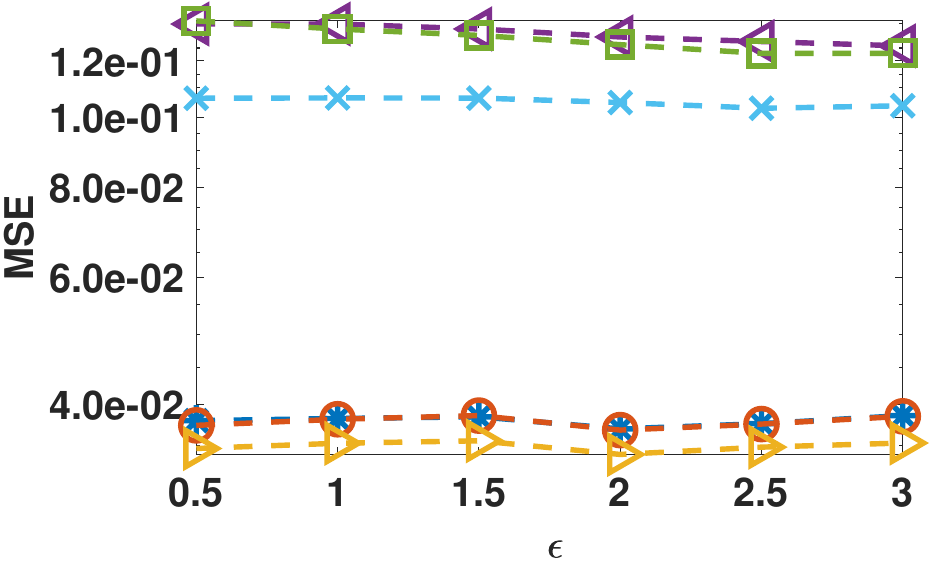}
			%\caption{fig2}
		\end{minipage}%
	}%
	\subfigure[\textbf{Sin-data}, JSD, $d=5$.]{
		\begin{minipage}[t]{0.24\linewidth}
			\centering
			\includegraphics[width=1\textwidth]{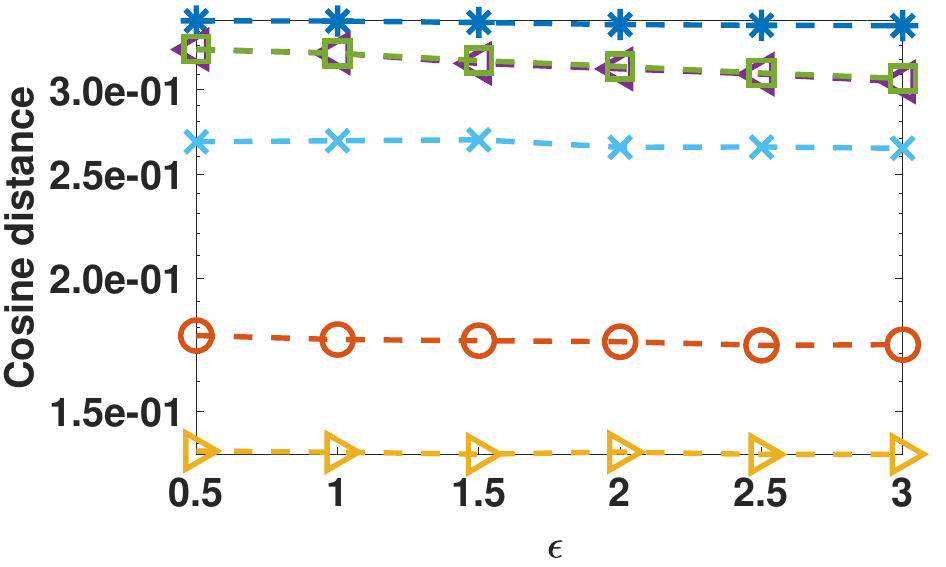}
			%\caption{fig2}
		\end{minipage}
	}%
	\subfigure[\textbf{Sin-data}, JSD, $d=10$.]{
		\begin{minipage}[t]{0.24\linewidth}
			\centering
			\includegraphics[width=1\textwidth]{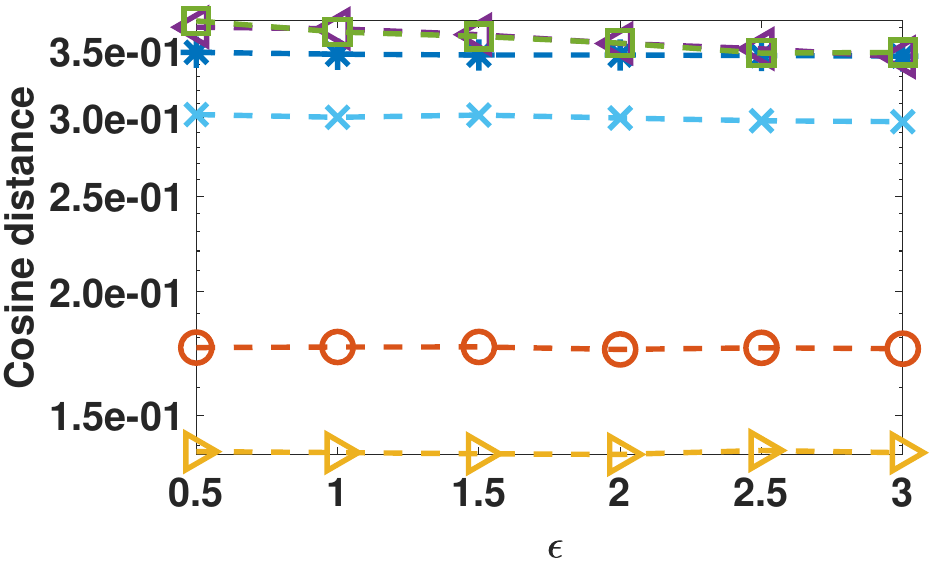}
			%\caption{fig2}
		\end{minipage}
	}%
	\\
	\caption{Performance comparison of budget-split and sample-split in high-dimensional data}
	\label{high_dim}
\end{figure*}

\subsubsection{Sensitivity analysis on $[l, u]$}  
We represent $l$ and $u$ using $\delta$ (where $l = 0 - \delta$ and $u = 1 + \delta$), and evaluate performance for different $\delta$ values across four datasets: \textbf{Constant} (a time series with a constant value of $x = 0.1$), \textbf{Pulse} (zeros with a value of 1 inserted every five points), \textbf{Sinusoidal} (generated according to a sin(x) distribution), and \textbf{C6H6}. The experimental results, as shown in Figure \ref{l_r}, reveal consistent trends across all datasets: MSE decreases as $\epsilon$ increases. For a fixed $\epsilon$, MSE follows a U-shaped curve as $\delta$ varies from $-1$ to $0.5$, with the optimal $\delta$ (corresponding to the minimum MSE) differing across $\epsilon$ values. Generally, smaller $\epsilon$ values are associated with larger optimal $\delta$ values.  

The MSE remains relatively stable in the vicinity of these optimal $\delta$ values. The recommended $\delta$ values (marked with stars), derived from Equation \ref{Teped}, are close to these optimal values and fall within regions of stable MSE. We recommend setting $\delta$ within the range $-0.25 \leq \delta \leq 0.25$, regardless of the original data distribution. For larger $\epsilon$ values, smaller $\delta$ values are preferable, whereas for smaller $\epsilon$ values, larger $\delta$ values are recommended.

\begin{figure*}[t]
	\hspace{0.25in}
	%	\vspace{-0.06in}
	{
		\begin{minipage}{14cm}
			\centering
			\includegraphics[scale=0.7]{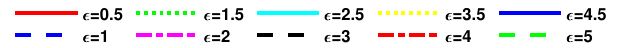}
		\end{minipage}
	}
	\\
	\vspace{-0.12in}
	\centering
	\subfigure[\textbf{Constant}.]{
		\begin{minipage}[t]{0.24\linewidth}
			\centering
			\includegraphics[width=1\textwidth]{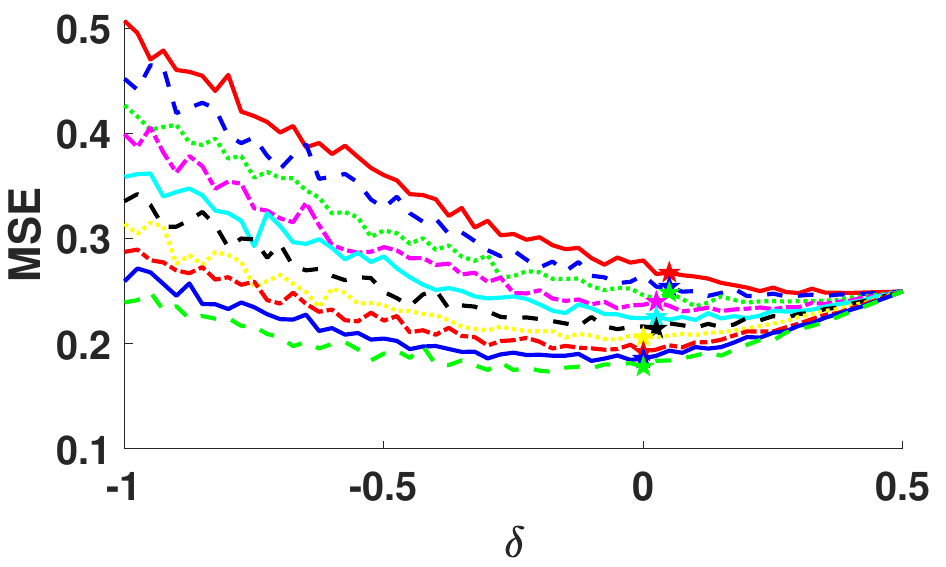}
			%\caption{fig1}
		\end{minipage}%
	}%
	\subfigure[\textbf{Pulse}, MSE, $d=10$.]{
		\begin{minipage}[t]{0.24\linewidth}
			\centering
			\includegraphics[width=1\textwidth]{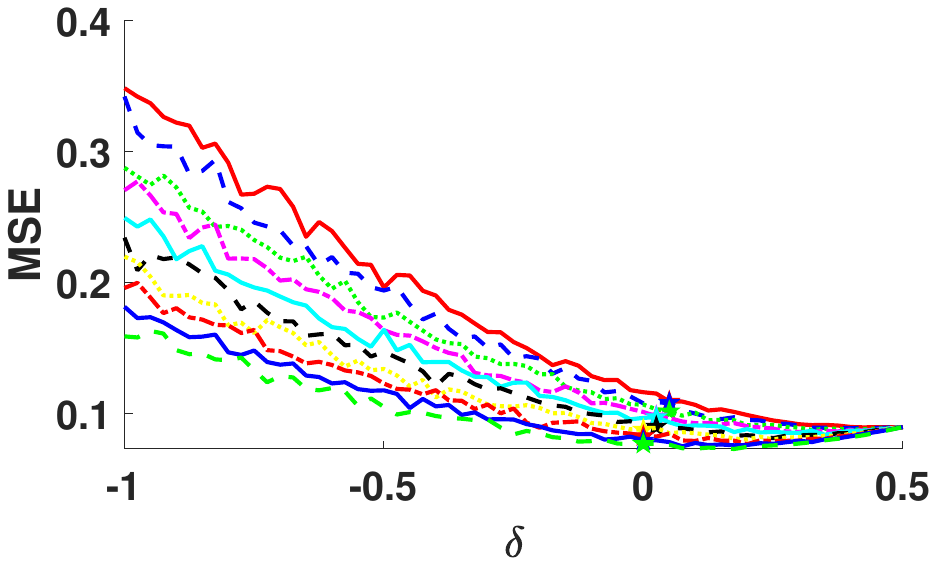}
			%\caption{fig2}
		\end{minipage}%
	}%
	\subfigure[\textbf{Sinusoidal}.]{
		\begin{minipage}[t]{0.24\linewidth}
			\centering
			\includegraphics[width=1\textwidth]{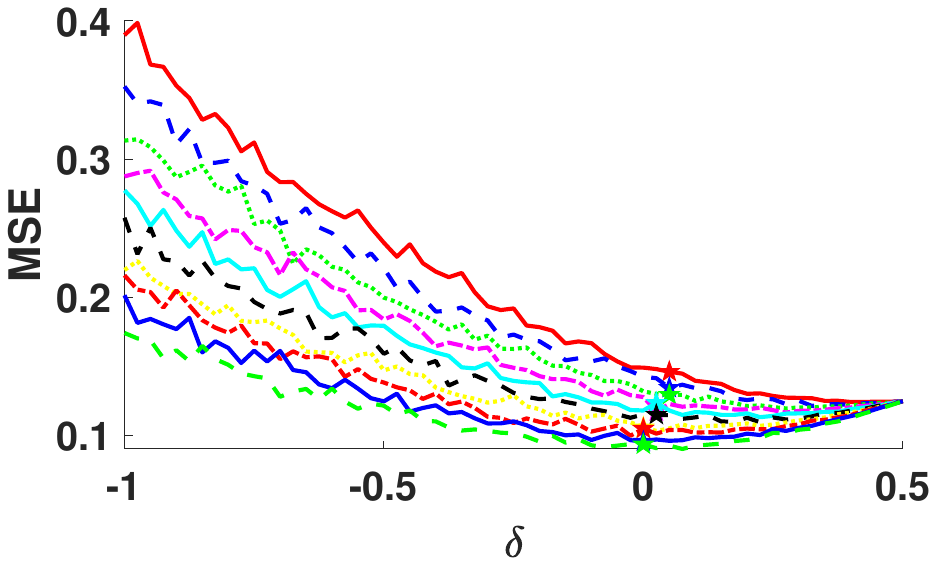}
			%\caption{fig2}
		\end{minipage}
	}%
	\subfigure[\textbf{C6H6}.]{
		\begin{minipage}[t]{0.24\linewidth}
			\centering
			\includegraphics[width=1\textwidth]{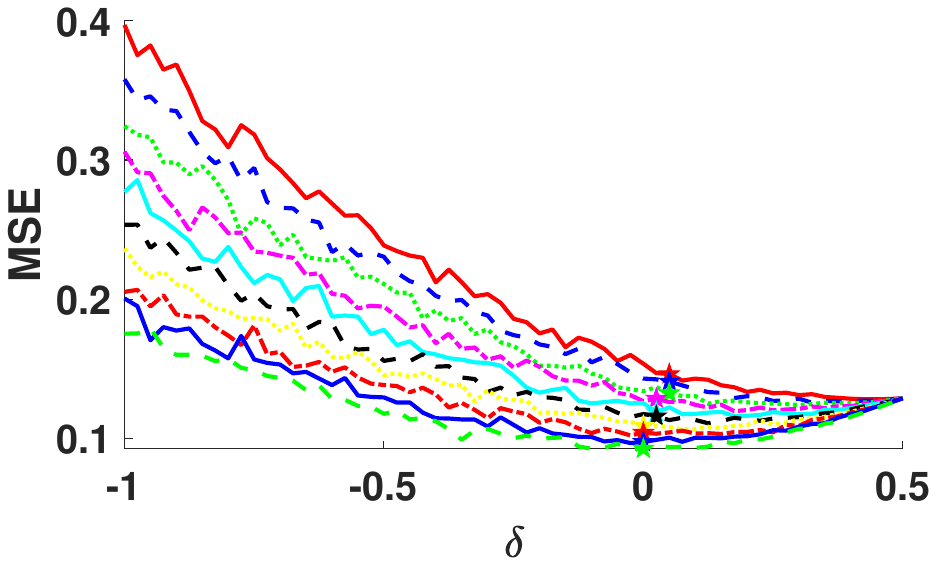}
			%\caption{fig2}
		\end{minipage}
	}%
	\\
	\vspace{-0.1in}
	\caption{Sensitivity analysis of $\delta$ on MSE across datasets, $w=q=10$}
	\label{l_r}
\end{figure*}
%\vspace{-0.2in}
\color{black}
\section{Related Work}
\label{relatedwork}

Local Differential Privacy (LDP)~\cite{chen2016private, duchi2013local, kasiviswanathan2011can} is an extension of Differential Privacy (DP)\cite{dwork2008differential, dwork2006calibrating, mcsherry2007mechanism} tailored for individual users in distributed systems. LDP enables statistical analysis across data types, supporting mean, frequency, and distribution estimations. Within LDP, two primary privacy protection approaches exist: event-level~\cite{wang2021continuous} and user-level~\cite{bao2021cgm}. However, event-level LDP provides insufficient stream protection, while user-level LDP compromises utility through privacy budget partitioning. The $w$-event LDP~\cite{wang2020towards} bridges this gap, offering a balanced framework for time series data protection.

Several methods address streaming data under user-level LDP, primarily focusing on discrete data collection~\cite{joseph2018local, erlingsson2019amplification,xue2022ddrm}. These approaches track statistical information by monitoring data changes to manage privacy budget consumption effectively. Specifically, \cite{joseph2018local} proposes a method that efficiently controls individual user contributions when the server's previous estimation is accurate, thereby conserving privacy budget. Erlingsson et al. \cite{erlingsson2019amplification} introduce a sanitization and reporting mechanism for continuous frequency estimation, where each user's value can change at most $C$ times with increments/decrements of $\pm 1$, and users randomly report one of their $C$ changes. \cite{xue2022ddrm} presents DDRM, which employs binary trees to dynamically record temporal differences and suppresses privacy budget consumption during periods of data stability, taking advantage of the common occurrence of unchanged values in time series data. However, these methods are designed for discrete time series data and cannot be directly applied to continuous time series, which is the focus of our study.

\color{black}
In this paper, we address two types of tasks. The first type focuses on estimating individual-level statistics. Notable existing works include ToPL \cite{wang2021continuous}, Pattern LDP \cite{wang2020towards}, and PrivShape \cite{mao2024privshape}. ToPL \cite{wang2021continuous} adopts event-level LDP constraints, which necessitate small privacy budgets for each timestamp, ultimately compromising utility in stream data publication. Pattern LDP \cite{wang2020towards} attempts to enhance utility by transmitting only significant changes with increased per-point privacy budgets, but this approach risks privacy leakage through the exposure of change points. PrivShape \cite{mao2024privshape} only records key nodes, which may result in information loss when queried intervals contain few or no key nodes. Our work explores $w$-event level analysis as a balanced compromise between these extremes. The second type focuses on estimating statistical characteristics in crowd-level statistics. As demonstrated in LDP-IDS \cite{ren2022ldp}, this includes analyzing the distribution of population means. This method builds upon the BA \cite{kellaris2014differentially} concept of utilizing general privacy budgets by not uploading points with minimal changes, thus conserving remaining privacy budgets. Additionally, it implements a sampling policy where users can upload their information in discontinuous windows, avoiding privacy budget splitting. Beyond simple time series, trajectory data represents a more complex form of time series, as explored in \cite{kellaris2014differentially}, \cite{zhang2023trajectory}, and \cite{sun2023synthesizing}. These works maintain high utility while satisfying local differential privacy requirements.
\color{black} 
\section{Conclusions}
\label{conclusion}In this paper, we propose novel algorithms to collect and publish stream data under LDP. By effectively utilizing perturbations in two complementary ways, our methods significantly enhance utility compared to existing approaches. We present the IPP, APP, and CAPP algorithms and prove they satisfy $w$-event differential privacy. Moreover, we devise an optimized sampling scheme, further improving the accuracy of subsequence mean statistics. Experiments demonstrate the effectiveness of our techniques. This research opens a new direction for high-utility stream data analysis with privacy guarantees. Future work could extend this approach to handle more data types and apply it to emerging contexts like IoT and edge computing.

\section*{Acknowledgements}
\label{sec::ack}
This work was supported by the National Natural Science Foundation of China (Grant No: 62372122 and 92270123), and the Research Grants Council, Hong Kong SAR, China (Grant No:  15203120, 15225921, 15209922, and 15224124). 
%\begin{acks}
%This work was supported by the [...] Research Fund of [...] (Number [...]). Additional funding was provided by [...] and [...]. We also thank [...] for contributing [...].
%\end{acks}

%\clearpage
\bibliographystyle{plain}
\bibliography{references}
\clearpage

\end{document}